\documentclass[journal,comsoc]{IEEEtran}
\usepackage[T1]{fontenc}
\ifCLASSINFOpdf
\else
\fi
\usepackage{lipsum}
\usepackage{amsmath,amssymb,amsthm,mathrsfs,amsfonts,dsfont}
\usepackage{subfigure}
\usepackage{graphicx,dblfloatfix}
\usepackage{algorithm,algorithmic}
\usepackage{tabularx}
\usepackage{adjustbox}
\usepackage{color}
\usepackage{colortbl}
\usepackage[noadjust]{cite}
\DeclareMathOperator*{\argmax}{arg\,max}

\newcommand\norm[1]{\left\lVert#1\right\rVert}
\usepackage{xcolor}

\newtheorem{proposition}{Proposition}
\newtheorem{corollary}{Corollary}
\newtheorem{remark}{Remark}

\begin{document}
\title{Optimal Power Allocation in Downlink Multicarrier NOMA Systems: Theory and Fast Algorithms}
\author{Sepehr Rezvani, \IEEEmembership{Student~Member,~IEEE}, Eduard A. Jorswieck, \IEEEmembership{Fellow,~IEEE}, Roghayeh Joda, \IEEEmembership{Member,~IEEE}, and
	Halim Yanikomeroglu, \IEEEmembership{Fellow,~IEEE}
	\thanks{S. Rezvani and E. A. Jorswieck are with the Department of Information Theory and Communication Systems, Technische Universität Braunschweig, 38106 Braunschweig, Germany (e-mails: \{rezvani, jorswieck\}@ifn.ing.tu-bs.de).}
	\thanks{R. Joda is with Communication Department, ICT Research Institute, Tehran, Iran (e-mail: r.joda@itrc.ac.ir). She is currently a visiting researcher at the School of Electrical Engineering and Computer Science, University of Ottawa, Ottawa K1N 6N5, Canada (e-mail: rjoda@uottawa.ca).}
	\thanks{H. Yanikomeroglu is with the Department of Systems and Computer Engineering, Carleton University, Ottawa, Canada (e-mail: halim@sce.carleton.ca).}
}


\maketitle

\begin{abstract}
	In this work, we address the problem of finding globally optimal power allocation strategies to maximize the users sum-rate (SR) as well as system energy efficiency (EE) in the downlink of single-cell multicarrier non-orthogonal multiple access (MC-NOMA) systems. Each NOMA cluster includes a set of users in which the well-known superposition coding (SC) combined with successive interference cancellation (SIC) technique is applied among them. By obtaining the closed-form expression of intra-cluster power allocation, we show that MC-NOMA can be equivalently transformed to a virtual orthogonal multiple access (OMA) system, where the effective channel gain of these virtual OMA users is obtained in closed-form. Then, the SR and EE maximization problems are solved by using very fast water-filling and Dinkelbach algorithms, respectively. The equivalent transformation of MC-NOMA to the virtual OMA system brings new theoretical insights, which are discussed throughout the paper. The extensions of our analysis to other scenarios, such as considering users rate fairness, admission control, long-term performance, and a number of future next-generation multiple access (NGMA) schemes enabling recent advanced technologies, e.g., reconfigurable intelligent surfaces are discussed. Extensive numerical results are provided to show the performance gaps between single-carrier NOMA (SC-NOMA), OMA-NOMA, and OMA.
\end{abstract}


\IEEEpeerreviewmaketitle

\section{Introduction}
\allowdisplaybreaks
\subsection{Evolution of NOMA: From Fully SC-SIC to Hybrid-NOMA}
\IEEEPARstart{T}{he} rapidly growing demands for high data rate services along with energy constrained networks necessitates the characterization and analysis of the next-generation multiple access (NGMA) techniques in wireless communication systems. It is proved that the capacity region of degraded single-input single-output (SISO) Gaussian broadcast channels (BCs) can be achieved by performing linear superposition coding (SC) at the transmitter side combined with coherent multiuser detection algorithms, like successive interference cancellation (SIC) at the receivers side \cite{PHDdiss,10.5555/1146355,NIFbook,915665,1246013}.
The SC can be performed in code or power domains \cite{10.5555/3294318}. The SC-SIC technique is also called non-orthogonal multiple access (NOMA) \cite{10.5555/3294318}. Based on the adopted SC technique, NOMA can be divided into two main categories, namely code-domain NOMA, and power-domain NOMA \cite{10.5555/3294318,6692652,8357810}. In our work, we consider power-domain NOMA, and subsequently, the term NOMA is referred to as power-domain NOMA.
In addition to the superior spectral efficiency of NOMA compared to orthogonal multiple access (OMA), i.e., frequency division multiple access (FDMA), and time division multiple access (TDMA) \cite{915665,1246013}, academic and industrial research has demonstrated that NOMA can support massive connectivity, which is important for ensuring that the fifth generation (5G) wireless networks can effectively support Internet of Things (IoT) functionalities \cite{7973146,9311149}.
The concept of NOMA has been considered in the 3rd generation partnership project (3GPP) long-term evolution advanced (LTE-A) standard, where NOMA is referred to as multiuser superposition transmission (MUST) \cite{3gppNOMA2015}. NOMA is also introduced on many existing as well as future wireless systems, because of its high compatibility with other communication technologies \cite{7973146}. 
For example, a significant number of works addressed the integration of NOMA to simultaneous wireless information and power transfer \cite{7973146,8792153}, cognitive radio networks \cite{7973146,8792153}, cooperative communications \cite{8792153,9382272,8456624}, millimeter wave communications \cite{8792153,8456624}, mobile edge computing networks \cite{8792153,8758862,9311149}, and reconfigurable intelligent surfaces (RISs) \cite{9122596,9475160,9365004}. In \cite{9365004}, it is shown that the channel capacity of the multiuser downlink RIS system can be achieved by NOMA with time sharing. To this end, NOMA is a promising candidate solution for the beyond-5G (B5G)/sixth generation (6G) wireless networks \cite{9154358}.

The SIC complexity is cubic in the number of multiplexed users \cite{7263349}.
Another issue is error propagation, which increases with the number of multiplexed users \cite{7263349}. Hence, single-carrier NOMA (SC-NOMA), where the signal of all the users is multiplexed, is still impractical for a large number of users. In this line, NOMA is introduced on multicarrier systems, called multicarrier NOMA (MC-NOMA), where the users are grouped into multiple clusters each operating in an isolated resource block, and SC-SIC is applied among users within each cluster \cite{6692652}.
Note that the space division multiple access (SDMA) can also be introduced on NOMA, where the clusters are isolated by zero-forcing beamforming \cite{7263349}.
MC-NOMA with disjoint clusters is based on SC-SIC and FDMA/TDMA, where each user occupies only one resource block, thus receives a single symbol. In FDMA-NOMA, no user benefits from the well-known multiplexing gain in the fading channels. To this end, NOMA is introduced on orthogonal frequency division multiple access (OFDMA), called OFDMA-NOMA or Hybrid-NOMA \cite{7263349,7676258,10.5555/3294318,8823873,9154358}. Hybrid-NOMA is the general case of MC-NOMA, where each user can occupy more than one subchannel, and SC-SIC is applied to each isolated subchannel. Therefore, all the users can benefit from the multiplexing gain.

\subsection{Related Works and Open Problems}
It is well-known that the dynamic resource allocation is necessary in downlink SC/MC-NOMA to achieve a preferable performance, as well as guaranteed quality of services (QoSs) for mission-critical applications \cite{9154358}. Maximizing users sum-rate (SR) is one of the important objectives of resource allocation optimization, which is widely addressed not only for SC/MC-NOMA, but also for the other multiple access techniques. In downlink SC-NOMA, maximizing users SR leads to the full base station's (BS's) power consumption \cite{9583874}. The energy consumption is becoming a social and economical issue due to the rapid increase of the data traffic and number of mobile devices \cite{Jorswieckfractional}. Hence, minimizing the BSs power consumption while guaranteeing users minimum rate demands is another important objective of resource allocation optimization. To strike a balance between users SR and BS's power consumption, maximizing the well-known fractional system energy efficiency (EE) function, defined as $\frac{\text{Receivers Sum-Rate}}{\text{Transmitter's Total Power Consumption}}$, has attracted lots of attention \cite{Jorswieckfractional,6213038}. The EE is measured in bit/Joule, thus measuring the amount of data transmitted per Joule of the consumed transmitter's energy \cite{Jorswieckfractional}. In the following, we review the related works which addressed resource allocation optimization for maximizing SR/EE in the downlink of single-cell SC/MC-NOMA systems.

\subsubsection{SC-NOMA}
In our previous work \cite{9583874}, we derived the closed-form expression of optimal powers to maximize the SR of $M$-user SC-NOMA system with minimum rate demands under the optimal channel-to-noise ratio (CNR)-based decoding order. The work in \cite{8861078} addresses the problem of simultaneously maximizing users SR and minimizing total power consumption defined as a utility function for SC-NOMA. However, the analysis in \cite{8861078} is affected by a detection constraint for successful SIC which is not necessary, since SISO Gaussian BCs are degraded. Hence, \textit{the closed-form expression of optimal powers to maximize system EE in SC-NOMA is still an open problem}.

\subsubsection{MC-NOMA}
The joint power and subchannel allocation in MC-NOMA is proved to be strongly NP-hard \cite{7587811,8422362,9044770}. In this way, these two problems are decoupled in most of the prior works. For any given set of clusters, the optimal power allocation for SR/EE maximization in MC-NOMA is more challenging compared to SC-NOMA. In MC-NOMA, there exists a competition among multiple clusters to get the cellular power. Actually, the optimal power allocation in MC-NOMA includes two components: 1) \textit{Inter-cluster power allocation}: optimal power allocation among clusters to get the cellular power budget; 2) \textit{Intra-cluster power allocation}: optimal power allocation among multiplexed users to get the clusters power budget. From the optimization perspective, the analysis in \cite{9583874} is also valid for MC-NOMA with any predefined power budget for each cluster, e.g., the considered models in \cite{7557079,8114362}. In this case, the intra-cluster power allocation can be equivalently decoupled into multiple SC-NOMA subproblems. There has been some efforts in finding the optimal joint intra- and inter-cluster power allocation, thus globally optimal power allocation, for MC-NOMA to maximize SR/EE \cite{7982784}. In \cite{7982784}, FDMA-NOMA with $2$ users per cluster is considered. The authors first obtain the closed-form expression of optimal intra-cluster power allocation for each $2$-order cluster. Then, by substituting these closed-forms to the original problems, the optimal inter-cluster power allocation is obtained in efficient manners for various objectives. \textit{In \cite{7982784}, all the analysis is based on allocating more power to each weaker user to guarantee successful SIC, which is not necessary, due to the degradation of SISO Gaussian BCs \cite{NIFbook,8823873}. Another concern is the generalization of the special FDMA-NOMA scheme with $2$-order clusters to Hybrid-NOMA with arbitrary number of multiplexed users.} 

The works on Hybrid-NOMA mainly focus on achieving the maximum multiplexing gain, where each user receives different symbols on the assigned subchannels. It is straightforward to show that Hybrid-NOMA with per-symbol/subchannel minimum rate constraints can be equivalently transformed to FDMA-NOMA, since a user on different assigned subchannels can be viewed as independent users with individual per-subchannel minimum rate demands.
The fractional EE maximization problem for downlink FDMA/Hybrid-NOMA with per-symbol minimum rate demands is addressed by \cite{7523951,9276828,8119791,8448840,9032198}. In this line, the EE maximization problem is solved by using the suboptimal difference-of-convex (DC) approximation method \cite{7523951}, Dinkelbach algorithm with Fmincon optimization software \cite{9276828}, and Dinkelbach algorithm with subgradient method \cite{8119791,8448840,9032198}. Despite the potentials, there are some fundamental questions that are not yet solved in the literature for the SR/EE maximization problems of downlink Hybrid-NOMA with minimum rate constraints as follows:
\begin{enumerate}
	\item \textit{What are the closed-form of optimal powers for the SR/EE maximization problems?}
	\item \textit{Is the equal power allocation strategy a good solution for the SR/EE maximization problems?}
	\item \textit{Is there any users rate fairness guarantee in the SR/EE maximization problems?}
	\item \textit{When the full cellular power consumption is energy efficient?}
	\item \textit{How can we equivalently transform Hybrid-NOMA to a FDMA system?}
\end{enumerate}
The answer of the first question brings new theoretical insights on the impact of minimum rate demands and channel gains on the optimal power coefficients among multiplexed users. Also, by analyzing the heterogeneity of optimal power allocation among multiplexed users/clusters, we can analytically observe which of the equal intra/inter-cluster power allocation strategies are mostly infeasible/near-optimal. The optimality conditions analysis for the SR/EE maximization problem shows us which users get additional rate rather than their individual minimum rate demands, which is important for guaranteeing users rate fairness. If we guarantee that the full power consumption leads to the maximum EE, the EE maximization problem can be reduced to the SR maximization problem, which subsequently decreases the complexity of the solution methods used in \cite{8119791,8448840,9032198}. Finally, transforming a Hybrid-NOMA system with $N$ subchannels each having $K$ users to a FDMA system with $N$ subchannels will reduce the dimension of the SR/EE maximization problems of Hybrid-NOMA. This decreases the complexity of the solution algorithms, e.g., the pure convex solvers used in \cite{8119791,8448840,9032198}. Moreover, Hybrid-NOMA-to-FDMA transformation facilitates the implementation of Hybrid-NOMA, since the optimization algorithms which are already developed for FDMA can be easily adopted to be used for Hybrid-NOMA.

In general, finding the optimal power allocation for SR/EE maximization problem in downlink Hybrid-NOMA with per-user minimum rate demands\footnote{Minimum rate constraint for each user over all the assigned subchannels.} is more challenging, due to the nonconvexity of minimum rate constraints. The works in \cite{7560605,7587811,7812683,7996797,8489876,8737495,8885490,8422362,9044770} address the problem of weighted SR/SR maximization for Hybrid-NOMA without guaranteeing users minimum rate demands. In Hybrid-NOMA with per-user minimum rate constraints, \cite{9264208} proposes a suboptimal power allocation strategy for the EE maximization problem based on the combination of the DC approximation method and Dinkelbach algorithm. Also, a suboptimal penalty function method is proposed in \cite{8807992}. We show that most of our analysis for Hybrid-NOMA with per-symbol minimum rate demands also hold for Hybrid-NOMA with per-user minimum rate demands by using the fundamental relations between these two schemes.

\subsection{Our Contributions}
In this work, we address the problem of finding optimal power allocation for maximizing SR/EE of the downlink single-cell Hybrid-NOMA system including multiple clusters each having an arbitrary number of multiplexed users. We assume that each user has a predefined minimum rate demand on each assigned subchannel \cite{7982784,7523951,9276828,8119791,8448840,9032198}.
Our main contributions are listed as follows:
\begin{itemize}
	\item We prove that for the three main objective functions as total power minimization, SR maximization and EE maximization, in each cluster, 
	only the cluster-head\footnote{The user with the highest decoding order which cancels the signal of all the other multiplexed users.} user deserves additional power while all the other users get power to only maintain their minimal rate demands\footnote{For the total power minimization problem, the cluster-head users also get power to only maintain their minimal rate demands.}.
	\item We obtain the closed-form expression of intra-cluster power allocation within each cluster. We prove that the intra-cluster power allocation is mainly affected by the minimum rate demand of users with lower decoding order leading to high heterogeneity of intra-cluster power allocation. As a result, the equal intra-cluster power allocation will be infeasible in most of the cases. The users exact CNRs merely impact on the intra-cluster power allocation, specifically for high signal-to-interference-plus-noise ratio (SINR) regions.
	\item The feasible power allocation region of Hybrid-NOMA with per-symbol minimum rate demands is defined as the intersection of closed boxes along with affine maximum cellular power constraint. Then, the optimal value for the power minimization problem is obtained in closed form.
	\item For the SR/EE maximization problem, we show that Hybrid-NOMA can be transformed to an equivalent virtual FDMA system. Each cluster acts as a virtual OMA user whose effective CNR is obtained in closed form. Moreover, each virtual OMA user requires a minimum power to satisfy its multiplexed users minimum rate demands, which is obtained in closed form.
	\item A very fast water-filling algorithm is proposed to solve the 
	SR maximization problem in Hybrid-NOMA. The EE maximization problem is solved by using the Dinkelbach algorithm with inner Lagrange dual with subgradient method or barrier algorithm with inner Newton's method. Different from \cite{7523951,9276828,8119791,8448840,9032198}, the closed-form of optimal powers among multiplexed users is applied to further reduce the dimension of the problems, thus reducing the complexity of the iterative algorithms, as well as increase the accuracy of the solutions, which is a win-win strategy.
	\item We propose a necessary and sufficient condition for the equal inter-cluster power allocation strategy to be optimal. We show that in the high SINR regions, the effective CNR of the virtual OMA users merely impacts on the inter-cluster power allocation showing the low heterogeneity of inter-cluster power allocation.
	\item We propose a sufficient condition to verify whether the full cellular power consumption is energy efficient or not. When this condition is fulfilled, we guarantee that at the optimal point of the EE maximization problem, the cellular power constraint is active, so the EE maximization problem can be solved by using our proposed water-filling algorithm.
\end{itemize}
Our optimality conditions analysis show that although \textit{usually} more power will be allocated to the weaker user when all the multiplexed users have the same minimum rate demands, there still exists a critical users rate fairness issue in the SR/EE maximization problem. To this end, we propose a new rate fairness scheme for the downlink of Hybrid-NOMA systems which is a mixture of the well-known proportional fairness among cluster-head users, and weighted minimum rate fairness among non-cluster-head users. The extension of our analysis for the pure Hybrid-NOMA system to other more general/complicated scenarios as well as the integration of Hybrid-NOMA to recent advanced technologies, e.g., reconfigurable intelligent surfaces are discussed in the paper.
Extensive numerical results are provided to evaluate the performance of SC-NOMA, FDMA-NOMA with different maximum number of multiplexed users, and FDMA in terms of outage probability, minimum BS's power consumption, maximum SR and EE. The performance comparison between FDMA-NOMA and SC-NOMA brings new theoretical insights on the suboptimality-level of FDMA-NOMA due to user grouping based on FDMA. In this work, we answer the question \textit{"How much performance gain can be achieved if we increase the order of NOMA clusters, and subsequently, decrease the number of user groups?"} for a wide range of the number of users and their minimum rate demands. The latter knowledge is highly necessary since multiplexing a large number of users would cause high complexity cost at the users' hardware.
The complete source code of the simulations including a user guide is available in \cite{sourcecode}.

\subsection{Paper Organization}
The rest of this paper is organized as follows: The system model is presented in Section \ref{Section sysmodel}. The globally optimal power allocation strategies are presented in Section \ref{Section solution}. The possible extensions of our analysis and future research directions are presented in Section \ref{section extensions}. The numerical results are presented in Section \ref{Section simulation}. Our concluding remarks are presented in Section \ref{Section conclusion}. The abbreviations used in the paper are summarized in Table \ref{table ABBREVIATION}.
\begin{table}[tp]
	\caption{Abbreviations.}
	\begin{center} \label{table ABBREVIATION}
		\scalebox{1}{\begin{tabular}{|c|c|}
			\hline \rowcolor[gray]{0.910}
			\textbf{Abbreviation} & \textbf{Definition} 
			\\ \hline \rowcolor[gray]{0.960}
			3GPP & Third generation partnership project
			\\ \hline \rowcolor[gray]{0.960}
			5G & Fifth generation
			\\ \hline \rowcolor[gray]{0.960}
			6G & Sixth generation
			\\ \hline \rowcolor[gray]{0.960}
			AWGN & Additive white Gaussian noise
			\\ \hline \rowcolor[gray]{0.960}
			B5G & beyond-5G
			\\ \hline \rowcolor[gray]{0.960}
			BC & Broadcast channel
			\\ \hline \rowcolor[gray]{0.960}
			BS & Base station
			\\ \hline \rowcolor[gray]{0.960}
			CNR & Channel-to-noise ratio
			\\ \hline \rowcolor[gray]{0.960}
			CSI & Channel state information
			\\ \hline \rowcolor[gray]{0.960}
			DC & Difference-of-convex
			\\ \hline \rowcolor[gray]{0.960}
			EE & Energy efficiency
			\\ \hline \rowcolor[gray]{0.960}
			FDMA & Frequency division multiple access
			\\ \hline \rowcolor[gray]{0.960}
			FD-NOMA & FDMA-NOMA
			\\ \hline \rowcolor[gray]{0.960}
			IoT & Internet of things
			\\ \hline \rowcolor[gray]{0.960}
			IPM & Interior point method
			\\ \hline \rowcolor[gray]{0.960}
			KKT & Karush–Kuhn–Tucker
			\\ \hline \rowcolor[gray]{0.960}
			LTE-A & Long-term evolution advanced
			\\ \hline \rowcolor[gray]{0.960}
			MC-NOMA & Multicarrier non-orthogonal multiple access
			\\ \hline \rowcolor[gray]{0.960}
			MUST & Multiuser superposition transmission
			\\ \hline \rowcolor[gray]{0.960}
			NGMA & Next-generation multiple access
			\\ \hline \rowcolor[gray]{0.960}
			NOMA & Non-orthogonal multiple access
			\\ \hline \rowcolor[gray]{0.960}
			OFDMA & Orthogonal frequency division multiple access
			\\ \hline \rowcolor[gray]{0.960}
			OMA & Orthogonal multiple access
			\\ \hline \rowcolor[gray]{0.960}
			QoS & Quality of service
			\\ \hline \rowcolor[gray]{0.960}
			RIS & Reconfigurable intelligent surface
			\\ \hline \rowcolor[gray]{0.960}
			SC & Superposition coding
			\\ \hline \rowcolor[gray]{0.960}
			SC-NOMA & Single-carrier non-orthogonal multiple access
			\\ \hline \rowcolor[gray]{0.960}
			SDMA & Space division multiple access
			\\ \hline \rowcolor[gray]{0.960}
			SIC & Successive interference cancellation
			\\ \hline \rowcolor[gray]{0.960}
			SINR & Signal-to-interference-plus-noise ratio
			\\ \hline \rowcolor[gray]{0.960}
			SISO & Single-input single-output
			\\ \hline \rowcolor[gray]{0.960}
			SR & Sum-rate
			\\ \hline \rowcolor[gray]{0.960}
			TDMA & Time division multiple access
			\\ \hline \rowcolor[gray]{0.960}
		\end{tabular}}
	\end{center}
\end{table}

\section{Hybrid-NOMA: OFDMA-Based SC-SIC} \label{Section sysmodel}
\subsection{Network Model and Achievable Rates}\label{subsection netmodel}
Consider the downlink channel of a multiuser system, where a BS serves $K$ users with limited processing capabilities in a unit time slot of a quasi-static channel. The set of users is denoted by $\mathcal{K}=\{1,\dots,K\}$. In this system, the total bandwidth $W$ (Hz) is equally divided into $N$ isolated subchannels with the set $\mathcal{N}=\{1,\dots,N\}$, where the bandwidth of each subchannel is $W_s=W/N$. NOMA is applied to each subchannel with maximum number of multiplexed users $U^{\rm{max}}$. Note that SC-NOMA is infeasible when $U^{\rm{max}}<K$. The set of multiplexed users on subchannel $n$ is denoted by $\mathcal{K}_n=\{k \in \mathcal{K} | \rho^n_{k}=1 \}$, in which $\rho^n_{k}$ is the binary channel allocation indicator, where if user $k$ occupies subchannel $n$, we set $\rho^n_{k}=1$, and otherwise, $\rho^n_{k}=0$.
The set of subchannels occupied by user $k \in \mathcal{K}$, is indicated by $\mathcal{N}_k=\{n \in \mathcal{N} | \rho^n_{k}=1 \}$. In FDMA-NOMA, each user belongs to only one cluster \cite{7557079,8114362,7982784}, thus we have $\mathcal{K}_n \cap \mathcal{K}_m = \emptyset, \forall n,m \in \mathcal{N},~ n \neq m$, or equivalently, $|\mathcal{N}_k|=1,~\forall k \in \mathcal{K}$, where $|.|$ indicates the cardinality of a finite set. In the following, we consider the more general case Hybrid-NOMA with $|\mathcal{N}_k| \geq 1,~\forall k \in \mathcal{K}$. The maximum number of multiplexed users $U^{\rm{max}}$ implies that $|\mathcal{K}_n| \leq U^{\rm{max}},~\forall n \in \mathcal{N}$. The exemplary models of SC-NOMA, FDMA-NOMA, FDMA, Hybrid-NOMA with multiplexing all users in all the subchannels (Hybrid-NOMA with $K$ multiplexed users per cluster), Hybrid-NOMA with $U^{\rm{max}}$ multiplexed users per subchannel, and OFDMA are illustrated in Fig. \ref{Fig Systemmodel}.
\begin{figure*}
	\centering
	\subfigure[SC-NOMA: Capacity-achieving, thus optimal, when $U^{\rm{max}} \geq K$, and infeasible when $U^{\rm{max}}<K$. No multiplexing gain is achieved.]{
		\includegraphics[scale=0.33]{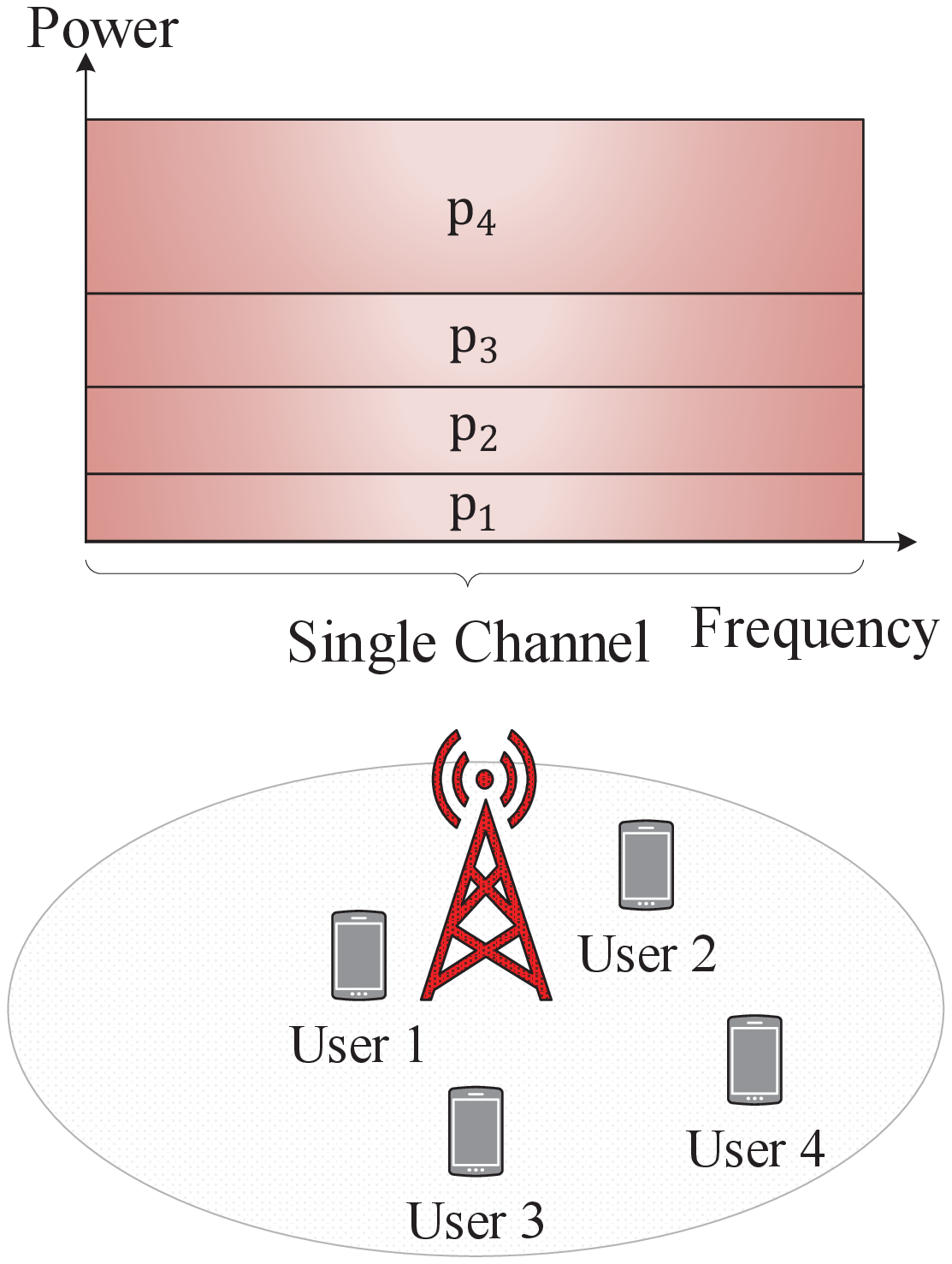}
		\label{Fig_SC-NOMA}
	}\hfill
	\subfigure[FDMA-NOMA: Suboptimal but feasible. SC-SIC is applied to each cluster. No multiplexing gain is achieved.]{
		\includegraphics[scale=0.33]{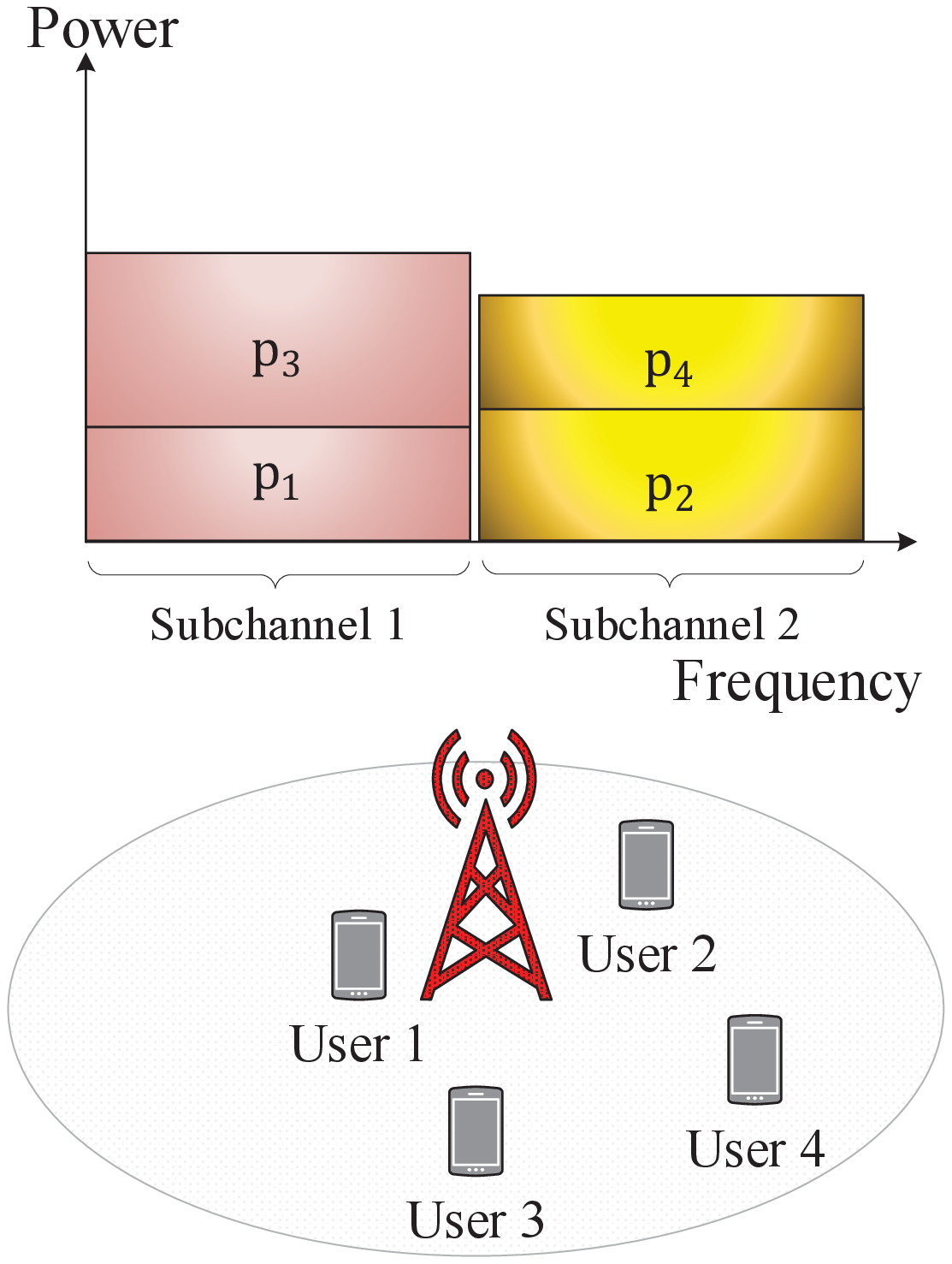}
		\label{Fig_FDMA-NOMA}
	}\hfill
	\subfigure[FDMA: Suboptimal but feasible with no SC-SIC. No multiplexing gain is achieved.]{
		\includegraphics[scale=0.33]{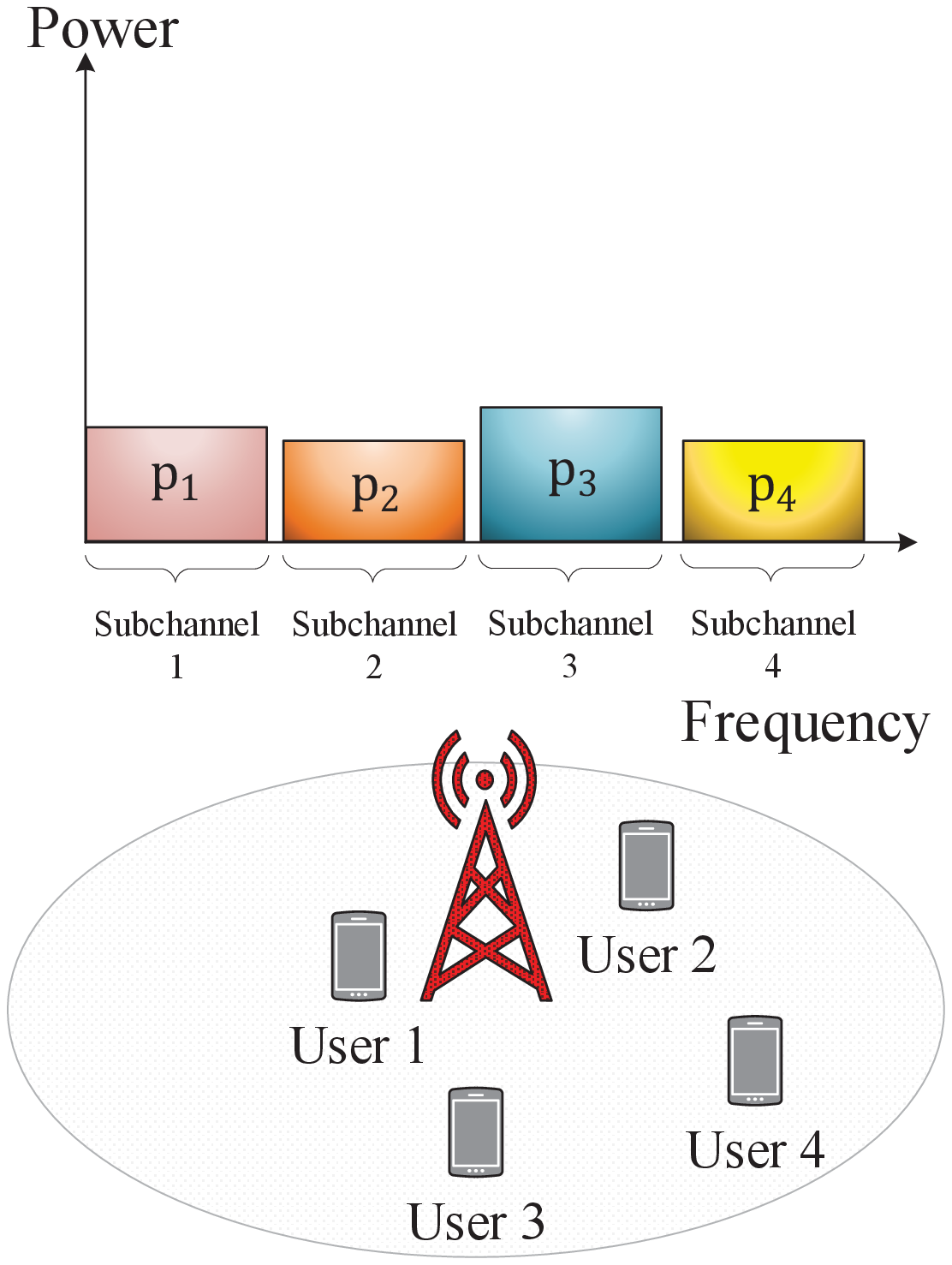}
		\label{Fig_FDMA}
	}\hfill
	\subfigure[Hybrid-NOMA with $K$ multiplexed users: Capacity-achieving, thus optimal, when $U^{\rm{max}} \geq K$, and infeasible when $U^{\rm{max}}<K$. Multiplexing gain can be achieved.]{
		\includegraphics[scale=0.33]{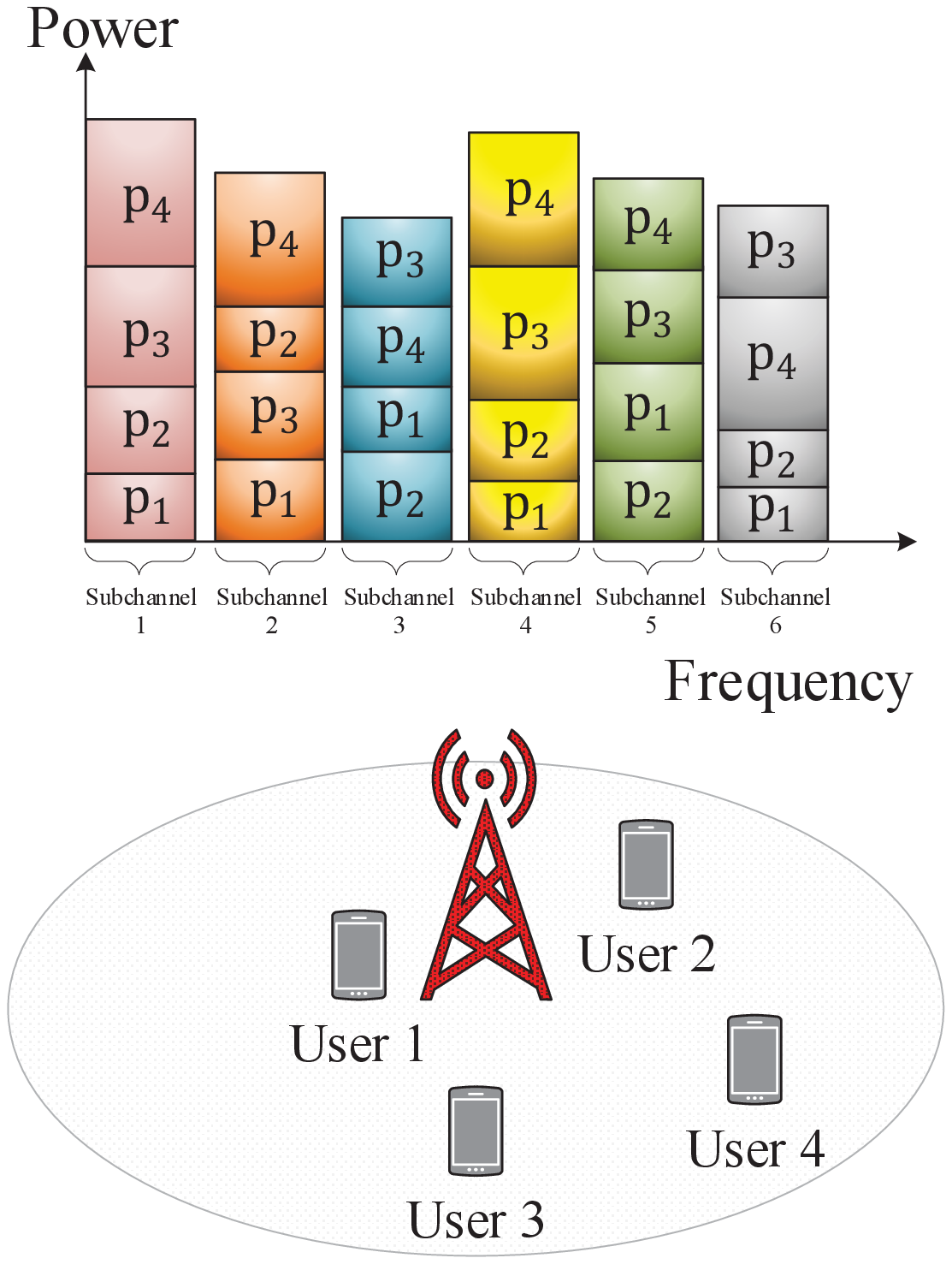}
		\label{Fig_Hybrid-NOMA-all}
	}\hfill
	\subfigure[Hybrid-NOMA with $2$ multiplexed users: Suboptimal but feasible. SC-SIC is applied among users within each cluster. Multiplexing gain can be achieved.]{
		\includegraphics[scale=0.33]{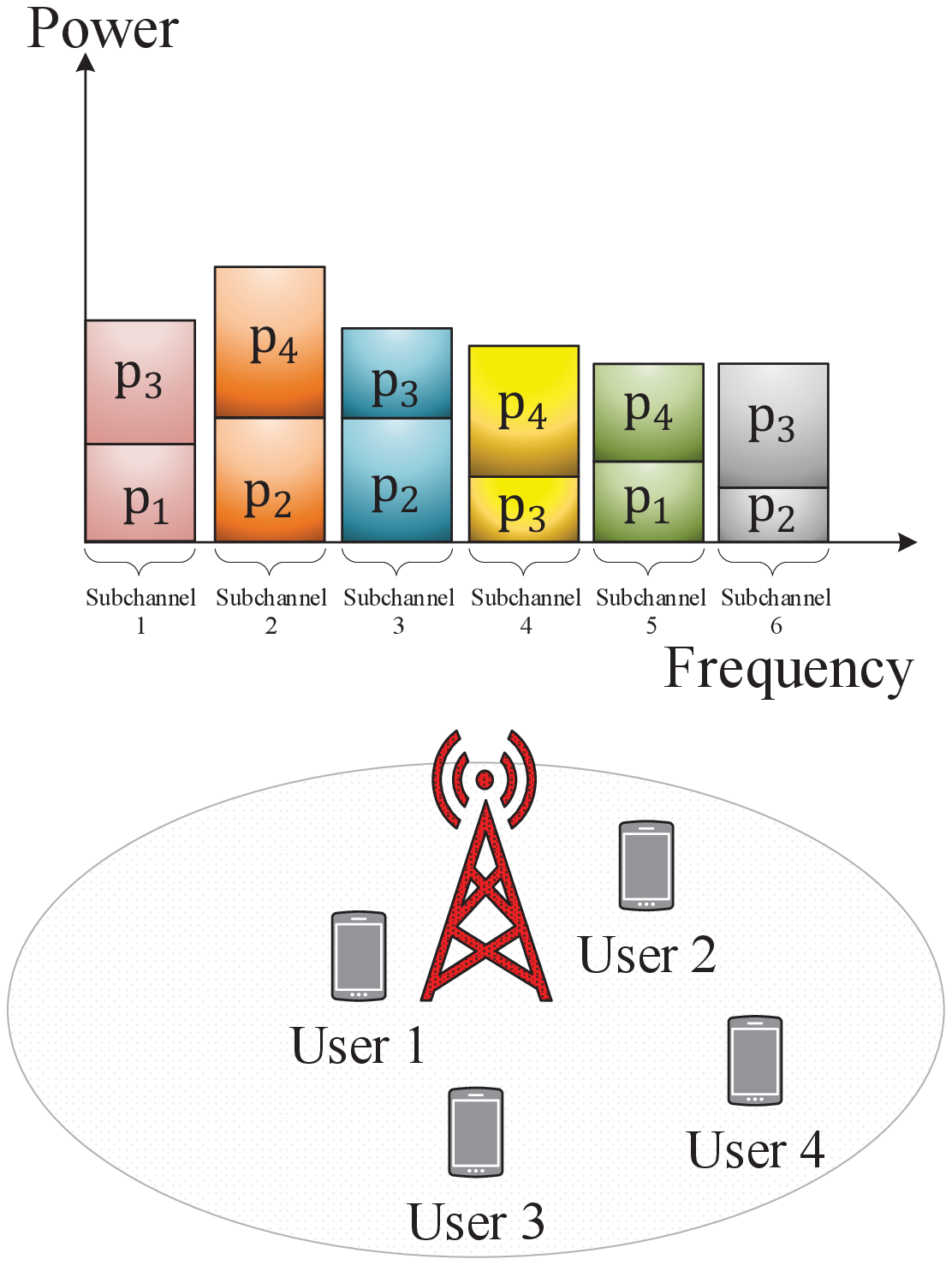}
		\label{Fig_Hybrid-NOMA}
	}\hfill
	\subfigure[OFDMA: Suboptimal but feasible with no SC-SIC. Multiplexing gain can be achieved.]{
		\includegraphics[scale=0.33]{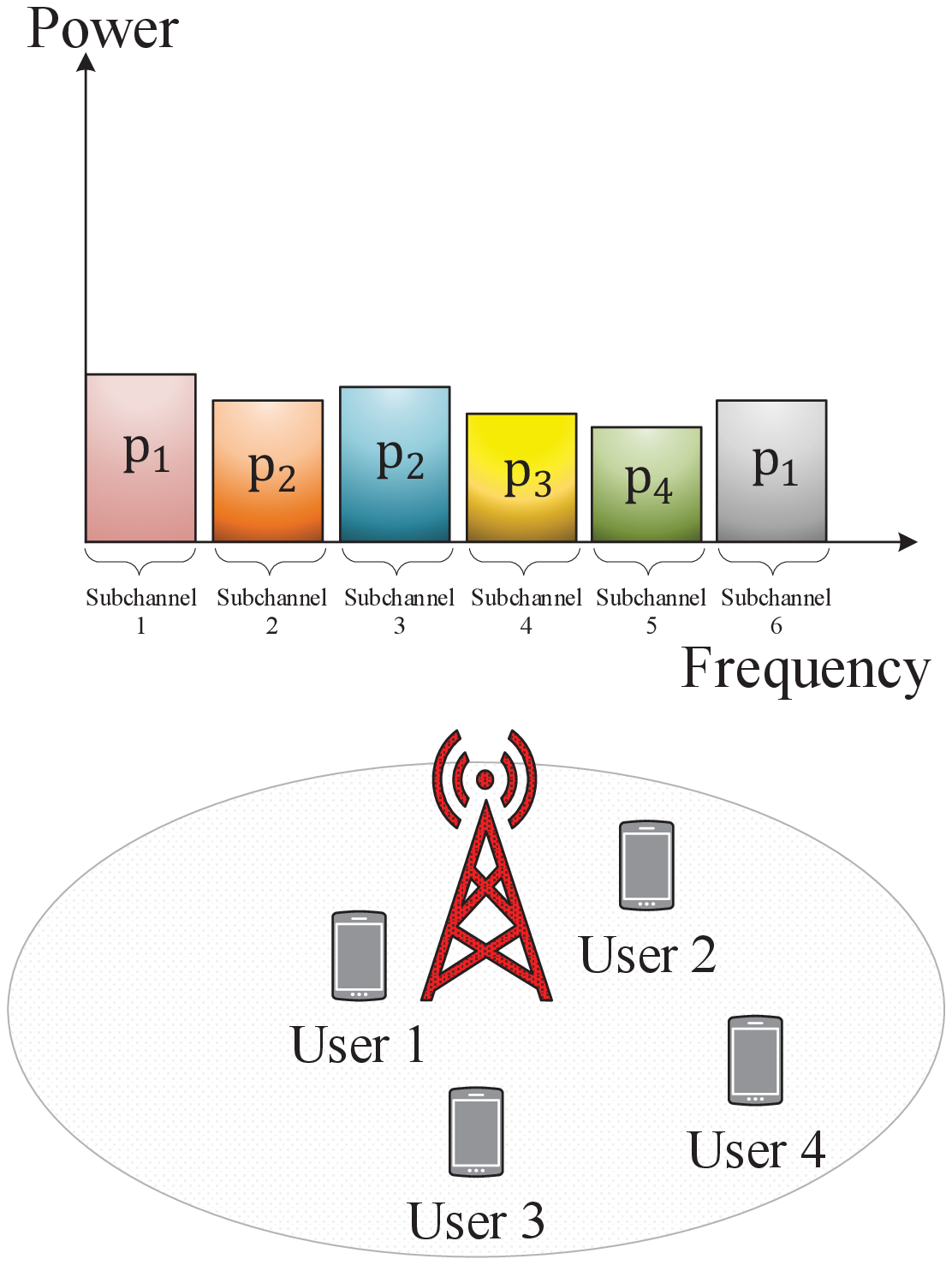}
		\label{Fig_OFDMA}
	}\hfill
	\caption
	{\ref{Fig_SC-NOMA}-\ref{Fig_FDMA}: Exemplary models of SC-NOMA, FDMA-NOMA, and FDMA, respectively, where a single symbol is transmitted to each user. \ref{Fig_Hybrid-NOMA-all}-\ref{Fig_OFDMA}: Exemplary models of Hybrid-NOMA with $K$ multiplexed users, Hybrid-NOMA with $2$ multiplexed users, and OFDMA, respectively, where an independent symbol is transmitted to each user on each assigned subchannel. In these examples, we set $U^{\rm{max}}=2$.}
	\label{Fig Systemmodel}
\end{figure*}

Each subchannel can be modeled as a SISO Gaussian BC. The transmitted signal by the BS on subchannel $n$ is formulated by
$x^n= \sum\limits_{k \in \mathcal{K}_n}
\sqrt{p^n_{k}} s^n_{k},$
where $s^n_{k} \sim \mathcal{CN} (0,1)$ and $p^n_{k} \geq 0$ are the modulated symbol from Gaussian codebooks, and transmit power of user $k \in \mathcal{K}$ on subchannel $n \in \mathcal{N}$, respectively. Obviously, $p^n_{k}=0, \forall n \in \mathcal{N},~k \notin \mathcal{K}_n$. The received signal at user $k$ on subchannel $n$ is 
\begin{equation}\label{received signal}
	y^n_{k}= \underbrace{\sqrt{p^n_{k}} g^n_{k} s^n_{k}}_\text{intended signal} + \underbrace{\sum\limits_{i \in \mathcal{K}_n\setminus\{k\}} \sqrt{p^n_{i}} g^n_{k} s^n_{i}}_\text{co-channel interference} + z^n_{k},
\end{equation}
where $g^n_{k}$ is the (generally complex) channel gain from the BS to user $k$ on subchannel $n$, and $z^n_{k}\sim \mathcal{CN} (0,\sigma^n_{k})$ is the additive white Gaussian noise (AWGN). We assume that the perfect channel state information (CSI) is available at the BS as well as users.

In Hybrid-NOMA, SC-SIC is applied to each multiuser subchannel according to the optimal CNR-based decoding order \cite{PHDdiss,10.5555/1146355,NIFbook,915665,1246013}. Let $h^n_{k}=|g^n_{k}|^2/\sigma^n_{k},~\forall n \in \mathcal{N}, k \in \mathcal{K}_n$. Then, the CNR-based decoding order is indicated by $h^n_{i} > h^n_{j} \Rightarrow i \to j,\forall i,j \in \mathcal{K}_n$, where $i \to j$ represents that user $i$ fully decodes (and then cancels) the signal of user $j$ before decoding its desired signal on subchannel $n$. Moreover, the signal of user $i$ is fully treated as noise at user $j$ on subchannel $n$. In summary, the SIC protocol in each isolated subchannel is the same as the SIC protocol of SC-NOMA. We call the stronger user $i$ as the user with higher decoding order in the user pair $i,j \in \mathcal{K}_n$. In each subchannel $n$, the index of the cluster-head user is denoted by $\Phi_n=\argmax\limits_{k \in \mathcal{K}_n} h^n_{k}$. When $|\mathcal{K}_n|=1$, the single user can be defined as the cluster-head user on subchannel $n$ since it does not experience any interference. The SINR of each user $i \in \mathcal{K}_n$ for decoding the desired signal of user $k$ on subchannel $n$ is
$\gamma^n_{k,i} = \frac{p^n_{k} h^n_{i}}{\sum\limits_{j \in \mathcal{K}_n, \atop h^n_{j} > h^n_{k}} p^n_{j} h^n_{i} + 1}$ \cite{NIFbook}. User $i \in \mathcal{K}_n$ is able to fully decode the signal of user $k$ if and only if $\gamma^n_{k,i} \geq \gamma^n_{k,k}$, where $\gamma^n_{k,k} \equiv \gamma^n_{k}= \frac{p^n_{k} h^n_{k}}{\sum\limits_{j \in \mathcal{K}_n, \atop h^n_{j} > h^n_{k}} p^n_{j} h^n_{k} + 1}$ is the SINR of user $k$ for decoding its own signal $s^n_{k}$. According to the Shannon's capacity formula, the achievable rate (in bps) of user $k \in \mathcal{K}_n$ on subchannel $n \in \mathcal{N}$ after successful SIC is given by \cite{10.5555/1146355,NIFbook,9583874}
$$
R^n_{k} (\boldsymbol{p}^n) = \min\limits_{i \in \mathcal{K}_n \atop h^n_{i} \geq h^n_{k}}\left\{W_s \log_2\left( 1+ \gamma^n_{k,i} (\boldsymbol{p}^n) \right)\right\},
$$
where $\boldsymbol{p}^n=[p^n_{k}]_{1 \times K}$, is the
vector of allocated powers to all the users on subchannel $n$. The matrix of power allocation among all the users and subchannels is denoted by $\boldsymbol{p}=[p^n_{k}]_{N \times K}$. Therefore, $\boldsymbol{p}^n$ is the $n$-th row of matrix $\boldsymbol{p}$.
For the user pair $i,j \in \mathcal{K}_n$ with $h^n_{i} > h^n_{k}$, the condition $\gamma^n_{k,i} (\boldsymbol{p}^n) \geq \gamma^n_{k} (\boldsymbol{p}^n)$ or equivalently $p^n_{k} h^n_{i} \geq p^n_{k} h^n_{k}$ holds independent of $\boldsymbol{p}^n$. Accordingly, at any $\boldsymbol{\boldsymbol{p}}^n$, the achievable rate of each user $k \in \mathcal{K}_n$ on subchannel $n$ is equal to its channel capacity formulated by \cite{NIFbook}
\begin{equation}\label{achiev se}
	R^n_{k} (\boldsymbol{p}^n) = W_s \log_2\left( 1+ \gamma^n_{k} (\boldsymbol{p}^n) \right).
\end{equation}
The overall achievable rate of user $k \in \mathcal{K}$ can thus be obtained by $R_{k} (\boldsymbol{p}) = \sum\limits_{n \in \mathcal{N}_k} R^n_{k} (\boldsymbol{p}^n)$.

\subsection{Optimization Problem Formulations}
Assume that the set of clusters, i.e., $\mathcal{K}_n,~\forall n \in \mathcal{N},$ is predefined. The general power allocation problem for maximizing users SR in Hybrid-NOMA is formulated by
\begin{subequations}\label{sumrate problem}
	\begin{align}\label{obf sumrate problem}
		\max_{ \boldsymbol{p} \geq 0}\hspace{.0 cm}	
		~~ & \sum\limits_{n \in \mathcal{N}} \sum\limits_{k \in \mathcal{K}_n} R^n_{k} (\boldsymbol{p}^n)
		\\
		\text{s.t.}~~
		\label{Constraint minrate}
		& R^n_{k} (\boldsymbol{p}^n) \geq R^{\rm{min}}_{k,n},~\forall k \in \mathcal{K},~n \in \mathcal{N}_k,
		\\
		\label{Constraint cell power}
		& \sum\limits_{n \in \mathcal{N}} \sum\limits_{k \in \mathcal{K}_n} p^n_{k} \leq P^{\rm{max}},
		\\
		\label{Constraint mask power}
		& \sum\limits_{k \in \mathcal{K}_n} p^n_{k} \leq P^{\rm{mask}}_n, \forall n \in \mathcal{N},
	\end{align}
\end{subequations}
where \eqref{Constraint minrate} is the per-subchannel minimum rate constraint, in which $R^{\rm{min}}_{k,n}$ is the individual minimum rate demand of user $k$ on subchannel $n$ \cite{7982784,7523951,9276828,8119791,8448840,9032198}. \eqref{Constraint cell power} is the cellular power constraint, where $P^{\rm{max}}$ denotes the maximum available power of the BS. \eqref{Constraint mask power} is the maximum per-subchannel power constraint, where $P^{\rm{mask}}_n$ denotes the maximum allowable power on subchannel\footnote{We do not impose any specific condition on $P^{\rm{mask}}_n$. We only take into account $P^{\rm{mask}}_n$ in our analysis to keep the generality, such that $P^{\rm{mask}}_n \geq P^{\rm{max}},~\forall n \in \mathcal{N}$, as special case.} $n$. For convenience, we denote the general power allocation matrix as $\boldsymbol{p}=[\boldsymbol{p}^n],\forall n \in \mathcal{N}$.

The overall system EE is formulated by $E(\boldsymbol{p})=\frac{\sum\limits_{n \in \mathcal{N}} \sum\limits_{k \in \mathcal{K}_n} R^n_{k} (\boldsymbol{p}^n)}{\sum\limits_{n \in \mathcal{N}} \sum\limits_{k \in \mathcal{K}_n} p^n_{k} + P_{\rm{C}}}$, where constant $P_{\rm{C}}$ is the circuit power consumption \cite{6213038,Jorswieckfractional}. The power allocation problem for maximizing system EE under the individual minimum rate demand of users in Hybrid-NOMA is formulated by
\begin{equation}\label{EE problem}
	\max_{ \boldsymbol{p} \geq 0}~E(\boldsymbol{p})
	~~~~~~~~~\text{s.t.}~\eqref{Constraint minrate}\text{-}\eqref{Constraint mask power}.
\end{equation}
The main notations of the paper are summarized in Table \ref{table notations}.
\begin{table}[tp]
	\caption{Main Notations.}
	\begin{center} \label{table notations}
		\scalebox{1}{\begin{tabular}{|c|c|}
				\hline \rowcolor[gray]{0.910}
				\textbf{Notation} & \textbf{Description} 
				\\ \hline \rowcolor[gray]{0.960}
				$\mathcal{K}$ & Set of all the users
				\\ \hline \rowcolor[gray]{0.960}
				$\mathcal{K}_n$ & Set of users on subchannel $n$
				\\ \hline \rowcolor[gray]{0.960}
				$\mathcal{N}$ & Set of subchannels
				\\ \hline \rowcolor[gray]{0.960}
				$\mathcal{N}_k$ & Set of subchannels occupied by user $k$
				\\ \hline \rowcolor[gray]{0.960}
				$W_s$ & Bandwidth of each subchannel
				\\ \hline \rowcolor[gray]{0.960}
				$U^{\rm{max}}$ & Maximum number of multiplexed users
				\\ \hline \rowcolor[gray]{0.960}
				$\rho^n_{k}$ & Channel allocation indicator for user $k$ and subchannel $n$
				\\ \hline \rowcolor[gray]{0.960}
				$p^n_{k}$ & Allocated power to user $k$ on subchannel $n$
				\\ \hline \rowcolor[gray]{0.960}
				$h^n_{k}$ & CNR of user $k$ on subchannel $n$
				\\ \hline \rowcolor[gray]{0.960}
				$\Phi_n$ & Index of the cluster-head user on subchannel $n$
				\\ \hline \rowcolor[gray]{0.960}
				$R^n_{k}$ & Achievable rate of user $k$ on subchannel $n$
				\\ \hline \rowcolor[gray]{0.960}
				$\boldsymbol{p}^n$ & Vector of allocated powers on subchannel $n$
				\\ \hline \rowcolor[gray]{0.960}
				$\boldsymbol{p}$ & Matrix of power allocation 
				\\ \hline \rowcolor[gray]{0.960}
				$R^{\rm{min}}_{k,n}$ & Minimum rate demand of user $k$ on subchannel $n$
				\\ \hline \rowcolor[gray]{0.960}
				$P^{\rm{max}}$ & Maximum transmit power of the BS
				\\ \hline \rowcolor[gray]{0.960}
				$P^{\rm{mask}}_n$ & Maximum allowable power of subchannel $n$
				\\ \hline \rowcolor[gray]{0.960}
				$P_{\rm{C}}$ & BS's circuit power consumption
				\\ \hline \rowcolor[gray]{0.960}
				$E(\boldsymbol{p})$ & System EE
				\\ \hline \rowcolor[gray]{0.960}
				$q_n$ & Power consumption of cluster $n$
				\\ \hline \rowcolor[gray]{0.960}
				$Q^{\rm{min}}_n$ & Lower-bound of $q_n$
				\\ \hline \rowcolor[gray]{0.960}
				$H_n$ & Effective CNR of virtual OMA user $n$
				\\ \hline \rowcolor[gray]{0.960}
				$P_{\rm{EE}}$ & Power consumption of the BS
				\\ \hline \rowcolor[gray]{0.960}
				$P_{\rm{min}}$ & Lower-bound of $P_{\rm{EE}}$ 
				\\ \hline \rowcolor[gray]{0.960}
				$\lambda$ & Fractional parameter
				\\ \hline \rowcolor[gray]{0.960}
		\end{tabular}}
	\end{center}
\end{table}

\section{Solution Algorithms}\label{Section solution}
In this section, we propose globally optimal power allocation algorithms for the SR and EE maximization problems. The closed-form of optimal powers for the total power minimization problem is also derived to characterize the feasible set of our target problems.

\subsection{Sum-Rate Maximization Problem}
Here, we propose a water-filling algorithm to find the globally optimal solution of \eqref{sumrate problem}. The SR of users in each cluster, i.e., $\sum\limits_{k \in \mathcal{K}_n} R^n_{k} (\boldsymbol{p}^n)$ is strictly concave in $\boldsymbol{p}^n$, since its Hessian is negative definite \cite{Boydconvex}. For more details, please see Appendix A in \cite{8352643}. The overall SR in \eqref{obf sumrate problem} is thus strictly concave in $\boldsymbol{p}$, since it is the positive summation of strictly concave functions. Besides, the power constraints in \eqref{Constraint cell power} and \eqref{Constraint mask power} are affine, so are convex. The minimum rate constraint in \eqref{Constraint minrate} can be equivalently transformed to the following affine form as
$2^{(R^{\rm{min}}_{k,n}/W_s)} \bigg(\sum\limits_{j \in \mathcal{K}_n, \atop h^n_{j} > h^n_{k}} p^n_{j} h^n_{k} + 1\bigg) \leq 
\sum\limits_{j \in \mathcal{K}_n, \atop h^n_{j} > h^n_{k}} p^n_{j} h^n_{k} + 1 + p^n_k h^n_k,~\forall k \in \mathcal{K},~n \in \mathcal{N}_k.$
Accordingly, the feasible set of \eqref{sumrate problem} is convex. Summing up, problem \eqref{sumrate problem} is convex in $\boldsymbol{p}$.
Let us define $q_n=\sum\limits_{k \in \mathcal{K}_n} p^n_{k}$ as the power consumption of cluster $n$. Problem \eqref{sumrate problem} can be equivalently transformed to the following joint intra- and inter-cluster power allocation problem as
\begin{subequations}\label{SCuser problem}
	\begin{align}\label{obf SCuser problem}
		\max_{\boldsymbol{p} \geq 0, \boldsymbol{q} \geq 0}\hspace{.0 cm}	
		~~ & \sum\limits_{n \in \mathcal{N}} \sum\limits_{k \in \mathcal{K}_n} R^n_{k} (\boldsymbol{p}^n)
		\\
		\text{s.t.}~~
		\label{Constraint minrate SCuser}
		& R^n_{k} (\boldsymbol{p}^n) \geq R^{\rm{min}}_{k,n},~\forall k \in \mathcal{K},~n \in \mathcal{N}_k,
		\\
		\label{Constraint cell q}
		& \sum\limits_{n \in \mathcal{N}} q_n \leq P^{\rm{max}},
		\\
		\label{Constraint sumpow}
		& \sum\limits_{k \in \mathcal{K}_n} p^n_{k} = q_n,~\forall n \in \mathcal{N},
		\\
		\label{Constraint mask q}
		& 0 \leq q_n \leq P^{\rm{mask}}_n, \forall n \in \mathcal{N}.
	\end{align}
\end{subequations}
In the following, we first convert the feasible set of \eqref{SCuser problem} to the intersection of closed-boxes along with the affine cellular power constraint.
\begin{proposition}\label{proposition feasiblecluster}
	The feasible set of \eqref{SCuser problem} is the intersection of $q_n \in \left[Q^{\rm{min}}_n,P^{\rm{mask}}_n\right],~\forall n \in \mathcal{N}$, and cellular power constraint $\sum\limits_{n \in \mathcal{N}} q_n \leq P^{\rm{max}}$, where the lower-bound constant $Q^{\rm{min}}_n$ is
	\begin{multline}\label{Qmin func}
		Q^{\rm{min}}_n=
		\sum\limits_{k \in \mathcal{K}_n} \beta^n_{k} \Bigg(\prod\limits_{j \in \mathcal{K}_n \atop h^n_{j} > h^n_{k}} \left(1+\beta^n_{j}\right) +\frac{1}{h^n_k}+ \\
		\sum\limits_{j \in \mathcal{K}_n \atop h^n_{j} > h^n_{k}} \frac{ \beta^n_{j} \prod\limits_{l \in \mathcal{K}_n \atop h^n_{k} < h^n_{l} < h^n_{j}} \left(1+\beta^n_{l}\right)}{h^n_j}\Bigg),
	\end{multline}
	in which $\beta^n_k=2^{(R^{\rm{min}}_{k,n}/W_s)} -1,~\forall n \in \mathcal{N},~k \in \mathcal{K}_n$.
\end{proposition}
\begin{proof}
	Please see Appendix \ref{appendix Lemma feasible}.
\end{proof}
The feasibility of problems \eqref{sumrate problem} and \eqref{EE problem} can be immediately determined as follows:
\begin{corollary}\label{corol feasible region}
	Problems \eqref{sumrate problem} and \eqref{EE problem} are feasible if and only if $Q^{\rm{min}}_n \leq P^{\rm{mask}}_n,~\forall n \in \mathcal{N}$, and $\sum\limits_{n \in \mathcal{N}} Q^{\rm{min}}_n \leq P^{\rm{max}}$.
\end{corollary}
In the following, we find the closed-form of optimal intra-cluster power allocation as a linear function of any given feasible $\boldsymbol{q}$, thus satisfying Proposition \ref{proposition feasiblecluster}.
\begin{proposition}\label{proposition jointcluster}
	For any given feasible $\boldsymbol{q}=[q_1,\dots,q_n]$, the optimal intra-cluster powers for each cluster $n \in \mathcal{N}$ can be obtained by
	\begin{equation}\label{opt IC i}
		{p^n_{k}}^*=\left(\beta^n_{k} \prod\limits_{j \in \mathcal{K}_n \atop h^n_{j} < h^n_{k}} \left(1-\beta^n_{j}\right)\right) q_n + c^n_{k},~\forall k \in \mathcal{K}_n\setminus\{\Phi_n\},
	\end{equation}
	and
	\begin{equation}\label{opt IC M}
		{p^n_{\Phi_n}}^*= \left(1 - \sum\limits_{i \in \mathcal{K}_n \atop h^n_{i} < h^n_{\Phi_n}} \beta^n_{i} \prod\limits_{j \in \mathcal{K}_n \atop h^n_{j} < h^n_{i}} \left(1-\beta^n_{j}\right)\right) q_n - \sum\limits_{i \in \mathcal{K}_n \atop h^n_{i} < h^n_{\Phi_n}} c^n_{i},
	\end{equation}
	where $\beta^n_{k} = \frac{2^{(R^{\rm{min}}_{k,n}/W_s)}-1} {2^{(R^{\rm{min}}_{k,n}/W_s)}},~ \forall n \in \mathcal{N},~k \in \mathcal{K}_n$, and
	$c^n_{k}=\beta^n_{k} \left( \frac{1}{h^n_{k}} - \sum\limits_{j \in \mathcal{K}_n \atop h^n_{j} < h^n_{k}} \frac{\prod\limits_{l \in \mathcal{K}_n \atop h^n_{j} < h^n_{l} < h^n_{k}} \left(1-\beta^n_{l}\right) \beta^n_{j}} {h^n_{j}}\right),~ \forall n \in \mathcal{N},~k \in \mathcal{K}_n$.
\end{proposition}
\begin{proof}
	Please see Appendix \ref{appendix closedform pow sum-rate}.
\end{proof}
Since the closed-form expressions of optimal intra-cluster power allocation in Proposition \ref{proposition jointcluster} are valid for any given feasible $\boldsymbol{q}$, we can substitute \eqref{opt IC i} and \eqref{opt IC M} directly to problem \eqref{SCuser problem}. For convenience, we first rewrite \eqref{opt IC M} as
\begin{equation}\label{opt IC M qmin}
	{p^n_{\Phi_n}}^*=\alpha_n q_n - c_n,~\forall n \in \mathcal{N},
\end{equation}
where $\alpha_n=\left(1 - \sum\limits_{i \in \mathcal{K}_n \atop h^n_{i} < h^n_{\Phi_n}} \beta^n_{i} \prod\limits_{j \in \mathcal{K}_n \atop h^n_{j} < h^n_{i}} \left(1-\beta^n_{j}\right)\right),~\forall n \in \mathcal{N}$, and $c_n=\sum\limits_{i \in \mathcal{K}_n \atop h^n_{i} < h^n_{\Phi_n}} c^n_{i},~\forall n \in \mathcal{N}$, are nonnegative constants. According to the proof of Proposition \ref{proposition jointcluster} and \eqref{opt IC M qmin}, the SR function of each cluster $n \in \mathcal{N}$ at the optimal point can be formulated as a function of $q_n$ given by
\begin{multline}\label{optimal value}
	R^n_{\rm{opt}} (q_n)=\sum\limits_{k \in \mathcal{K}_n} R^n_{k} (\boldsymbol{p^*}^n) = \sum\limits_{k \in \mathcal{K}_n \atop k \neq \Phi_n} (R^{\rm{min}}_{k,n}) + R^n_{\Phi_n} (q_n) \\
	=\sum\limits_{k \in \mathcal{K}_n \atop k \neq \Phi_n} (R^{\rm{min}}_{k,n}) + W_s \log_2\left( 1 + \left(\alpha_n q_n - c_n\right) h^n_{\Phi_n} \right), \forall n \in \mathcal{N}.
\end{multline}
By utilizing Proposition \ref{proposition feasiblecluster} and \eqref{optimal value}, the joint intra- and inter-cluster power allocation problem \eqref{SCuser problem} can be equivalently transformed to the following inter-cluster power allocation problem
\begin{subequations}\label{SCuser problem eq} 
	\begin{align}\label{obf SCuser problem eq}
		\max_{\boldsymbol{q}}\hspace{.0 cm} 
		~~ & \sum\limits_{n \in \mathcal{N}} W_s \log_2\left( 1 + \left(\alpha_n q_n - c_n\right) h^n_{\Phi_n} \right)
		\\
		\text{s.t.}~~
		\label{Constraint cell eq}
		& \sum\limits_{n \in \mathcal{N}} q_n = P^{\rm{max}},
		\\
		\label{Constraint mask eq}
		& q_n \in [Q^{\rm{min}}_n,P^{\rm{mask}}_n],~\forall n \in \mathcal{N}.
	\end{align}
\end{subequations}
Let us define $\boldsymbol{\tilde{q}}=[\tilde{q}_n],\forall n \in \mathcal{N}$, where $\tilde{q}_n=q_n - \frac{c_n}{\alpha_n},~\forall n \in \mathcal{N}$. Hence, \eqref{SCuser problem eq} can be transformed to the following equivalent OMA problem as
\begin{subequations}\label{eq1 problem} 
	\begin{align}\label{obf eq1 problem}
		\max_{\boldsymbol{\tilde{q}}}\hspace{.0 cm} 
		~~ & \sum\limits_{n \in \mathcal{N}} W_s \log_2\left( 1 + \tilde{q}_n H_n \right)
		\\
		\text{s.t.}~~
		\label{Constraint 1}
		& \sum\limits_{n \in \mathcal{N}} \tilde{q} = \tilde{P}^{\rm{max}},
		\\
		\label{Constraint mask eq 2}
		& \tilde{q}_n \in [\tilde{Q}^{\rm{min}}_n,\tilde{P}^{\rm{mask}}_n],\forall n \in \mathcal{N},
	\end{align}
\end{subequations}
where $H_n=\alpha_n h^n_{\Phi_n},~\forall n \in \mathcal{N}$, $\tilde{P}^{\rm{max}}=P^{\rm{max}} - \sum\limits_{n \in \mathcal{N}} \frac{c_n}{\alpha_n}$, $\tilde{Q}^{\rm{min}}_n=Q^{\rm{min}}_n - \frac{c_n}{\alpha_n},~\forall n \in \mathcal{N}$, and $\tilde{P}^{\rm{mask}}_n=P^{\rm{mask}}_n - \frac{c_n}{\alpha_n},~\forall n \in \mathcal{N}$. 
Constraint \eqref{Constraint 1} is the affine cellular power constraint, and \eqref{Constraint mask eq 2} is derived based on Proposition \ref{proposition feasiblecluster}. The objective function \eqref{obf eq1 problem} is strictly concave in $\boldsymbol{\tilde{q}}$, and the feasible set of \eqref{eq1 problem} is affine, so is convex. Accordingly, problem \eqref{eq1 problem} is convex.
The equivalent FDMA problem \eqref{eq1 problem} can be optimally solved by using the well-known water-filling algorithm \cite{1381759,1576943,6547819,10.1155/2008/643081,8995606}. After some mathematical manipulations, the optimal $\tilde{q}^*_n$ can be obtained as
\begin{align}\label{bisection opt form}
	\tilde{q}^*_n=
	\begin{cases}
		\frac{W_s/(\ln 2)}{\nu^*} - \frac{1}{H_n}, &\quad \left(\frac{W_s/(\ln 2)}{\nu^*} - \frac{1}{H_n}\right) \in [\tilde{Q}^{\rm{min}}_n,\tilde{P}^{\rm{mask}}_n], \\
		0, &\quad \text{otherwise}, \\ 
	\end{cases}
\end{align}
such that $\boldsymbol{\tilde{q}}^*=[\tilde{q}^*_n],~\forall n \in \mathcal{N},$ satisfies \eqref{Constraint 1}. Moreover, $\nu^*$ is the dual optimal corresponding to constraint \eqref{Constraint 1}. For more details, please see Appendix \ref{appendix bisection sumrate}.
The pseudo-code of the bisection method for finding $\nu^*$ is presented in Alg. \ref{Alg bisection}.
\begin{algorithm}[tp]
	\caption{The bisection method for finding $\nu^*$ in \eqref{bisection opt form}.} \label{Alg bisection}
	\begin{algorithmic}[1]
		\STATE Initialize tolerance $\epsilon$, lower-bound $\nu_l$, upper-bound $\nu_h$, and maximum iteration $L$.
		\FOR {$l=1:L$}
		\STATE Set $\nu_m=\frac{\nu_l + \nu_h}{2}$.
		\\
		\STATE\textbf{if}~~$\sum\limits_{n \in \mathcal{N}} \max \left\{\tilde{Q}^{\rm{min}}_n, \min \left\{ \left(\frac{W_s/(\ln 2)}{\nu_m} - \frac{1}{H_n}\right) , \tilde{P}^{\rm{mask}}_n\right\} \right\} < \tilde{P}^{\rm{max}}$~~\textbf{then}
		\\
		\STATE~~~~~Set $\nu_h=\nu_m$.
		\\
		\STATE\textbf{else}
		\\~~~~~Set $\nu_l=\nu_m$.
		\\
		\STATE\textbf{end if}
		\STATE\textbf{if}~~$\frac{\tilde{P}^{\rm{max}} - \sum\limits_{n \in \mathcal{N}} \max \left\{\tilde{Q}^{\rm{min}}_n, \min \left\{ \left(\frac{W_s/(\ln 2)}{\nu_m} - \frac{1}{H_n}\right) , \tilde{P}^{\rm{mask}}_n\right\} \right\}}{\tilde{P}^{\rm{max}}} \leq \epsilon$~~\textbf{then}
		\\
		\STATE~~~~~\textbf{break}.
		\\
		\STATE\textbf{end if}
		\ENDFOR
	\end{algorithmic}
\end{algorithm}
After finding $\boldsymbol{\tilde{q}}^*$, we obtain $\boldsymbol{q}^*$ by using $q^*_n=(\tilde{q}^*_n + \frac{c_n} {\alpha_n}),~\forall n \in \mathcal{N}$. Then, we find the optimal intra-cluster power allocation according to Proposition \ref{proposition jointcluster}. Since problems \eqref{sumrate problem} and \eqref{eq1 problem} are equivalent, the obtained globally optimal solution for \eqref{eq1 problem} is also globally optimal for \eqref{sumrate problem}.

\subsection{Energy Efficiency Maximization Problem}
In this subsection, we find a globally optimal solution for problem \eqref{EE problem}.
The feasible region of problem \eqref{EE problem} is identical to the feasible region of problem \eqref{sumrate problem}. Hence, Proposition \ref{proposition feasiblecluster} can be used to characterize the feasible region of problem \eqref{EE problem}. 
Let us define $P_{\rm{EE}}=\sum\limits_{n \in \mathcal{N}} q_n$ as the cellular power consumption in the EE maximization problem \eqref{EE problem}. For any given $P_{\rm{EE}}$, problem \eqref{EE problem} can be equivalently transformed to the SR maximization problem \eqref{sumrate problem} in which $P^{\rm{max}}=P_{\rm{EE}}$. As a result, the globally optimal solution of \eqref{EE problem} can be obtained by exploring different values of $P_{\rm{EE}} \in \left[\sum\limits_{n \in \mathcal{N}} Q^{\rm{min}}_n,P^{\rm{max}}\right]$, and applying the water-filling Alg. \ref{Alg bisection}, in which $P^{\rm{max}}=P_{\rm{EE}}$. Exploring $P_{\rm{EE}} \in \left[\sum\limits_{n \in \mathcal{N}} Q^{\rm{min}}_n,P^{\rm{max}}\right]$ may be computationally prohibitive, specifically when the stepsize of exhaustive search is small and/or $\sum\limits_{n \in \mathcal{N}} Q^{\rm{min}}_n \to 0$, e.g., when the users have small minimum rate demands.

The SR function in the numerator of the EE function in \eqref{EE problem} is strictly concave in $\boldsymbol{p}$. The denominator of the EE function is an affine function, so is convex. Therefore, problem \eqref{EE problem} is a concave-convex fractional program with a pseudoconcave objective function \cite{6213038,Jorswieckfractional}. The pseudoconcavity of the objective function in \eqref{EE problem} implies that any stationary point is indeed globally optimal and the Karush–Kuhn–Tucker (KKT) optimality conditions are sufficient if a constraint qualification is fulfilled \cite{6213038,Jorswieckfractional}. For more details, please see Appendix \ref{appendix converge Dinkelbach}. Hence, the globally optimal solution of \eqref{EE problem} can be obtained by using the well-known Dinkelbach algorithm \cite{6213038,Jorswieckfractional}. In this algorithm, we iteratively solve the following problem
\begin{align}\label{EEfrac problem}
	\max_{ \boldsymbol{p} \geq 0}~&F(\lambda,\boldsymbol{p})=\left(\sum\limits_{n \in \mathcal{N}} \sum\limits_{k \in \mathcal{K}_n} R^n_{k} (\boldsymbol{p}^n)\right) - \lambda
	\left(\sum\limits_{n \in \mathcal{N}} \sum\limits_{k \in \mathcal{K}_n} p^n_{k} + P_{\rm{C}}\right) \nonumber
	\\
	\text{s.t.}~&\eqref{Constraint minrate}\text{-}\eqref{Constraint mask power},
\end{align}
where $\lambda \geq 0$ is the fractional parameter, and $F(\lambda,\boldsymbol{p})$ is strictly concave in $\boldsymbol{p}$. 
This algorithm is described as follows: We first initialize parameter $\lambda_{(0)}$ such that $F\left(\lambda_{(0)},\boldsymbol{p}^*\right) \geq 0$.
At each iteration $(t)$, we set $\lambda_{(t)}=E(\boldsymbol{p}^*_{(t-1)})$, where $\boldsymbol{p}^*_{(t-1)}$ is the optimal solution obtained from the prior iteration $(t-1)$. After that, we find $\boldsymbol{p}^*_{(t)}$ by solving \eqref{EEfrac problem} in which $\lambda=\lambda_{(t)}$. We repeat the iterations until $|F\left(\lambda_{(t)},\boldsymbol{p}^*_{t}\right)| \leq \Upsilon$, where $\Upsilon$ is a tolerance tuning the optimality gap. The pseudo-code of the Dinkelbach algorithm for solving \eqref{EE problem} is presented in Alg. \ref{Alg dinkelbach}.
\begin{algorithm}[tp]
	\caption{The Dinkelbach method for solving the energy efficiency maximization problem.} \label{Alg dinkelbach}
	\begin{algorithmic}[1]
		\STATE Initialize parameter $\lambda_{(0)}$ satisfying $F(\lambda_{(0)},\boldsymbol{p}^*) \geq 0$, tolerance $\Upsilon$ (sufficiently small), and $t=0$.
		\WHILE {$F(\lambda,\boldsymbol{p}^*) > \Upsilon$}
		\STATE Set $\lambda=E(\boldsymbol{p}^*)$. Then, solve \eqref{EEfrac problem} and find $\boldsymbol{p}^*$.
		\STATE\textbf{if}~~$|F\left(\lambda,\boldsymbol{p}^*\right)| \leq \Upsilon$~\textbf{then}
		\STATE~~~~\textbf{break}.
		\STATE\textbf{end if}
		\ENDWHILE
	\end{algorithmic}
\end{algorithm}
Similar to the transformation of \eqref{sumrate problem} to \eqref{SCuser problem}, we define $q_n=\sum\limits_{k \in \mathcal{K}_n} p^n_{k}$ as the power consumption of cluster $n$. The main problem \eqref{EEfrac problem} can be equivalently transformed to the following joint intra- and inter-cluster power allocation problem as
\begin{equation}\label{EEq problem}
	\max_{\boldsymbol{p} \geq 0, \boldsymbol{q} \geq 0}	
	~\left(\sum\limits_{n \in \mathcal{N}} \sum\limits_{k \in \mathcal{K}_n} R^n_{k} (\boldsymbol{p}^n)\right) - \lambda
	\left(\sum\limits_{n \in \mathcal{N}} q_n + P_{\rm{C}}\right)
	~~~~~\text{s.t.}~\eqref{Constraint minrate SCuser}\text{-}\eqref{Constraint mask q}.
\end{equation}
The feasible set of problems \eqref{SCuser problem} and \eqref{EEq problem} is identical, thus the feasibility of \eqref{EEq problem} can be characterized by Proposition \ref{proposition feasiblecluster}. 
\begin{proposition}\label{proposition intraEE}
	For any given feasible $\boldsymbol{q}$, the optimal intra-cluster power allocation in problem \eqref{EEq problem} can be obtained by using \eqref{opt IC i} and \eqref{opt IC M}.
\end{proposition}
\begin{proof}
	When $\boldsymbol{q}$ is fixed, the second term $\lambda
	\left(\sum\limits_{n \in \mathcal{N}} q_n + P_{\rm{C}}\right)$ in \eqref{EEq problem} is a constant. Hence, the objective function of \eqref{EEq problem} can be equivalently rewritten as maximizing users SR given by $\max\limits_{\boldsymbol{p} \geq 0}	
	~\left(\sum\limits_{n \in \mathcal{N}} \sum\limits_{k \in \mathcal{K}_n} R^n_{k} (\boldsymbol{p}^n)\right)$, which is independent of $\lambda$. Hence, for any given feasible $\boldsymbol{q}$, problems \eqref{EEq problem} and \eqref{SCuser problem} are identical. Accordingly, Proposition \ref{proposition jointcluster} also holds for any given feasible $\boldsymbol{q}$, and $\lambda$ in \eqref{EEq problem}.
\end{proof}
Similar to the SR maximization problem \eqref{SCuser problem}, we substitute \eqref{opt IC i} and \eqref{opt IC M} to problem \eqref{EEq problem}. By utilizing Proposition \ref{proposition feasiblecluster} and \eqref{optimal value}, the joint intra- and inter-cluster power allocation problem \eqref{EEq problem} can be equivalently transformed to the following inter-cluster power allocation problem
\begin{subequations}\label{EE problem eq} 
	\begin{align}\label{obf EE problem eq}
		\max_{\boldsymbol{q}}\hspace{.0 cm} 
		& \hat{F}(\boldsymbol{q})= \left(\sum\limits_{n \in \mathcal{N}} W_s \log_2\left( 1 + \left(\alpha_n q_n - c_n\right) h^n_{\Phi_n} \right) \right) - \lambda
		\left(\sum\limits_{n \in \mathcal{N}} q_n\right)
		\\
		\text{s.t.}~~
		\label{Constraint EEcell eq}
		& \sum\limits_{n \in \mathcal{N}} q_n \leq P^{\rm{max}},~~~q_n \in [Q^{\rm{min}}_n,P^{\rm{mask}}_n],\forall n \in \mathcal{N},
	\end{align}
\end{subequations}
where $\alpha_n$ and $c_n$ are defined in \eqref{opt IC M qmin}. Note that since $\lambda$ and $P_{\rm{C}}$ are constants, the term $-\lambda P_{\rm{C}}$ can be removed from \eqref{EEq problem}, so is removed in \eqref{obf EE problem eq} during the equivalent transformation. The differences between problems \eqref{SCuser problem eq} and \eqref{EE problem eq} are the additional term $-\lambda \left(\sum\limits_{n \in \mathcal{N}} q_n\right)$ in $\hat{F}(\boldsymbol{q})$, and also inequality constraint \eqref{Constraint EEcell eq}. 
\begin{proposition}\label{proposition bisectionEE}
	At the optimal point of the EE maximization problem \eqref{EEfrac problem},  if 
	\begin{equation}\label{EE fullbudget}
		\frac{W_s/(\ln 2)}{\lambda} - \frac{1 - c_n h^n_{\Phi_n}}{\alpha_n h^n_{\Phi_n}} > P^{\rm{mask}}_n,~\forall n \in \mathcal{N},
	\end{equation}
	the cellular power constraint \eqref{Constraint cell power} is active, meaning that $\sum\limits_{n \in \mathcal{N}} q^*_n = P^{\rm{max}}$.
\end{proposition}
\begin{proof} 
	The optimal solution of \eqref{EE problem eq} is unique if and only if the objective function \eqref{obf EE problem eq} is strictly concave. For the case that the concave function in \eqref{obf EE problem eq} is increasing in $\boldsymbol{q}$, we can guarantee that at the optimal point, the cellular power constraint \eqref{Constraint EEcell eq} is active. In other words, for the case that $\frac{\partial \hat{F}(\boldsymbol{q})}{\partial q_n} > 0,~\forall n \in \mathcal{N}$, for any $q_n \in [Q^{\rm{min}}_n,P^{\rm{mask}}_n],~\forall n \in \mathcal{N}$, the optimal $\boldsymbol{q}^*$ satisfies $\sum\limits_{n \in \mathcal{N}} q^*_n = P^{\rm{max}}$. In this case, the cellular power constraint \eqref{Constraint EEcell eq} can be replaced with $\sum\limits_{n \in \mathcal{N}} q_n = P^{\rm{max}}$, thus the optimization problem \eqref{EE problem eq} can be equivalently transformed to the SR maximization problem \eqref{SCuser problem eq} whose globally optimal solution can be obtained by Alg. \ref{Alg bisection}.
	In the following, we find a sufficient condition, where it is guaranteed that $\frac{\partial \hat{F}(\boldsymbol{q})}{\partial q_n} > 0,~\forall n \in \mathcal{N}$, for any $q_n \in [Q^{\rm{min}}_n,P^{\rm{mask}}_n],~\forall n \in \mathcal{N}$. The condition $\frac{\partial \hat{F}(\boldsymbol{q})}{\partial q_n} > 0,~\forall n \in \mathcal{N}$ can be rewritten as
	$
	\frac{W_s \alpha h^n_{\Phi_n}/(\ln 2)}{1 + \left(\alpha_n q_n - c_n\right) h^n_{\Phi_n}} - \lambda > 0,~\forall n \in \mathcal{N}.
	$
	After some mathematical manipulations, the latter inequality is rewritten as
	\begin{equation}\label{suf EE bound}
		q_n < \frac{W_s/(\ln 2)}{\lambda} - \frac{1-c_n h^n_{\Phi_n}}{\alpha_n h^n_{\Phi_n}},~\forall n \in \mathcal{N}.
	\end{equation}
	The right-hand side of \eqref{suf EE bound} is a constant providing an upper-bound for the region of $\boldsymbol{q}$ such that $\frac{\partial \hat{F}(\boldsymbol{q})}{\partial q_n} > 0,~\forall n \in \mathcal{N}$. The inequality in \eqref{suf EE bound} holds for any $q_n \in [Q^{\rm{min}}_n,P^{\rm{mask}}_n],~\forall n \in \mathcal{N}$, if and only if $\frac{W_s/(\ln 2)}{\lambda} - \frac{1-c_n h^n_{\Phi_n}}{\alpha_n h^n_{\Phi_n}} > P^{\rm{mask}}_n,~\forall n \in \mathcal{N}$, and the proof is completed.
\end{proof}
If \eqref{EE fullbudget} holds for the given $\lambda$, we guarantee that $\sum\limits_{n \in \mathcal{N}} q^*_n = P^{\rm{max}}$, meaning that the EE problem \eqref{EE problem eq} can be equivalently transformed to the SR maximization problem \eqref{SCuser problem eq} whose globally optimal solution is obtained by using Alg. \ref{Alg bisection}.

For the case that \eqref{EE fullbudget} does not hold, Alg. \ref{Alg bisection} may be suboptimal for \eqref{EE problem eq}. In this case, similar to the transformation of \eqref{SCuser problem eq}  to \eqref{eq1 problem}, we define $\boldsymbol{\tilde{q}}=[\tilde{q}_n],~\forall n \in \mathcal{N}$, where $\tilde{q}_n=q_n - \frac{c_n}{\alpha_n},~\forall n \in \mathcal{N}$. Problem \eqref{EE problem eq} can thus be rewritten as
\begin{subequations}\label{EEeq problem} 
	\begin{align}\label{obf EEeq problem}
		\max_{\boldsymbol{\tilde{q}}}\hspace{.0 cm} 
		~~ & \sum\limits_{n \in \mathcal{N}} W_s \log_2\left( 1 + \tilde{q}_n H_n \right) - \lambda
		\left(\sum\limits_{n \in \mathcal{N}} \tilde{q}_n\right)
		\\
		\text{s.t.}~~
		\label{ConstraintEEeq 1}
		& \sum\limits_{n \in \mathcal{N}} \tilde{q}_n \leq \tilde{P}^{\rm{max}},~~~q_n \in [\tilde{Q}^{\rm{min}}_n,\tilde{P}^{\rm{mask}}_n],\forall n \in \mathcal{N},
	\end{align}
\end{subequations}
where $H_n=\alpha_n h^n_{\Phi_n},~\forall n \in \mathcal{N}$, $\tilde{P}^{\rm{max}}=P^{\rm{max}} - \sum\limits_{n \in \mathcal{N}} \frac{c_n}{\alpha_n}$, $\tilde{Q}^{\rm{min}}_n=Q^{\rm{min}}_n - \frac{c_n}{\alpha_n},~\forall n \in \mathcal{N}$, and $\tilde{P}^{\rm{mask}}_n=P^{\rm{mask}}_n - \frac{c_n}{\alpha_n},~\forall n \in \mathcal{N}$. The equivalent FDMA convex problem \eqref{EEeq problem} can be solved by using the Lagrange dual method with subgradient algorithm or interior point methods (IPMs) \cite{Boydconvex,MinimizationMethods,boydsubgradient}. 
The derivations of the subgradient algorithm for solving \eqref{EEeq problem} is provided in Appendix \ref{appendix subgradient}. Moreover, the derivations of the barrier algorithm with inner Newton's method for solving \eqref{EEeq problem} is provided in Appendix \ref{appendix barrier}.
\begin{algorithm}[tp]
	\caption{The mixed water-filling/subgradient method for solving problem \eqref{EEfrac problem}.} \label{Alg mixedWFSub}
	\begin{algorithmic}[1]
		\STATE Calculate $\Phi_n=\left(\frac{W_s/(\ln 2)}{\lambda} - \frac{1 - c_n h^n_{\Phi_n}}{\alpha_n h^n_{\Phi_n}}\right) - P^{\rm{mask}}_n,~\forall n \in \mathcal{N}$.
		\STATE\textbf{if}~~$\min\limits_{n \in \mathcal{N}}\{\Phi_n\} > 0$~\textbf{then}
		\STATE~~~~Find $\boldsymbol{q}^*$ by using the water-filling Alg. \ref{Alg bisection}.
		\STATE\textbf{else}~~
		\STATE~~~~Initialize Lagrange multiplier $\nu^{(0)}$, step size $\epsilon_\text{s}$, and iteration index $t=0$.
		\STATE~~~~\textbf{repeat}
		\STATE~~~~~~~~Set $t:=t+1$.
		\STATE~~~~~~~~~~~~Find $\boldsymbol{\tilde{q}}^{(t)}$ by using $\tilde{q}^{(t)}_n= \left[\frac{W_s/(\ln 2) }{\lambda + \nu^{(t-1)}} - \frac{1}{H_n}\right]_{\tilde{Q}^{\rm{min}}_n}^{\tilde{P}^{\rm{mask}}_n},~\forall n$.
		\STATE~~~~~~~~Update $\nu^{(t)}=\left[\nu^{(t-1)} - \epsilon_\text{s} \left( \tilde{P}^{\rm{max}} - \sum\limits_{n \in \mathcal{N}} \tilde{q}^{(t)}_n \right)\right]^+$.
		\STATE~~~~\textbf{until} convergence of $\boldsymbol{\tilde{q}}^{(t)}$.
		\STATE~~~~Find $\boldsymbol{q}^*$ by using $q_n=\tilde{q}^*_n + \frac{c_n}{\alpha_n},~\forall n \in \mathcal{N}$.
		\STATE\textbf{end if}
	\end{algorithmic}
\end{algorithm}
According to the above, depending on the value of $\lambda$ at each Dinkelbach iteration, \eqref{EEfrac problem} can be solved by using Alg. \ref{Alg bisection} or subgradient/barrier method. The pseudo-codes of our proposed algorithms for solving \eqref{EEfrac problem} in Step 3 of Alg. \ref{Alg dinkelbach} based on the subgradient and barrier methods are presented in Algs. \ref{Alg mixedWFSub} and \ref{Alg mixedWFBarrier}, respectively.
\begin{algorithm}[tp]
	\caption{The mixed water-filling/barrier method for solving problem \eqref{EEfrac problem}.} \label{Alg mixedWFBarrier}
	\begin{algorithmic}[1]
		\STATE Calculate $\Phi_n=\left(\frac{W_s/(\ln 2)}{\lambda} - \frac{1 - c_n h^n_{\Phi_n}}{\alpha_n h^n_{\Phi_n}}\right) - P^{\rm{mask}}_n,~\forall n \in \mathcal{N}$.
		\STATE\textbf{if}~~$\min\limits_{n \in \mathcal{N}}\{\Phi_n\} > 0$~\textbf{then}
		\STATE~~~~Find $\boldsymbol{q}^*$ by using the water-filling Alg. \ref{Alg bisection}.
		\STATE\textbf{else}~~
		\STATE~~~~Initialize $\boldsymbol{\tilde{q}}$, $0<\alpha<0.5$, $0<\beta<1$, $\mu > 1$, $t \gg 1$, $0 < \epsilon_N \ll 1$, and $0 < \epsilon_B \ll 1$.
		\STATE~~~~\textbf{repeat}
		\STATE~~~~~~~~Set $\Delta \boldsymbol{\tilde{q}} = -\nabla U(\boldsymbol{\tilde{q}})~(\nabla^2 U(\boldsymbol{\tilde{q}}))^{-1}$.
		\item~~~~~~~~Set $\lambda_B=-\Delta \boldsymbol{\tilde{q}}.\nabla U(\boldsymbol{\tilde{q}})^\text{T}$.
		\STATE~~~~~~~~\textbf{if}~~$\lambda_B/2 \leq \epsilon_N$~\textbf{then}
		\STATE~~~~~~~~~~~~\textbf{break}
		\STATE~~~~~~~~\textbf{end if}
		\STATE~~~~~~~~Initialize $l=1$.
		\STATE~~~~~~\textbf{while} $(\tilde{q}_n + l \Delta \boldsymbol{\tilde{q}}_n) \notin [\tilde{Q}^{\rm{min}}_n,\tilde{P}^{\rm{mask}}_n],~\forall n \in \mathcal{N}$, \textbf{or} $\sum\limits_{n \in \mathcal{N}} \left(\tilde{q}_n + l \Delta \boldsymbol{\tilde{q}}_n\right) > \tilde{P}^{\rm{max}}$ \textbf{do}
		\STATE~~~~~~~~~~~~$l:=\beta l$
		\STATE~~~~~~~~\textbf{end while}
		\STATE~~~~~~~~\textbf{while} $U(\boldsymbol{\boldsymbol{\tilde{q}}}+l \Delta \boldsymbol{\tilde{q}}) > U(\boldsymbol{\boldsymbol{\tilde{q}}}) + \alpha l \Delta \boldsymbol{\tilde{q}}.\nabla U(\boldsymbol{\tilde{q}})^\text{T}$ \textbf{do}
		\STATE~~~~~~~~~~~~$l:=\beta l$
		\STATE~~~~~~~~\textbf{end while}
		\STATE~~~~~~~~Set $\boldsymbol{\tilde{q}} = \boldsymbol{\tilde{q}}+ l \Delta \boldsymbol{\tilde{q}}$.
		\STATE~~~~~~~~\textbf{if}~~$1/t \leq \epsilon_B$~\textbf{then}
		\STATE~~~~~~~~~~~~\textbf{break}
		\STATE~~~~~~~~\textbf{end if}
		\STATE~~~~~~~~Set $t:=\mu t$.
		\STATE\textbf{end if}
	\end{algorithmic}
\end{algorithm}
After finding $\boldsymbol{q}^*$ via Algs. \ref{Alg mixedWFSub} or \ref{Alg mixedWFBarrier}, we find the optimal intra-cluster power allocation by using \eqref{opt IC i} and \eqref{opt IC M}. 

\subsection{Important Theoretical Insights of the Optimal Power Allocation for Maximizing SR/EE}
Here, we present the important theoretical insights of optimal power allocation for the SR and EE maximization problems.
\subsubsection{Sum-Rate Maximization}\label{subsection insights sum-rate}
In Hybrid-NOMA, it is guaranteed that at the optimal point, the cellular power constraint is active, meaning that all the available BS's power will be distributed among clusters. According to the proof of Proposition \ref{proposition jointcluster}, it is guaranteed that at the optimal point, only the cluster-head users get additional power, and all the other users get power to only maintain their minimal rate demands on each subchannel. Hence, the remaining cellular power will be distributed among the cluster-head users. According to the analysis of KKT optimality conditions in Appendix \ref{appendix bisection sumrate}, it is shown that there is a competition among cluster-head users to get the rest of the cellular power.
\begin{remark}\label{remark virtualuser}
	In the transformation of \eqref{sumrate problem} to \eqref{eq1 problem}, the Hybrid-NOMA system is equivalently transformed to a virtual FDMA system including a single virtual BS with maximum power $\tilde{P}^{\rm{max}}=P^{\rm{max}} - \sum\limits_{n \in \mathcal{N}} \frac{c_n}{\alpha_n}$, and $N$ virtual OMA users operating in $N$ subchannels with maximum allowable power $\tilde{P}^{\rm{mask}}_n = P^{\rm{mask}}_n - \frac{c_n}{\alpha_n},~\forall n \in \mathcal{N}$. Each cluster $n \in \mathcal{N}$ is indeed a virtual OMA user whose CNR is $H_n=\alpha_n h^n_{\Phi_n}$, which depends on $\alpha_n$ that is a function of the minimum rate demand of users with lower decoding order in cluster $n$, and the CNR of the cluster-head user, whose index is $\Phi_n$.
	The allocated power to the virtual OMA user $n$ is formulated by $\tilde{q}_n=q_n - \frac{c_n}{\alpha_n}$.
	Each virtual OMA user $n$ has also a minimum power demand $\tilde{Q}^{\rm{min}}_n=Q^{\rm{min}}_n - \frac{c_n}{\alpha_n}$, in order to guarantee the individual minimum rate demand of its multiplexed users in $\mathcal{K}_n$ on subchannel $n$. For any given virtual clusters power budget $\boldsymbol{\tilde{q}}=[\tilde{q}_n],\forall n \in \mathcal{N}$, the achievable rate of each virtual OMA user is the SR of its multiplexed users, which is the sum-capacity of subchannel $n$.
\end{remark}
Based on the definition of virtual OMA users for the SR maximization problem in Hybrid-NOMA and the KKT optimality conditions analysis, the exemplary models in Fig. \ref{Fig Systemmodel} can be equivalently transformed to their corresponding virtual FDMA systems shown in Fig. \ref{Fig virtual system}.
\begin{figure*}
	\centering
	\subfigure[SC-NOMA: A single virtual user.]{
		\includegraphics[scale=0.33]{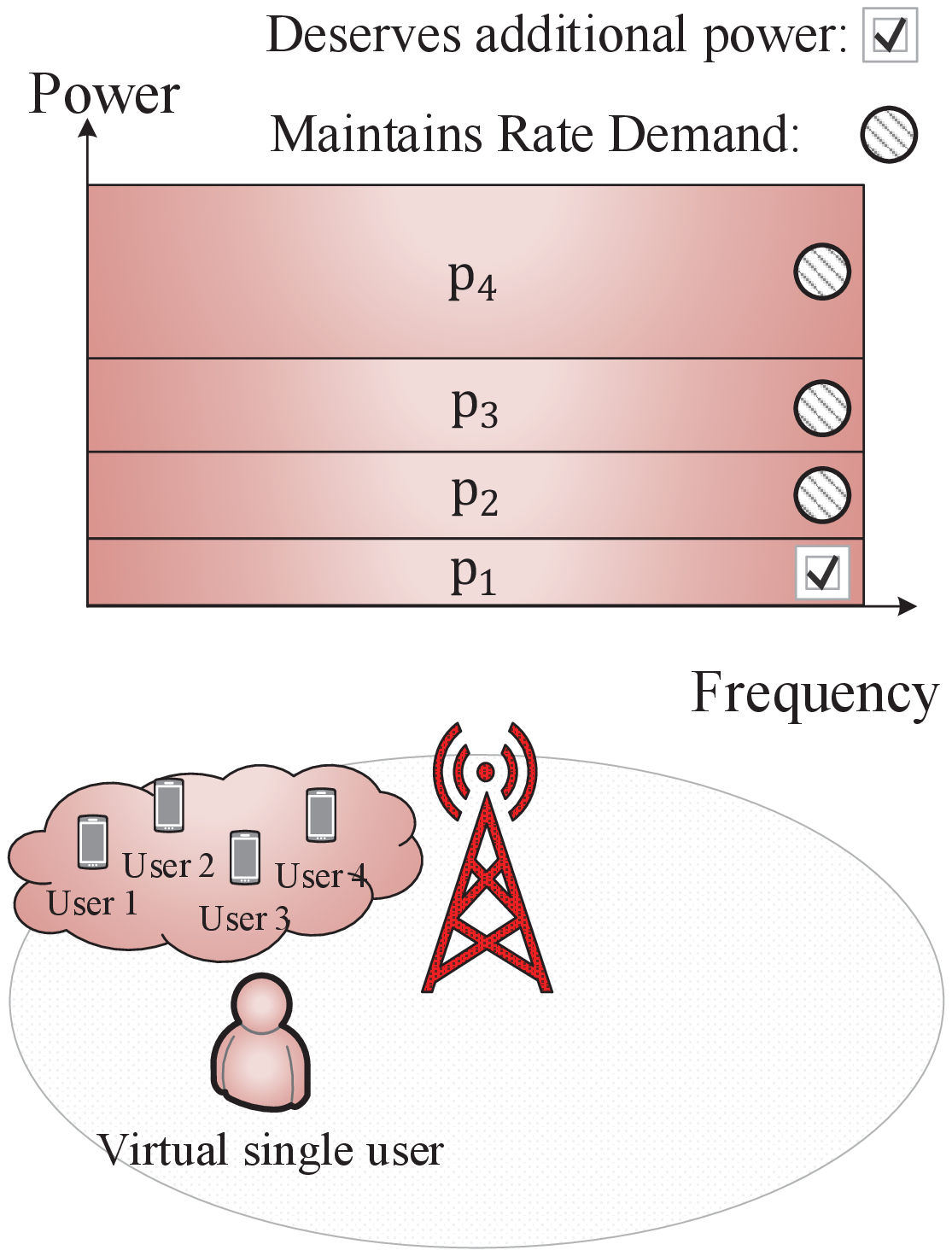}
		\label{Fig_SC-NOMA-Virtual}
	}\hfill
	\subfigure[FDMA-NOMA: FDMA with 2 virtual OMA users.]{
		\includegraphics[scale=0.33]{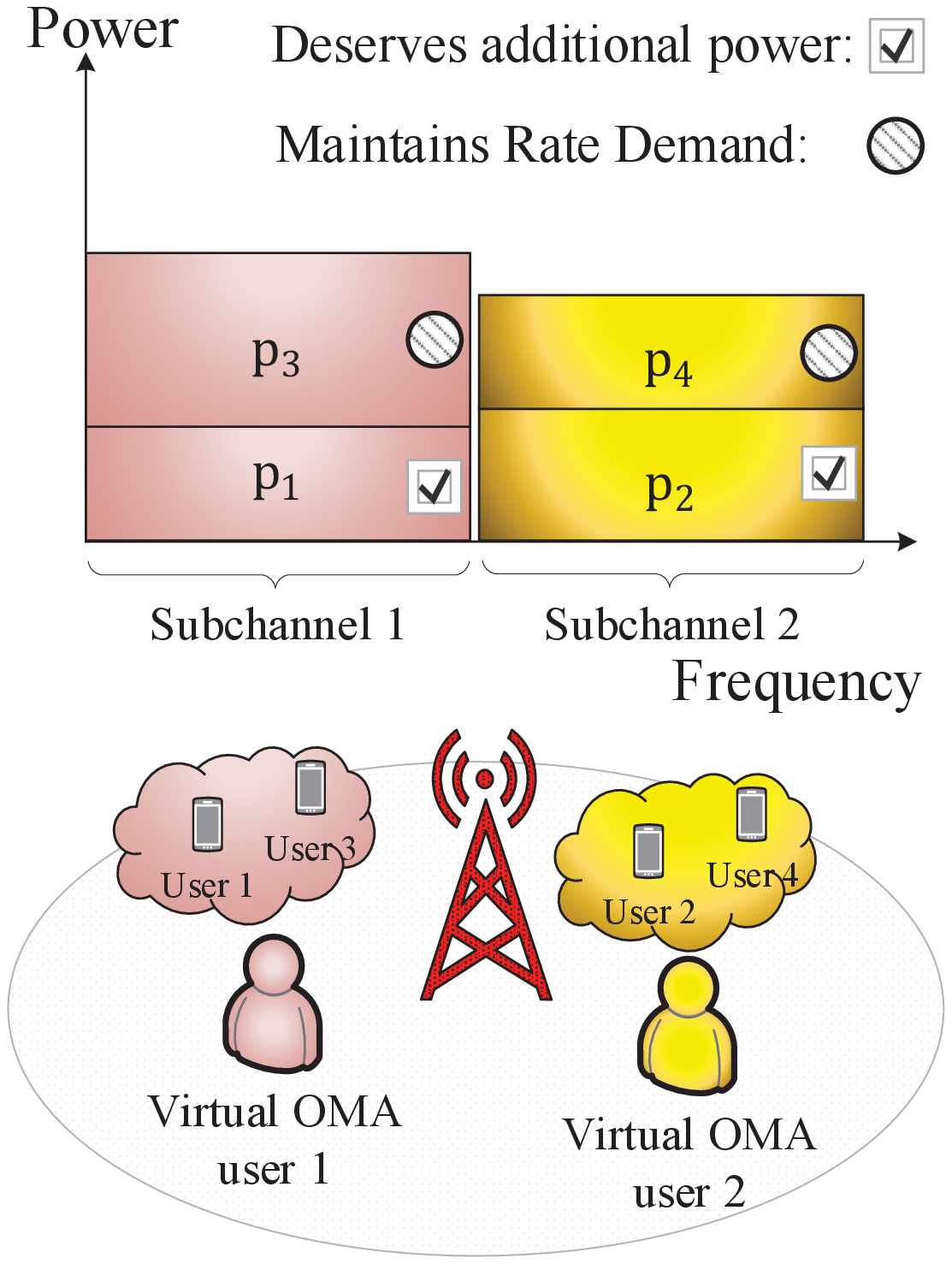}
		\label{Fig_FDMA-NOMA-Virtual}
	}\hfill
	\subfigure[FDMA: Real OMA users.]{
		\includegraphics[scale=0.33]{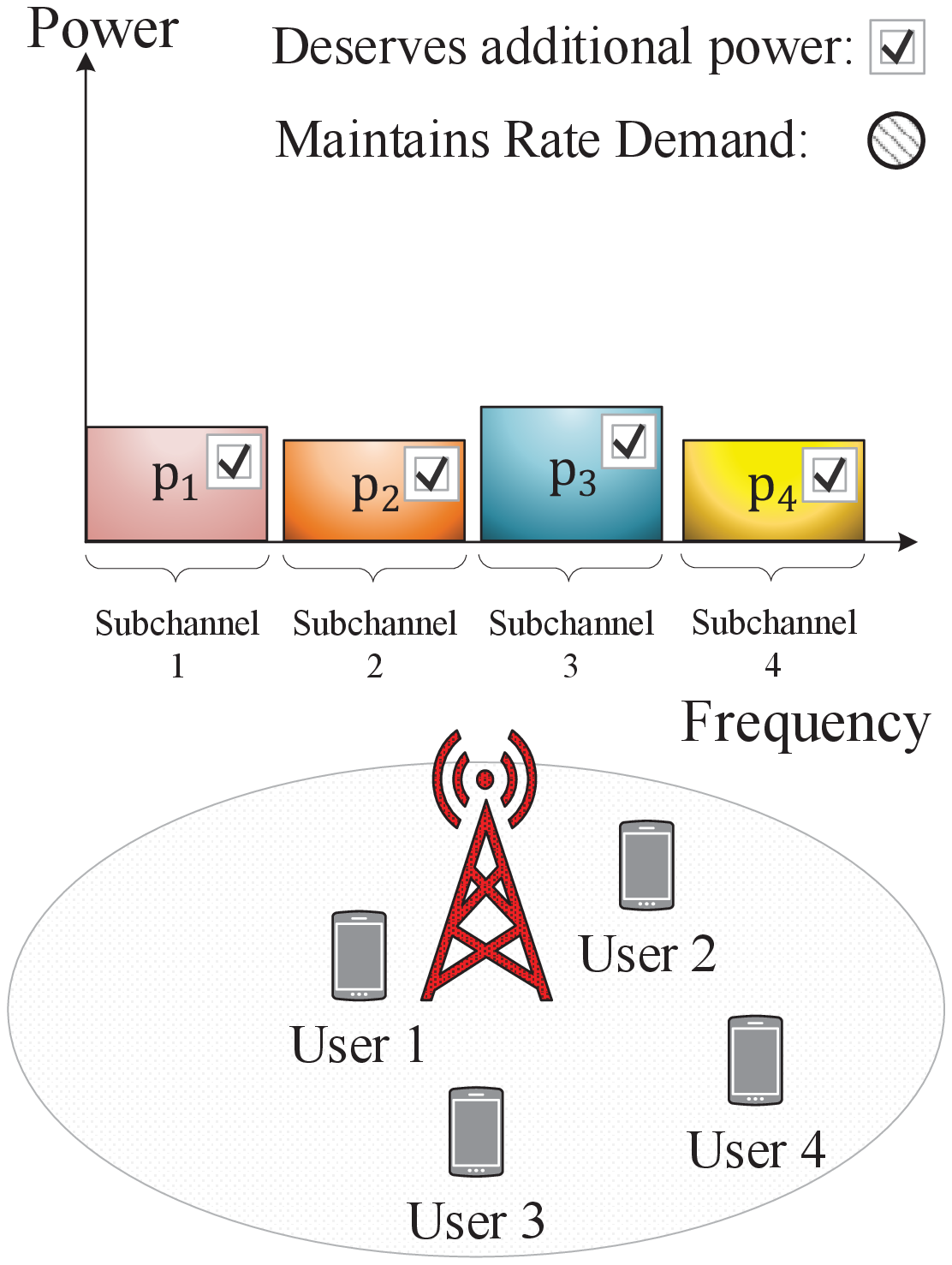}
		\label{Fig_FDMA-Virtual}
	}\hfill
	\subfigure[Hybrid-NOMA with $K$ multiplexed users: FDMA with $6$ virtual OMA users.]{
		\includegraphics[scale=0.33]{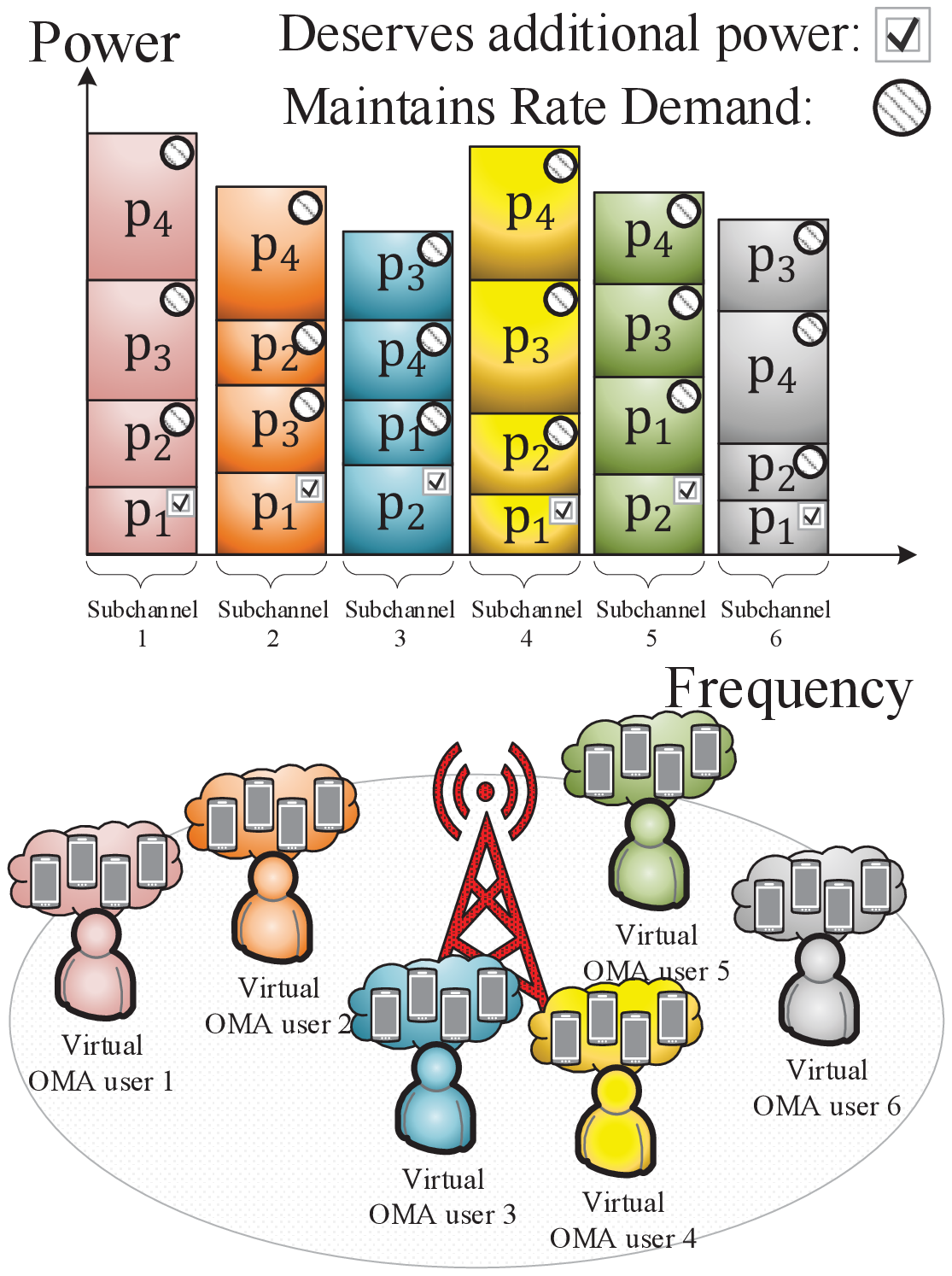}
		\label{Fig_Hybrid-NOMA-all-Virtual}
	}\hfill
	\subfigure[Hybrid-NOMA with $2$ multiplexed users: FDMA with $6$ virtual OMA users.]{
		\includegraphics[scale=0.33]{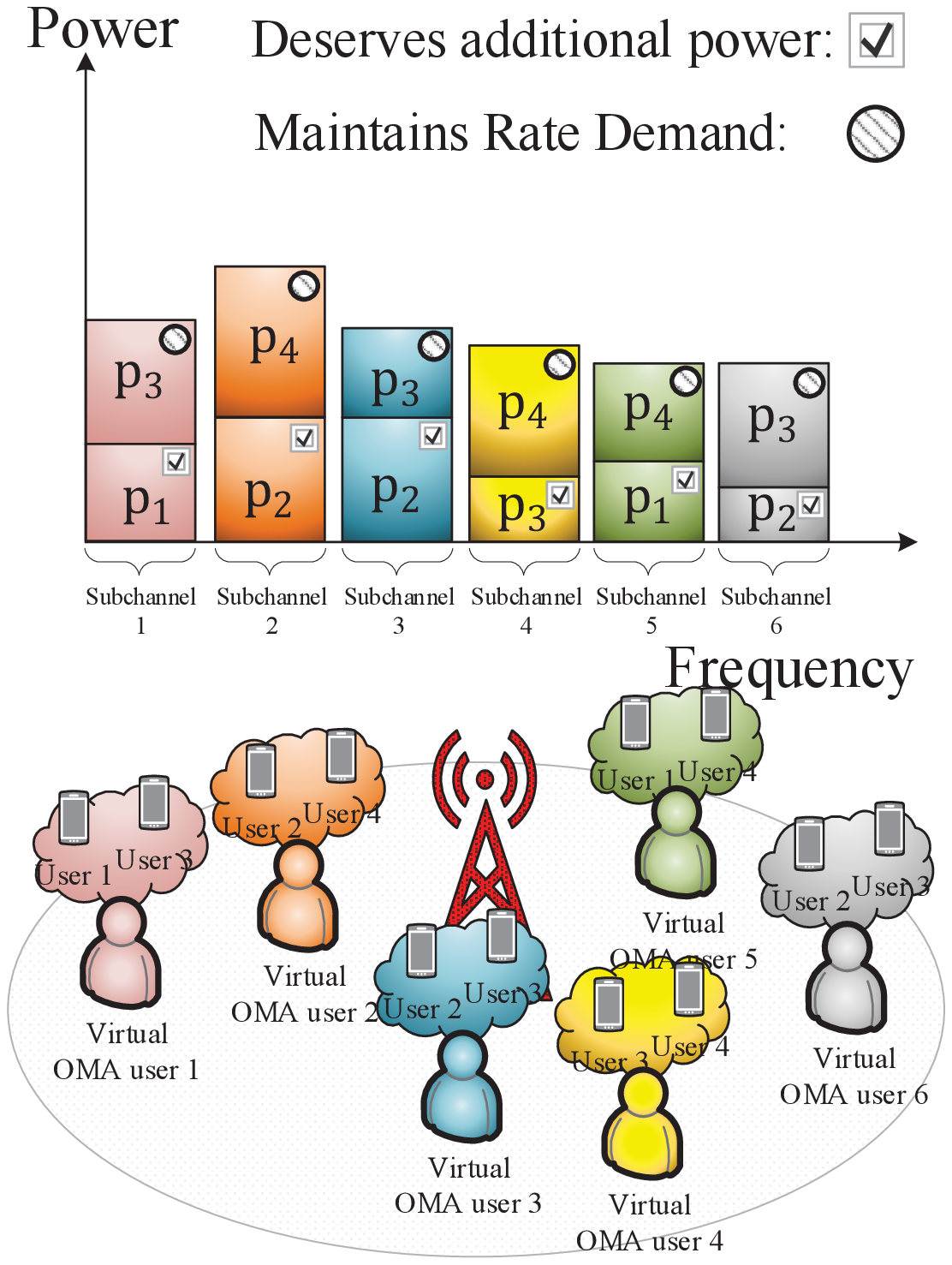}
		\label{Fig_Hybrid-NOMA-Virtual}
	}\hfill
	\subfigure[OFDMA: FDMA with $6$ virtual OMA users.]{
		\includegraphics[scale=0.33]{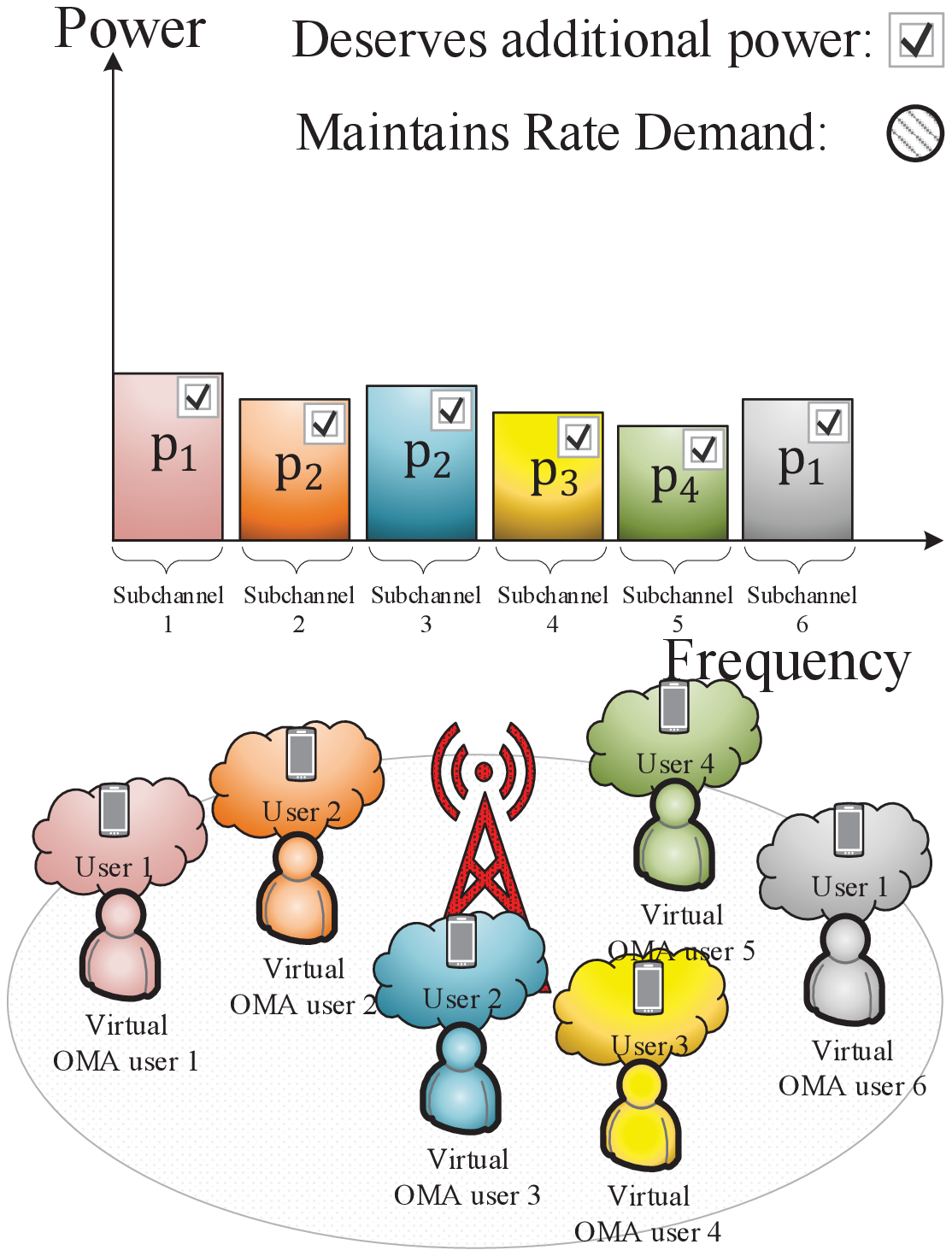}
		\label{Fig_OFDMA-Virtual}
	}\hfill
	\caption
	{The equivalent virtual FDMA models of Fig. \ref{Fig Systemmodel} including virtual OMA users (see Remark \ref{remark virtualuser}).}
	\label{Fig virtual system}
\end{figure*}
Note that FDMA/OFDMA is a special case of FDMA-NOMA/Hybrid-NOMA, where each subchannel is assigned to a single user. Hence, each OMA user acts as a cluster-head user, and subsequently, the virtual users are identical to the real OMA users, i.e., $\alpha_n=1$, $H_n=h^n_{\Phi_n}$, and $c_n=0$, for each $n \in \mathcal{N}$. As a result, each user in FDMA/OFDMA deserves additional power. In summary, the analysis for finding the optimal power allocation to maximize SR/EE of Hybrid-NOMA with per-symbol minimum rate constraints and FDMA is quite similar, and the only differences are $\alpha_n$ and $c_n$.
\begin{remark}\label{remark deviation cn}
	Remark \ref{remark virtualuser} shows that when $c_n \to 0$, the difference term $\frac{c_n}{\alpha_n} \to 0$ in $\tilde{q}_n$, $\tilde{Q}^{\rm{min}}_n$, and $\tilde{P}^{\rm{mask}}_n$. Subsequently, when $c_n \to 0,~\forall n \in \mathcal{N}$, we have $\sum\limits_{n \in \mathcal{N}} \frac{c_n}{\alpha_n} \to 0$ in $\tilde{P}^{\rm{max}}$. Accordingly, when $c_n \to 0,~\forall n \in \mathcal{N}$, we guarantee that $\tilde{q}_n=q_n,~\forall n \in \mathcal{N}$, $\tilde{Q}^{\rm{min}}_n=Q^{\rm{min}}_n,~\forall n \in \mathcal{N}$, $\tilde{P}^{\rm{mask}}_n=P^{\rm{mask}}_n,~\forall n \in \mathcal{N}$, and $\tilde{P}^{\rm{max}} = P^{\rm{max}},~\forall n \in \mathcal{N}$. In other words, when $c_n \to 0,~\forall n \in \mathcal{N}$, the network parameters of Hybrid-NOMA will be exactly the same as its virtual FDMA system.
\end{remark}
In each cluster $n$, the term $c^n_{k}$ tends to zero when $h^n_{k} \to \infty,~ k \in \mathcal{K}_n$. The numerical results verify that in most of the channel realizations, specifically high CNR regions, $c^n_{k} \approx 0,~\forall n \in \mathcal{N}, k \in \mathcal{K}_n$ \cite{9583874}. With assuming $c^n_{k} \approx 0,~\forall n \in \mathcal{N}, k \in \mathcal{K}_n$, we have $c_n \approx 0, \forall n \in \mathcal{N}$ in \eqref{opt IC M qmin}. Hence, the results in Remark \ref{remark deviation cn} are valid for the high CNR regions.

When $c^n_{k} \approx 0,~\forall n \in \mathcal{N}, k \in \mathcal{K}_n$, the optimal intra-cluster powers in \eqref{opt IC i} and \eqref{opt IC M} can be approximated, respectively, as 
\begin{equation}\label{opt IC i approx}
	{p^n_{k}}^*\approx\left(\beta^n_{k} \prod\limits_{j \in \mathcal{K}_n \atop h^n_{j} < h^n_{k}} \left(1-\beta^n_{j}\right)\right) q_n,~ \forall n \in \mathcal{N},~k \in \mathcal{K}_n\setminus\{\Phi_n\},
\end{equation}
and
\begin{equation}\label{opt IC M approx}
	{p^n_{\Phi_n}}^* \approx \left(1 - \sum\limits_{i \in \mathcal{K}_n \atop h^n_{i} < h^n_{\Phi_n}} \beta^n_{i} \prod\limits_{j \in \mathcal{K}_n \atop h^n_{j} < h^n_{i}} \left(1-\beta^n_{j}\right)\right) q_n,~ \forall n \in \mathcal{N}.
\end{equation}
For the case that the users in $\mathcal{K}_n$ have the same minimum rate demands $R^{\rm{min}}_{k,n}/W_s=R$ in bps/Hz, it is straightforward to show that \eqref{opt IC i approx} and \eqref{opt IC M approx} can be reformulated, respectively, by
\begin{equation}\label{power sameminrate i}
	{p^n_{k}}^*\approx\frac{2^{R}-1}{ \left(2^{R}\right)^{\Theta_k} } q_n,~ \forall n \in \mathcal{N},~k \in \mathcal{K}_n\setminus\{\Phi_n\},
\end{equation}
and 
\begin{equation}\label{power sameminrate M}
	{p^n_{\Phi_n}}^*\approx\frac{1}{ \left(2^{R}\right)^{|\mathcal{K}_n|-1} } q_n,~ \forall n \in \mathcal{N},
\end{equation}
where $\Theta_k = \left|\left\{i \in \mathcal{K}_n | h^n_{i} \leq h^n_{k} \right\}\right|$. 
\begin{corollary}\label{corollary Intra diverse}
	The approximated closed-form expressions \eqref{power sameminrate i} and \eqref{power sameminrate M} verify the high heterogeneity of optimal power coefficients among multiplexed users, thus the importance of finding optimal intra-cluster power allocation. For instance, the equal intra-cluster power allocation is infeasible in most of the cases, due to violating the minimum rate constraints in \eqref{Constraint minrate SCuser}.
\end{corollary}
For the special case $|\mathcal{K}_n|=2$, and $R=1$ bps/Hz, we have ${p^n_{1}}^* \approx {p^n_{2}}^* \approx \frac{1}{2} q_n$, meaning that the equal intra-cluster power allocation is nearly optimal.

The inter-cluster power allocation is necessary when $\sum\limits_{n \in \mathcal{N}} P^{\rm{mask}}_n > P^{\rm{max}}$, i.e., there is at least one cluster which is not allowed to operate at its maximum power $P^{\rm{mask}}_n$. In this case, the distributed inter-cluster power allocation leads to violating the cellular power constraint \eqref{Constraint cell power}, since in the distributed power allocation among clusters, constraint \eqref{Constraint mask power} will be active. Alternatively, when $\sum\limits_{n \in \mathcal{N}} P^{\rm{mask}}_n \leq P^{\rm{max}}$, we guarantee that $q^*_n=P^{\rm{mask}}_n,~\forall n \in \mathcal{N}$.
There are a number of works, e.g., \cite{7557079,8114362}, assuming $P^{\rm{mask}}_n=P^{\rm{max}}/N,~\forall n \in \mathcal{N}$, i.e., equal inter-cluster power allocation while maintaining the cellular power constraint \eqref{Constraint cell power}. In this case, $q_n=P^{\rm{max}}/N,~\forall n \in \mathcal{N}$, and the optimal intra-cluster power allocation can be obtained by using Proposition \ref{proposition jointcluster}.
In the following, we investigate the optimality condition for the equal inter-cluster power allocation. 
\begin{proposition}\label{proposition samechannels}
	When $c^n_{k} \approx 0,~\forall n \in \mathcal{N}, k \in \mathcal{K}_n$, the equal inter-cluster power allocation, i.e., $q_n=P^{\rm{max}}/N,~\forall n \in \mathcal{N}$, is optimal if and only if 1) $P^{\rm{max}}/N \in \left[Q^{\rm{min}}_n,P^{\rm{mask}}_n\right],~\forall n \in \mathcal{N}$; 2) $\frac{h^i_{\Phi_i}}{h^j_{\Phi_j}} = \frac{\alpha_j}{\alpha_i},~\forall i,j \in \mathcal{N}$.
\end{proposition}
\begin{proof}
	The equal inter-cluster power allocation should be feasible to problem \eqref{eq1 problem}. According to Proposition \ref{proposition feasiblecluster}, $q_n=P^{\rm{max}}/N,~\forall n \in \mathcal{N}$, is feasible if and only if $P^{\rm{max}}/N \in [Q^{\rm{min}}_n,P^{\rm{mask}}_n],~\forall n \in \mathcal{N}$.
	
	According to \eqref{bisection opt form}, two clusters/virtual OMA users $i,j \in \mathcal{N}$ get the same virtual powers, i.e., $\tilde{q}^*_i = \tilde{q}^*_j$, if and only if $H_i = H_j$. According to \eqref{opt IC M qmin}, for each cluster $n \in \mathcal{N}$, when $c^n_{k} \approx 0,~\forall k \in \mathcal{K}_n$, we have $c_n \approx 0$.
	According to Remark \ref{remark deviation cn}, when $c_n \approx 0,~\forall n \in \mathcal{N}$, we guarantee that $\tilde{q}_n=q_n,~\forall n \in \mathcal{N}$. As a result, for two clusters $i,j \in \mathcal{N}$, we have $q^*_i = q^*_j$, if and only if $H_i = H_j$. By using $H_n=\alpha_n h^n_{\Phi_n},~\forall n \in \mathcal{N}$, defined in \eqref{eq1 problem}, $q^*_i = q^*_j$, if and only if $\frac{h^i_{\Phi_i}}{h^j_{\Phi_j}} = \frac{\alpha_j}{\alpha_i}$. Hence, $\tilde{q}^*_i = \tilde{q}^*_j,~\forall i,j \in \mathcal{N}$, with $\sum\limits_{n \in \mathcal{N}} q^*_n = P^{\rm{max}}$, or equivalently $q^*_n=P^{\rm{max}}/N,~\forall n \in \mathcal{N}$, if and only if $\frac{h^i_{\Phi_i}}{h^j_{\Phi_j}} = \frac{\alpha_j}{\alpha_i},~\forall i,j \in \mathcal{N}$, and the proof is completed.
\end{proof}
According to Proposition \ref{proposition samechannels}, in Hybrid-NOMA, when $c_n \approx 0$, the equal inter-cluster power allocation is optimal if and only if all the virtual OMA users have exactly the same CNRs. These results also hold for FDMA, where $\alpha_n=1,~\forall n \in \mathcal{N}$, $H_n=h^n_{\Phi_n},~\forall n \in \mathcal{N}$, and $c_n=0,~\forall n \in \mathcal{N}$. According to Remark \ref{remark virtualuser} and Proposition \ref{proposition samechannels}, the unique condition $\frac{h^i_{\Phi_i}}{h^j_{\Phi_j}} = \frac{\alpha_j}{\alpha_i},~\forall i,j \in \mathcal{N}$, for the optimality of the equal inter-cluster power allocation states that when the cluster-head users have exactly the same CNRs, i.e., $h^i_{\Phi_i}=h^j_{\Phi_j},~\forall i,j \in \mathcal{N}$, the equal inter-cluster power allocation strategy is optimal if and only if $\alpha_i=\alpha_j,~\forall i,j \in \mathcal{N}$. According to the definition of $\alpha_n$ in \eqref{opt IC M qmin}, one simple case that $\alpha_i \neq \alpha_j$ for some $i,j \in \mathcal{N}$ is considering different minimum rate demands for the users with lower decoding order.
\begin{corollary}
	In contrast to FDMA, the optimality condition of the equal inter-cluster power allocation strategy depends on the individual minimum rate demand of users with lower decoding order. This power allocation strategy can be suboptimal for Hybrid-NOMA even if the clusters have the same order and all the users in different clusters have the same CNRs. Moreover, the CNR of users with lower decoding order does not significantly affect the performance of the equal inter-cluster power allocation strategy.
\end{corollary}
For the case that Proposition \ref{proposition samechannels} holds, i.e., $H_i = H_j,~\forall i,j \in \mathcal{N}$ thus $q^*_n=P^{\rm{max}}/N,~\forall n \in \mathcal{N}$, the optimal $\nu^*$ in \eqref{bisection opt form} can be obtained based on the quality $P^{\rm{max}}/N=\frac{W_s/(\ln 2)}{\nu^*} - \frac{1}{H_n}$. Hence, we have 
$\nu^* = \frac{\frac{P^{\rm{max}}}{N} + \frac{1}{H_n}}{W_s/(\ln 2)}$.
In general, for the case that $H_n,~\forall n \in \mathcal{N}$, is significantly large, i.e., high CNR regions of virtual OMA users, the second term $\frac{1}{H_n}$ in \eqref{bisection opt form} tends to zero. In this case, we observe a low heterogeneity of inter-cluster power allocation among clusters, resulting in near-optimal performance for the equal inter-cluster power allocation strategy.

\subsubsection{EE Maximization}
Based on Proposition \ref{proposition intraEE}, we observe that the closed-form expressions of optimal intra-cluster power allocation are also valid for the EE maximization problem \ref{EE problem}. Hence, Remark \ref{remark virtualuser} and Fig. \ref{Fig virtual system} are also valid for the EE maximization problem. Besides, Proposition \ref{proposition bisectionEE} provides a sufficient condition during each Dinkelbach iteration in which the full cellular power consumption not only leads to the maximum SR, but also maximum EE. In other words, if \eqref{EE fullbudget} holds, the full cellular power consumption is energy efficient. The term $\frac{W_s/(\ln 2)}{\lambda} - \frac{1 - c_n h^n_{\Phi_n}}{\alpha_n h^n_{\Phi_n}}$ in \eqref{EE fullbudget} is increasing in $\lambda=\frac{\sum\limits_{n \in \mathcal{N}} \sum\limits_{k \in \mathcal{K}_n} R^n_{k} (\boldsymbol{p}^n)}{\sum\limits_{n \in \mathcal{N}} \sum\limits_{k \in \mathcal{K}_n} p^n_{k} + P_{\rm{C}}}$. The fractional parameter $\lambda$ is a decreasing function of $P_{\rm{C}}$. As a result, increasing $P_{\rm{C}}$ increases the term $\frac{W_s/(\ln 2)}{\lambda} - \frac{1 - c_n h^n_{\Phi_n}}{\alpha_n h^n_{\Phi_n}}$ in \eqref{EE fullbudget}. In other words, \eqref{EE fullbudget} holds when the circuit power consumption of the BS is significantly large. 
\begin{corollary}\label{corollary EEdeservpow}
	In both the SR and EE maximization problems of Hybrid-NOMA with per-symbol minimum rate constraints, in each cluster, only the cluster-head user deserves additional power, and all the other users get power to only maintain their minimal rate demands. Our analysis proves that in the SR maximization problem, the BS operates at its maximum power budget. However, for the EE maximization problem, the BS may operate at lower power depending on the condition in Proposition \ref{proposition bisectionEE}.
\end{corollary}

\subsection{Computational Complexity Analysis}
In this subsection, we discuss about the computational complexity order of our proposed Algs. \ref{Alg bisection}-\ref{Alg mixedWFBarrier}. To simplify the complexity analysis, we assume that $|\mathcal{K}_n|=K,~\forall n \in \mathcal{N}$, in this subsection.

Alg. \ref{Alg bisection} belongs to the family of water-filling solutions which is comprehensively discussed in the literature \cite{1381759,1576943,6547819,10.1155/2008/643081,8995606}.
The water-filling algorithms are mainly divided into two categorizes: 1) iterative algorithms, like bisection method, which stops until the error is below some tolerance threshold; 2) Exact algorithms based on hypothesis testing \cite{1381759}. It is difficult to obtain the exact complexity of the bisection method to achieve an $\epsilon$-suboptimal performance, however we numerically observed that the error will be less than $10^{-6}$ mostly within $20$ iterations. The exact algorithms have an exponential worst-case complexity on the order of $2^N$, however it is possible to obtain a linear worst-case complexity of $N$ \cite{1381759,6547819}. This linear complexity can be achieved by properly sorting the so-called sequences which is comprehensively discussed in \cite{1381759,6547819}.
Generally speaking, the number of water-filling iterations increases linearly with the number of subchannels $N$ \cite{1381759,6547819}. In each iteration, we obtain $\tilde{q}^*_n,~\forall n \in \mathcal{N},$ by using \eqref{bisection opt form}, which needs $N$ operations. Therefore, the complexity of Alg. \ref{Alg bisection} is on the order of $N^2$. Note that the complexity of Alg. \ref{Alg bisection} is approximately independent of the number of multiplexed users $|\mathcal{K}_n|,~\forall n \in \mathcal{N}$. This is due to the equivalent transformation of the Hybrid-NOMA problem \eqref{sumrate problem} to its corresponding virtual FDMA problem \eqref{eq1 problem}.
Increasing the number of multiplexed users $|\mathcal{K}_n|$ only increases the complexity of calculating $(Q^{\rm{min}}_n,\alpha_n)$ in the initialization step of Alg. \ref{Alg bisection} which is negligible.
\begin{table*}[tp]
	\caption{Computational Complexity of Solving the Sum-Rate Maximization Problem \eqref{sumrate problem}.}
	\begin{center} \label{table complexity SR}
		\scalebox{0.76}{\begin{tabular}{|c|c|c|c|c|c|c|c|}
				\hline \rowcolor[gray]{0.910}
				\textbf{NOMA-to-OMA Transformation} & \textbf{Alg. \ref{Alg bisection}} & \textbf{Subgradient Method} & \textbf{Barrier Method} \\
				\hline
				\rowcolor[gray]{0.960}
				\textbf{Complexity} & $N^2$ & $C_S(1+N)$ & $\left\lceil\frac{\ln\left(\frac{1}{\epsilon_B.\mu}\right)}{\ln (\mu)} \right\rceil C_N$ \\
				\hline \rowcolor[gray]{0.910}
				\textbf{Pure Methods} & \textbf{Water-Filling Algorithm \cite{8125713}} & \textbf{Subgradient Method} & \textbf{Barrier Method} \\
				\hline
				\rowcolor[gray]{0.960}
				\textbf{Complexity} & $N^2 K$ & $C_S(1+N+2KN)$ & $\left\lceil\frac{\ln\left(\frac{1+N+KN}{\epsilon_B.\mu}\right)}{\ln (\mu)} \right\rceil C_N$ \\
				\hline
		\end{tabular}}
	\end{center}
\end{table*}
\begin{table*}[tp]
	\caption{Computational Complexity of Solving the Energy Efficiency Maximization Problem \eqref{EE problem}.}
	\begin{center} \label{table complexity EE}
		\scalebox{0.76}{\begin{tabular}{|c|c|c|c|c|c|c|c|}
				\hline \rowcolor[gray]{0.910}
				\textbf{NOMA-to-OMA Transformation} & \textbf{Alg. \ref{Alg dinkelbach}-inner Alg. \ref{Alg bisection}} & \textbf{Alg. \ref{Alg dinkelbach}-inner Alg. \ref{Alg mixedWFSub}} & \textbf{Alg. \ref{Alg dinkelbach}-inner Alg. \ref{Alg mixedWFBarrier}} & \textbf{Exploring $P_{\rm{EE}}$-inner Alg. \ref{Alg bisection}} \\
				\hline
				\rowcolor[gray]{0.960}
				\textbf{Complexity} & $N^2$ & $C_F C_S(1+N)$ & $C_F\left\lceil\frac{\ln\left(\frac{1}{\epsilon_B.\mu}\right)}{\ln (\mu)} \right\rceil C_N$ & $\left\lceil\frac{P^{\rm{max}} - \sum\limits_{n \in \mathcal{N}} Q^{\rm{min}}_n}{\delta}\right\rceil N^2$ \\
				\hline \rowcolor[gray]{0.910}
				\textbf{Pure Methods} & \textbf{Alg. \ref{Alg dinkelbach}-inner Water-Filling \cite{8125713}} & \textbf{Alg. \ref{Alg dinkelbach}-inner Subgradient Method \cite{8119791,8448840,9032198}} & \textbf{Alg. \ref{Alg dinkelbach}-inner Barrier Method} & \textbf{Exploring $P_{\rm{EE}}$-inner Water-Filling \cite{8125713}} \\
				\hline
				\rowcolor[gray]{0.960}
				\textbf{Complexity} & $N^2 K$ & $C_F C_S(1+N+2KN)$ & $C_F \left\lceil \frac{\ln \left(\frac{1+N+KN}{\epsilon_B.\mu}\right)}{\ln (\mu)} \right\rceil C_N$ & $\left\lceil\frac{P^{\rm{max}}}{\delta}\right\rceil N^2 K$ \\
				\hline
		\end{tabular}}
	\end{center}
\end{table*}

Alg. \ref{Alg dinkelbach} which is based on the Dinkelbach method converts the original problem \eqref{EE problem} into a sequence of auxiliary problems, indexed by $\lambda$. The overall complexity of Alg. \ref{Alg dinkelbach} mainly depends on both the convergence rate of the subproblems, as well as the computational complexity of each subproblem. By defining $E(\boldsymbol{p})=\frac{f_1(\boldsymbol{p})}{f_2(\boldsymbol{p})}$, 
where $f_1(\boldsymbol{p})=\sum\limits_{n \in \mathcal{N}} \sum\limits_{k \in \mathcal{K}_n} R^n_{k} (\boldsymbol{p}^n)$, and $f_2(\boldsymbol{p})=\sum\limits_{n \in \mathcal{N}} \sum\limits_{k \in \mathcal{K}_n} p^n_{k} + P_{\rm{C}}$, the convergence rate of Alg. \ref{Alg dinkelbach} can be observed by formulating the update rule of the fractional parameter $\lambda$ as
$
\lambda_{(t+1)} = \frac{f_1\left(\boldsymbol{p}^*_{(t)}\right)}{f_2\left(\boldsymbol{p}^*_{(t)}\right)} = \lambda_{(t)} - \frac{f_1\left(\boldsymbol{p}^*_{(t)}\right) - \lambda_{(t)} f_2\left(\boldsymbol{p}^*_{(t)}\right)}{-f_2\left(\boldsymbol{p}^*_{(t)}\right)} = \lambda_{(t)} - \frac{F(\lambda_{(t)})}{F'(\lambda_{(t)})},
$
where $t$ is the iteration index of Alg. \ref{Alg dinkelbach}, and $F(\lambda_{(t)}) = f_1\left(\boldsymbol{p}^*_{(t)}\right) - \lambda_{(t)} f_2\left(\boldsymbol{p}^*_{(t)}\right)$ \cite{Jorswieckfractional}. It can be observed that Alg. \ref{Alg dinkelbach} follows the Newton's method, meaning that the Newton’s method is applied to the concave function $F(\lambda)$. Thus, Alg. \ref{Alg dinkelbach} exhibits a super-linear convergence rate \cite{Jorswieckfractional}. A detailed complexity analysis of the pure Newton's method can be found in Subsection 9.5.3 in \cite{Boydconvex}.
For a general concave function $F(\boldsymbol{x}),~\boldsymbol{x} \in \mathbb{R}^n$, if $F$ increases by at least $\Delta_F$ at each Newton's iteration, $\nabla^2F(\boldsymbol{x}) \leq -m$, and 
$\norm{\nabla^2 F(\boldsymbol{x}) - \nabla^2 F(\boldsymbol{y})}_2 \leq L \norm{\boldsymbol{x} - \boldsymbol{y}}_2,~\forall \boldsymbol{x},\boldsymbol{y} \in \mathbb{R}^n$, 
the number of Newton's iterations to achieve an $\epsilon$-suboptimal solution is bounded above by $C_F=\frac{F(\boldsymbol{x}^*) - F(\boldsymbol{x}_{(0)})}{\Delta_F} + \log_2 \log_2 (\epsilon_0/\epsilon)$, where $\epsilon_0=2m^3/L^2$ \cite{Boydconvex}. For the accuracy around $\epsilon \approx 5.10^{-20} \epsilon_0$, we have $\log_2 \log_2 (\epsilon_0/\epsilon) \approx 6$ \cite{Boydconvex}, thus in this case, the number of Newton's iterations is bounded above by $C_F\approx \frac{F(\boldsymbol{x}^*) - F(\boldsymbol{x}_{(0)})}{\Delta_F} + 6$.

In each iteration of Alg. \ref{Alg dinkelbach}, if Proposition \ref{proposition bisectionEE} holds, we solve \eqref{EEfrac problem} for the given $\lambda$ by using the water-filling Alg. \ref{Alg bisection}, whose overall complexity is $N^2$. 
For the case that Proposition \ref{proposition bisectionEE} does not hold in each Dinkelbach iteration, we solve \eqref{EEfrac problem} for the given $\lambda$ by using the subgradient or barrier methods presented in Algs. \ref{Alg mixedWFSub} and \ref{Alg mixedWFBarrier}, respectively.
The duality gap of the barrier method in Alg. \ref{Alg mixedWFBarrier} after $L$ iterations is $1/(\mu^L t^{0})$, where $t^{0}$ is the initial $t$, and $\mu$ is the stepsize for updating $t$ in the barrier method. Therefore, after exactly $\lceil \frac{\ln \left(\frac{1}{\epsilon_B.\mu}\right)}{\ln (\mu)} \rceil$ barrier iterations, Alg. \ref{Alg mixedWFBarrier} achieves $\epsilon_B$-suboptimal solution \cite{Boydconvex}. In each barrier iteration, we apply the Newton's method. In general, it is difficult to obtain the exact complexity order of the pure Newton's method \cite{Boydconvex}. According to Subsection 11.5.3 in \cite{Boydconvex}, when the self-concordance assumption holds, the total number of Newton's iterations over all the barrier iterations to achieve an $\epsilon_B$-suboptimal solution is bounded above by $\lceil \frac{\ln \left(\frac{1}{\epsilon_B.\mu}\right)}{\ln (\mu)} \rceil \left( \frac{\mu -1 - \ln \mu}{\Delta_F} + \log_2 \log_2 (1/\epsilon_N) \right)$, where $\epsilon_N$ is the tolerance of the Newton's method in each barrier iteration. The complexity of other operations in the centering step of each barrier iteration is negligible. As a result, when Proposition \ref{proposition bisectionEE} does not hold, the overall worst-case complexity of Alg. \ref{Alg dinkelbach} with inner Alg. \ref{Alg mixedWFBarrier} is approximately on the order of $C_F \left(\lceil \frac{\ln \left(\frac{1}{\epsilon_B.\mu}\right)}{\ln (\mu)} \rceil \left( \frac{\mu -1 - \ln \mu}{\Delta_F} + \log_2 \log_2 (1/\epsilon_N) \right)\right)$, where $C_F$ denotes the number of Dinkelbach iterations in Alg. \ref{Alg dinkelbach}. 

The standard subgradient method produces a global optimum, however its exact computational complexity is still unknown in general \cite{MinimizationMethods,boydsubgradient}. It is shown that the subgradient method converges with polynomial complexity in the number of optimization variables and constraints \cite{MinimizationMethods,boydsubgradient}. In each subgradient iteration of Alg. \ref{Alg mixedWFSub}, we need to calculate $\boldsymbol{\tilde{q}}^{(t)}$ in Step 8 which requires $N$ operations. Then, we update the Lagrange multiplier $\nu$ whose complexity order is $1$. Thus, the overall complexity of Alg. \ref{Alg dinkelbach} with inner Alg. \ref{Alg mixedWFSub} is $C_F C_S(N+1)$, where $C_S$ indicates the number of subgradient iterations. 

The computational complexity order of our proposed as well as other existing globally optimal power allocation algorithms for solving the SR and EE maximization problems is summarized in Tables \ref{table complexity SR} and \ref{table complexity EE}, respectively.
In these tables, the term "pure" is referred to the case that we do not apply Propositions \ref{proposition feasiblecluster} and \ref{proposition jointcluster} (thus the equivalent transformation of Hybrid-NOMA to a FDMA system, denoted by "NOMA-to-OMA Transformation") in the convex solvers. The parameters $C_F$ and $C_S$ denote the number of Dinkelbach and subgradient iterations, respectively. Moreover, $C_N$ denotes the number of Newton's iteration in each barrier iteration. The parameter $\delta$ in Table \ref{table complexity EE} indicates the stepsize of exhaustive search for finding $P_{\rm{EE}}$.
In Table \ref{table complexity SR}, the pure water-filling algorithm needs to update $p^n_k,~\forall n \in \mathcal{N},~k \in \mathcal{K}_n$, which requires $NK$ operations\footnote{The pure water-filling algorithm in \cite{8125713} is for uplink MC-NOMA without considering users minimum rate constraints.} \cite{8125713}. Hence, the overall complexity of the pure water-filling algorithm is on the order of $N^2 K$. Therefore, Alg. \ref{Alg bisection} reduces the complexity of the pure water-filling algorithm by $K$ times, where $K$ is the number of multiplexed users in each subchannel. It is also possible to solve problem \eqref{sumrate problem} or its equivalent FDMA problem \eqref{eq1 problem} by using the subgradient or barrier (with inner Newton's algorithm) methods. As can be seen, the equivalent NOMA-to-OMA transformation also reduces the complexity of these solvers. Besides, Alg. \ref{Alg bisection} has the lowest computational complexity compared to the other existing methods. The latter conclusions also hold for the EE maximization problem shown in Table \ref{table complexity EE}. When Proposition \ref{proposition samechannels} holds, we can use Alg. \ref{Alg bisection} with the lowest computational complexity compared to the other existing convex solvers. The superiority of the Dinkelbach algorithm can be observed by comparing it with a greedy search over all the possible power consumption of the BS, denoted by $P_{\rm{EE}}$. Although Proposition \ref{proposition feasiblecluster} can reduce the search area, such that we can obtain the lower-bound of $P_{\rm{EE}}$ as $\sum\limits_{n \in \mathcal{N}} Q^{\rm{min}}_n$ (see \eqref{Qmin func}), as well as reduce the complexity of the pure water-filling algorithm by using Proposition \ref{proposition jointcluster}, the overall complexity of exploring $P_{\rm{EE}}$ is still large, when the stepsize $\delta$ is significantly small.

The numerical experiments show that Alg. \ref{Alg dinkelbach} converges in less than $6$ iterations, meaning that $C_F \approx 6$. In each Dinkelbach iteration, the subgradient method in Alg. \ref{Alg mixedWFSub} converges within $C_S \approx 15$ iterations. Besides, Alg. \ref{Alg mixedWFBarrier} converges within $10$ barrier iterations. For significantly large number of users around $100$ to $200$, the simulation codes in \cite{sourcecode} verify that the convergence time of our proposed algorithms are on the order of milliseconds. Based on our numerical experiments, we observed that the convergence time of the subgradient method in Alg. \ref{Alg mixedWFSub} is less than that of the barrier method in Alg. \ref{Alg mixedWFBarrier}.

\subsection{Subchannel Allocation in MC-NOMA}
The optimal subchannel allocation problem, i.e., finding optimal $\boldsymbol{\rho}=[\rho^n_{k}]$ or equivalently cluster sets $\mathcal{K}_n$, in MC-NOMA is classified as integer nonlinear programming problem. The subchannel allocation is determined on the top of power allocation. Therefore, the exact closed form of inter-cluster power allocation is required for solving the subchannel allocation problem. Although Alg. \ref{Alg bisection} approaches the globally optimal solution with a fast convergence speed, the exact value of $\nu^*$ and subsequently, closed-form of $\boldsymbol{q}^*$ is still unknown in general. A similar issue exists for the water-filling algorithms for the FDMA problems \cite{1381759,1576943,6547819,10.1155/2008/643081,8995606}. The Dinkelbach and subgradient methods also have similar issues, in which the exact value of optimal $\lambda$ and $\nu$ are unknown in general, respectively. The joint optimal user clustering and power allocation is known to be strongly NP-hard \cite{9044770,8422362,7587811}. 
Although the latter problem is strongly NP-hard, the optimal number of clusters or subchannels in FDMA-NOMA can be obtained as follows:
\begin{proposition}\label{proposition optgroup}
	In a $K$-user FDMA-NOMA system with limited number of multiplexed users $U^{\rm{max}}$, the optimal number of clusters is $N^*=\lceil K/U^{\rm{max}}\rceil$.
\end{proposition}
\begin{proof}
	Due to the degradation of SISO Gaussian BCs, it is proved that SC-NOMA is capacity achieving, meaning that the rate region of FDMA/TDMA is a subset of the rate region of SC-NOMA \cite{PHDdiss,10.5555/1146355,NIFbook,915665,1246013}. Hence, for the case that $K < U^{\rm{max}}$, the optimal user clustering is considering all the users in the same cluster, and apply SC-SIC among all the users, i.e., FDMA-NOMA turns into SC-NOMA. Now, consider $K=U^{\rm{max}}+C$, where $1 \leq C \leq U^{\rm{max}}$. In this case, SC-NOMA is infeasible, however, FDMA-NOMA divides $K$ users into two isolated clusters $\mathcal{K}_1$ and $\mathcal{K}_2$ satisfying $|\mathcal{K}_n|\leq U^{\rm{max}},~n=1,2$, due to the existing limitation on the number of multiplexed users. Each cluster set $\mathcal{K}_n,~n=1,2$ is a SISO Gaussian BC whose capacity region can be achieved by using SC-SIC. Hence, further dividing each user group $\mathcal{K}_n,~n=1,2$, based on FDMA/TDMA would result lower achievable rate. The latter result holds for any possible $2$ groups with $1 \leq C \leq U^{\rm{max}}$. Now, consider a general case $M U^{\rm{max}} + 1 \leq K \leq (M+1) U^{\rm{max}}$ with nonnegative integer $M$. In this case, the lowest possible number of isolated clusters is $M+1$. Further imposing FDMA/TDMA to any existing group would result in a suboptimal performance. Accordingly, the optimal number of clusters is exactly $\lceil K/U^{\rm{max}}\rceil$.
\end{proof}
Proposition \ref{proposition optgroup} shows that the achievable rate of FDMA with the highest isolation among users is a subset of the achievable rate of FDMA-NOMA with any given user clustering. Since our globally optimal power allocation algorithms are valid for any given user clustering, the existing suboptimal user clustering algorithms, such as heuristic methods in \cite{7557079,8114362,9276828,8119791,8448840}, and matching-based algorithms in \cite{7982784,7523951} can be applied. Another approach is the framework in \cite{6816086} which is the joint optimization of power and subchannel allocation with the relaxed-and-rounding method. However, the output is still suboptimal without any mathematical performance improvement guarantee. Roughly speaking, there is still no mathematical understanding analysis for performance comparison among the existing suboptimal user clustering algorithms. The optimal user clustering is still unknown, and is considered as a future work.

\section{Extensions and Future Research Directions}\label{section extensions}
Here, we discuss about the possible extensions of our analysis to more general scenarios. For each case, the potential challenges are discussed in details.
\subsection{Users Maximum Rate Constraint}
According to Propositions \ref{proposition jointcluster} and \ref{proposition intraEE}, we conclude that at the optimal point of the SR/EE maximization problems, only the cluster-head users get additional power. In practical systems, the achievable rate of users are also limited by a maximum value due to the discrete modulation and coding schemes \cite{10.5555/1146355,NIFbook}. In the SR/EE maximization of Hybrid-NOMA with significantly large number of subchannels and/or multiplexed users, it merely happens that a cluster-head user's rate within a subchannel exceeds the truncated Shannon's bound. This is due to the fact that 1) The clusters power budget will be typically low, on the order of few Watts, or even mWatts; 2) Mostly, a large portion of the clusters power budget will be allocated to the non-cluster-head users. For the sake of completeness, we discuss about the impact of considering per-subchannel maximum rate constraints in the SR/EE maximization problems.
To keep the generality, let us define $R^{\rm{max}}_{k,n}$ as the individual maximum allowable rate of user $k$ on subchannel $n$. The maximum rate constraint can thus be formulated as  
\begin{equation}\label{constraint max rate}
	R^n_{k} (\boldsymbol{p}^n) \leq R^{\rm{max}}_{k,n},~\forall n \in \mathcal{N}, k \in \mathcal{K}_n.
\end{equation}
By adding \eqref{constraint max rate} to the original SR maximization problem \eqref{sumrate problem}, constraints \eqref{Constraint minrate} and \eqref{constraint max rate} can be combined as
\begin{equation}\label{constraint maxmin rate}
	R^{\rm{min}}_{k,n} \leq R^n_{k} (\boldsymbol{p}^n) \leq R^{\rm{max}}_{k,n},~\forall n \in \mathcal{N}, k \in \mathcal{K}_n.
\end{equation}
Obviously, the minimum rate demands should be chosen such that $R^{\rm{min}}_{k,n} \leq R^{\rm{max}}_{k,n},~\forall n \in \mathcal{N}, k \in \mathcal{K}_n$, otherwise the feasible set of \eqref{constraint maxmin rate} will be empty. According to Proposition \ref{proposition feasiblecluster}, we can guarantee that at the optimal point of the total power minimization problem, $R^n_{k} (\boldsymbol{p^*}^n) = R^{\rm{min}}_{k,n} \leq R^{\rm{max}}_{k,n},~\forall n \in \mathcal{N}, k \in \mathcal{K}_n$, meaning that the maximum rate constraint \eqref{constraint max rate} has no impact on the lower-bound of $\boldsymbol{q}$. Thus, the lower-bound of $q_n$ can be obtained by \eqref{Qmin func}. On the other hand, the upper-bound of $q_n$ can be achieved by solving the per-cluster total power maximization problem (when the cellular power constraint is eliminated). Let us denote $Q^\text{max}_n$ as the power consumption of cluster $n$, where $R^n_{k} (\boldsymbol{p}^n) = R^{\rm{max}}_{k,n},~\forall k \in \mathcal{K}_n$. Similar to Appendix \ref{appendix Lemma feasible}, it can be easily shown that $Q^\text{max}_n$ can be obtained by \eqref{Qmin func} in which $\beta^n_k=2^{(R^{\rm{max}}_{k,n}/W_s)} -1,~\forall n \in \mathcal{N},~k \in \mathcal{K}_n$. In this way, the feasible set of problem \eqref{SCuser problem} with maximum rate constraints in \eqref{constraint max rate} can be characterized as the intersection of $q_n \in \left[Q^{\rm{min}}_n, \min\left\{Q^\text{max}_n,P^{\rm{mask}}_n\right\}\right],~\forall n \in \mathcal{N}$, and cellular power constraint $\sum\limits_{n \in \mathcal{N}} q_n \leq P^{\rm{max}}$. It is straightforward to show that when $\sum\limits_{n \in \mathcal{N}} \min\left\{Q^\text{max}_n,P^{\rm{mask}}_n\right\} \leq P^{\rm{max}}$, the cellular power constraint \eqref{Constraint cell power} will be always fulfilled, thus it can be removed from problems 
\eqref{sumrate problem} and \eqref{EE problem}. In this case, the SR maximization problem \eqref{sumrate problem} can be equivalently divided into $N$ SC-NOMA subproblems, since there is no longer the competition among clusters to get the cellular power budget. Subsequently, at the optimal point of the SR maximization problem, we guarantee that each cluster $n$ achieves its maximum allowable power, i.e., $q^*_n=\min\left\{Q^\text{max}_n,P^{\rm{mask}}_n\right\},~n \in \mathcal{N}$. Hence, the inter-cluster power allocation is required if and only if $\sum\limits_{n \in \mathcal{N}} \min\left\{Q^\text{max}_n,P^{\rm{mask}}_n\right\} > P^{\rm{max}}$.

Consider a simple $2$-user SC-NOMA system with $h_1 < h_2$, thus the optimal decoding order is $2 \to 1$.
The SR maximization problem \eqref{sumrate problem} with per-user maximum rate constraint can be formulated as follows:
\begin{subequations}\label{sumrate problem maxrate}
	\begin{align}\label{obf sumrate problem maxrate}
		\max_{ \boldsymbol{p} \geq 0}\hspace{.0 cm}	
		~~ & R_{1} (\boldsymbol{p}) + R_{2} (\boldsymbol{p})
		\\
		\text{s.t.}~~
		\label{Constraint minrate 2user}
		& R_{k} (\boldsymbol{p}) \geq R^{\rm{min}}_{k},~\forall k=1,2,
		\\
		\label{max rate constraint}
		& R_{k} (\boldsymbol{p}) \leq R^{\rm{max}}_{k},~\forall k=1,2,
		\\
		\label{Constraint cell power 2user}
		& p_{1} + p_{2} \leq \min\{P^{\rm{max}},Q^\text{max}\},
	\end{align}
\end{subequations}
where $Q^\text{max}$ is the maximum power consumption of the BS, due to constraint \eqref{max rate constraint}. The problem \eqref{sumrate problem maxrate} is convex with an affine feasible set. Assume that $R^{\rm{min}}_{k} < R^{\rm{max}}_{k},~k=1,2$. 
At the optimal point, Condition $\textbf{C1}:~p^*_{1} + p^*_{2} = \min\{P^{\rm{max}},Q^\text{max}\}$ always holds.
According to the proof of Proposition \ref{proposition jointcluster}, when \eqref{max rate constraint} for user $2$ is removed from \eqref{sumrate problem maxrate}, at the optimal point, the following properties hold
$$
\textbf{C2}:~p^*_{1} + p^*_{2} = P^{\rm{max}},~~~\textbf{C3}:~R_{1} (\boldsymbol{p}^*) = R^{\rm{min}}_{1}.
$$
In this case, based on Proposition \ref{proposition jointcluster}, the optimal powers can be obtained as
\begin{equation}\label{opt pow 2user}
	p^*_{1} = \beta_1 \left( P^{\rm{max}} + \frac{1}{h_{1}} \right),~~~p^*_{2} = \left( 1-\beta_1 \right) P^{\rm{max}} - \frac{\beta_1}{h_{1}},
\end{equation}
where $\beta_1=\frac{2^{(R^{\rm{min}}_{1}/W_s)}-1} {2^{(R^{\rm{min}}_{1}/W_s)}}$. Constraint \eqref{max rate constraint} for user $2$ can be rewritten as
\begin{equation}\label{maxrate 2}
	p_{2} \leq \left(2^{(R^{\rm{max}}_{2}/W_s)} - 1\right)/h_{2}.
\end{equation}
Hence, the maximum rate constraint of user $2$ is indeed a maximum power consumption constraint for this user. Let us define 
$$M_2=\left(2^{(R^{\rm{max}}_{2}/W_s)} - 1\right)/h_{2}.$$
According to Condition \textbf{C1}, \eqref{opt pow 2user} and \eqref{maxrate 2}, 
the optimal powers with imposing \eqref{max rate constraint} for both the users can be obtained as
\begin{align}\label{opt closedform 2user}
	p^*_{1} = & \min\left\{P^{\rm{max}},Q^\text{max}\right\} - \min\left\{ \left( 1-\beta_1 \right) P^{\rm{max}} - \frac{\beta_1}{h_{1}} , M_2 \right\},
	\nonumber \\
	& p^*_{2} = \min\left\{ \left( 1-\beta_1 \right) P^{\rm{max}} - \frac{\beta_1}{h_{1}} , M_2 \right\}.
\end{align}
Hence, if $\left( 1-\beta_1 \right) P^{\rm{max}} - \frac{\beta_1}{h_{1}} \leq M_2$, we guarantee that $R_{1} (\boldsymbol{p}^*)=R^{\rm{min}}_{1}$, $R^{\rm{min}}_{2} \leq R_{2} (\boldsymbol{p}^*) \leq R^{\rm{max}}_{2}$, and $p^*_{1} + p^*_{2} \leq P^{\rm{max}} < Q^\text{max}$. If $\left( 1-\beta_1 \right) P^{\rm{max}} - \frac{\beta_1}{h_{1}} > M_2$, and $P^{\rm{max}} \leq Q^\text{max}$, we guarantee that $R^{\rm{min}}_{1} \leq R_{1} (\boldsymbol{p}^*) \leq R^{\rm{max}}_{1}$, $R_{2} (\boldsymbol{p}^*) = R^{\rm{max}}_{2}$, and $p^*_{1} + p^*_{2} = P^{\rm{max}} \leq Q^\text{max}$. Finally, if $\left( 1-\beta_1 \right) P^{\rm{max}} - \frac{\beta_1}{h_{1}} > M_2$, and $P^{\rm{max}} > Q^\text{max}$, we guarantee that $R_{1} (\boldsymbol{p}^*) = R^{\rm{max}}_{1}$, $R_{2} (\boldsymbol{p}^*) = R^{\rm{max}}_{2}$, and $p^*_{1} + p^*_{2} = Q^\text{max} < P^{\rm{max}}$.
According to the above, Proposition \ref{proposition jointcluster} holds if and only if $\left( 1-\beta_1 \right) P^{\rm{max}} - \frac{\beta_1}{h_{1}} \leq M_2$. When user $2$ exceeds its maximum rate $R^{\rm{max}}_{2}$, we allocate power to user $2$ until $R_{2} (\boldsymbol{p}^*) = R^{\rm{max}}_{2}$, and the rest of the cellular power will be allocated to user $1$ until it achieves to its maximum rate.
The latter analysis can be generalized to the $K$-user SC-NOMA system. For more details, please see Appendix \ref{appendix max rate constraint}. The analysis in Appendix \ref{appendix max rate constraint} shows that there exists a closed-form of optimal power allocation for the general $K$-user SC-NOMA with per-user minimum and maximum rate constraints. During the power allocation, there exists a special user $i$, where all of the stronger users than user $i$ achieve their maximum rates, and all of the weaker users than user $i$ achieve their minimum rates. 
Due to the space limitations, obtaining the closed-form of optimal powers, and how to define the index of user $i$ for a given power budget is considered as a future work\footnote{In problem \eqref{sumrate problem} without maximum rate constraints, the cluster-head user, whose index is $\Phi_n$, is the special user $i$, thus none of the other multiplexed users deserve additional power. This is the main reason that we define the special notation $\Phi_n$ for the cluster-head user of subchannel $n$ in Subsection \ref{subsection netmodel}.}. After obtaining the closed-form of optimal powers as a function of the clusters power budget $\boldsymbol{q}$ in Hybrid-NOMA with per-subchannel maximum and minimum rate constraints, it might be possible to transform the Hybrid-NOMA problem to a FDMA problem, which can be considered as a future work.

\subsection{Hybrid-NOMA with Per-User Minimum Rate Constraints}
In our work, we considered a Hybrid-NOMA system, where the minimum rate demand of each user on each assigned subchannel is predefined, similar to \cite{7523951,9276828,8119791,8448840,9032198}. This scheme is the generalized model of FDMA-NOMA considered in \cite{7982784,7557079,8114362}. From the optimization perspective, the SR/EE maximization problem for Hybrid-NOMA with predefined minimum rate demand of each user on each assigned subchannel, and FDMA-NOMA has similar structures, and both of them are convex. A more general/complicated case is when we consider a per-user minimum rate constraint over all the assigned subchannels. 
The SR maximization problem for Hybrid NOMA with per-user minimum rate constraint can be formulated as 
\begin{subequations}\label{sumrateHybrid problem}
	\begin{align}\label{obf sumrateHybrid problem}
		\max_{ \boldsymbol{p} \geq 0}\hspace{.0 cm}	
		~~ & \sum\limits_{n \in \mathcal{N}} \sum\limits_{k \in \mathcal{K}_n} R^n_{k} (\boldsymbol{p}^n)
		\\
		\text{s.t.}~~
		& \eqref{Constraint cell power},~\eqref{Constraint mask power}, \nonumber
		\\
		\label{minrate Hybrid}
		& R_{k} (\boldsymbol{p}) \geq R^{\rm{min}}_{k},~\forall k \in \mathcal{K},
	\end{align}
\end{subequations}
where $R_{k} (\boldsymbol{p})= \sum\limits_{n \in \mathcal{N}_k} R^n_{k} (\boldsymbol{p}^n)$ denotes the achievable rate of user $k$ over all the assigned subchannels in $\mathcal{N}_k$. The term $R^n_{k} (\boldsymbol{p}^n)$ for each user $k \in \mathcal{K}_n \setminus \{\Phi_n\}$ is nonconcave in $\boldsymbol{p}^n$, due to the co-channel interference term $\sum\limits_{j \in \mathcal{K}_n, \atop h^n_{j} > h^n_{k}} p^n_{j} h^n_{k}$. Since each two terms $R^i_{k} (\boldsymbol{p}^i)$ and $R^j_{k} (\boldsymbol{p}^j)$ for subchannels $i,j \in \mathcal{N}_k$ includes disjoint set of powers, we can conclude that $R_{k} (\boldsymbol{p})= \sum\limits_{n \in \mathcal{N}_k} R^n_{k} (\boldsymbol{p}^n)$ is nonconcave when $|\mathcal{N}_k|>1$ and $\exists n \in \mathcal{N}_k,~k \neq \Phi_n$, which makes \eqref{minrate Hybrid} nonconcave. It is still unknown how to equivalently transform \eqref{minrate Hybrid} to a convex form. To this end, the globally optimal solution of \eqref{sumrateHybrid problem} with polynomial time complexity is not yet obtained in the literature. One suboptimal solution for \eqref{sumrateHybrid problem} is to approximate each nonconcave rate function $R^n_{k} (\boldsymbol{p}^n)$ to its first order Taylor series, and then apply the sequential programming method \cite{9264208,7954630,8758862}. A suboptimal penalty function method is also used in \cite{8807992}.
Let us define an auxiliary variable $r^n_k$ indicating the minimum rate demand of user $k \in \mathcal{K}_n$ on subchannel $n$ in bps. In this way, problem \eqref{sumrateHybrid problem} can be equivalently transformed to the following joint power and minimum rate allocation problem as
\begin{subequations}\label{sumrateHybrid problem1}
	\begin{align}\label{obf sumrateHybrid problem1}
		\max_{ \boldsymbol{p} \geq 0,~\boldsymbol{r} \geq 0}\hspace{.0 cm}	
		~~ & \sum\limits_{n \in \mathcal{N}} \sum\limits_{k \in \mathcal{K}_n} R^n_{k} (\boldsymbol{p}^n)
		\\
		\text{s.t.}~~
		& \eqref{Constraint cell power},~\eqref{Constraint mask power}, \nonumber
		\\
		\label{rate constraint}
		& R^n_{k} (\boldsymbol{p}^n) \geq r^n_k,~\forall n \in \mathcal{N}, k \in \mathcal{K}_n,
		\\
		\label{minrate epi}
		& \sum\limits_{n \in \mathcal{N}_k} r^n_k = R^{\rm{min}}_{k},~\forall k \in \mathcal{K},
	\end{align}
\end{subequations}
where $\boldsymbol{r}=[r^n_k],~\forall n \in \mathcal{N}, k \in \mathcal{K}_n$. For any given feasible $\boldsymbol{r}$ satisfying constraints in \eqref{minrate epi}, problem \eqref{sumrateHybrid problem1} can be equivalently transformed to the convex problem \eqref{sumrate problem} with minimum rate demands $r^n_k,~\forall n \in \mathcal{N}, k \in \mathcal{K}_n$.
Hence, our analysis and important theoretical insights hold for any given $\boldsymbol{r}$ in the SR/EE maximization problem of Hybrid-NOMA with per-user minimum rate constraints. According to the above, the only challenge which is not yet solved is how to find $\boldsymbol{r}^*$ in \eqref{sumrateHybrid problem1} or equivalently distribute $R^{\rm{min}}_{k}$ over the subchannels in $\mathcal{N}_k$.
\begin{corollary}\label{corol Hybrid}
	In Hybrid-NOMA with per-user minimum rate demands over all the assigned subchannels, if user $k \in \mathcal{K}$ is a non-cluster-head user in all the assigned subchannels, e.g., a cell-edge user, at the optimal point of SR/EE maximization, it gets power to only maintain its minimum rate demand $R^{\rm{min}}_{k}$, meaning that $R_{k} (\boldsymbol{p}^*)=\sum\limits_{n \in \mathcal{N}_k} R^n_{k} (\boldsymbol{p^*}^n)=R^{\rm{min}}_{k}$.
\end{corollary}
Accordingly, each user $k \in \mathcal{K}$ deserves additional power if and only if it is a cluster-head user in at least one of the assigned subchannels.
As a result, when the minimum rate demand of users are zero, in both the SR and EE maximization problems, only the cluster-head users get positive power, thus Hybrid-NOMA will be identical to OFDMA (also see Lemma 8 in \cite{7587811}). These results show that in both the SR and EE maximization problems of Hybrid-NOMA, there exists a critical fairness issue among users' achievable rate which is discussed in the following subsection.

\subsection{Users' Rate Fairness}
According to \eqref{power sameminrate i} and \eqref{power sameminrate M}, we observe that in the SR/EE maximization problems, a large portion of the clusters power budget will be allocated to the users with lower decoding order when all the multiplexed users have the same minimum rate demands within a cluster. It states that in contrast to FDMA, NOMA \textit{usually} allocates more power to the weaker users when all the multiplexed users have the same minimum rate demands. This result shows that NOMA provides users fairness in terms of power allocation. However, according to Corollaries \ref{corollary EEdeservpow} and \ref{corol Hybrid}, we observe that this users' power fairness does not necessarily lead to the users' rate fairness, since only one user in each cluster gets additional rate. Accordingly, substantial works are required to guarantee users' rate fairness. There exist many fairness schemes which are recently considered for SC/MC-NOMA, as proportional fairness \cite{9044770,8737495,7982784,7954630,7812683,7587811,7560605}, max-min fairness \cite{7982784}, and etc. In the following, we first discuss about the advantages/challenges of the proportional fairness scheme, where our objective is to tune the users achievable rate at the optimal point by maximizing the weighted SR of users. Then, we propose a new fairness scheme which is a mixture of proportional fairness and users weighted minimum rate demands.

\subsubsection{Proportional Fairness}\label{subsection pf}
In proportional fairness, we aim at maximizing the weighted SR of users formulated by $\sum\limits_{n \in \mathcal{N}} \sum\limits_{k \in \mathcal{K}_n} \omega_k R^n_{k} (\boldsymbol{p}^n)$, where $\omega_k$ is the weight of user $k \in \mathcal{K}$, that is a constant, and is determined on the top of resource allocation. The weighted SR maximization problem can thus be formulated by
\begin{equation}\label{WSR MC-NOMA problem}
	\max_{ \boldsymbol{p} \geq 0}\hspace{.0 cm}	
	~~ \sum\limits_{n \in \mathcal{N}} \sum\limits_{k \in \mathcal{K}_n} \omega_k R^n_{k} (\boldsymbol{p}^n),~~~~~\text{s.t.}~~\eqref{Constraint minrate}\text{-}\eqref{Constraint mask power}.
\end{equation}
The feasible region of problem \eqref{WSR MC-NOMA problem} can be characterized by using Proposition \ref{proposition feasiblecluster}.
For each cluster $n$, it can be shown that if $\omega_i \geq \omega_j,~\forall i,j \in \mathcal{K}_n,~h^n_i \geq h^n_j$, the weighted SR function $\sum\limits_{k \in \mathcal{K}_n} \omega_k R^n_{k} (\boldsymbol{p}^n)$ is negative definite. In this case, the globally optimal powers can be obtained by using Proposition \ref{proposition jointcluster}, meaning that the weights $\omega_k,~\forall k \in \mathcal{K}_n$ do on affect the optimal intra-cluster power allocation policy, thus users achievable rate. Moreover, Alg. \ref{Alg bisection} finds the globally optimal solution of problem \eqref{WSR MC-NOMA problem}, such that based on \eqref{bisection opt form}, the optimal $\tilde{q}^*_n$ can be obtained as
\begin{equation}\label{bisection opt WSR}
	\tilde{q}^*_n=
	\begin{cases}
		\frac{\omega_{\Phi_n} W_s/(\ln 2)}{\nu^*} - \frac{1}{H_n}, &\quad \left(\frac{W_s/(\ln 2)}{\nu^*} - \frac{1}{H_n}\right) \in [\tilde{Q}^{\rm{min}}_n,\tilde{P}^{\rm{mask}}_n], \\
		0, &\quad \text{otherwise}. \\ 
	\end{cases}
\end{equation}
The closed-form expression \eqref{bisection opt WSR} states that when $\omega_i \geq \omega_j,~\forall n \in \mathcal{N}, i,j \in \mathcal{K}_n,~h^n_i \geq h^n_j$, we can only tune the fairness among cluster-head users. It corresponds to tuning the fairness among clusters/virtual OMA users defined in Remark \ref{remark virtualuser}.
To tune fairness among the multiplexed users within each cluster in the proportional fairness scheme, we need to assign more weights to the weaker users. For the case that $\exists n \in \mathcal{N}, i \neq j \in \mathcal{K}_n,~\omega_i < \omega_j,~h^n_i \geq h^n_j$, the weighted SR function $\sum\limits_{n \in \mathcal{N}} \sum\limits_{k \in \mathcal{K}_n} \omega_k R^n_{k} (\boldsymbol{p}^n)$ could be nonconcave, which makes problem \eqref{WSR MC-NOMA problem} nonconvex \cite{7587811}. In this regard, the strong duality in \eqref{WSR MC-NOMA problem} does not hold, thus there exists a certain duality gap in the Lagrange dual method \cite{7587811}. Although there are some interesting approximation analysis for the weighted SR function \cite{9044770}, the globally optimal solution of problem \eqref{WSR MC-NOMA problem} for the case that $\exists n \in \mathcal{N}, i \neq j \in \mathcal{K}_n,~\omega_i < \omega_j,~h^n_i \geq h^n_j$, is still an open problem. In this case, one suboptimal solution is to apply the well-known sequential programming method \cite{7954630}. 

\subsubsection{Mixed Weighted Sum-Rate/Weighted Minimum Rate Fairness}\label{subsection PFMRF}
In contrast to FDMA, proportional fairness in SC/MC-NOMA leads to a nonconvex problem in general, which greatly increases the complexity of finding the globally optimal power allocation. Another issue in proportional fairness is properly determining users weights prior to resource allocation. It is still unknown how to properly choose the users weight in order to achieve the desired users data rates after the optimal power allocation optimization which is important to guarantee users rate fairness. According to \eqref{bisection opt form}, we can conclude that in FDMA (with $\alpha_n=1$, $H_n=h^n_{\Phi_n}$, and $c_n=0$, for each $n \in \mathcal{N}$), the users minimum rate demand merely impacts on the optimal power allocation policy, thus users achievable rate. It means that tuning the minimum rate demand of users in FDMA merely impacts on the users data rate at the optimal point. In contrast to FDMA, we observe that the non-cluster-head users minimum rate demands highly affect the optimal intra-cluster power allocation of SC/MC-NOMA formulated in Proposition \ref{proposition jointcluster}. In particular, we observe that all the non-cluster-head users achieve their predefined minimum rate demands on each assigned subchannel at the optimal point of the SR/EE maximization problems. Hence, by properly increasing the target minimum rate demands of the non-cluster-head users, we not only guarantee the multiplexed users rate fairness, but also the exact achievable rate of the non-cluster-head users on each subchannel before power allocation optimization.

Let us define $\Lambda_{k,n}$ as the weight of the minimum rate demand of user $k$ on subchannel $n$. In our proposed fairness scheme, we define the minimum rate demand of each non-cluster-head user $k \in \mathcal{K}_n\setminus\{\Phi_n\}$ as $S^{\rm{min}}_{k,n} = \Lambda_{k,n} R^{\rm{min}}_{k,n}$ with $\Lambda_{k,n} \geq 1$. Based on Remark \ref{remark virtualuser}, the minimum rate demand of the cluster-head user $\Phi_n,~\forall n \in \mathcal{N}$, merely impacts on the optimal power allocation formulated in Proposition \ref{proposition jointcluster}, thus the cluster-head users achievable rate. To this end, we set $\Lambda_{\Phi_n,n} = 1,~\forall n \in \mathcal{N}$. By using the fact that each cluster-head user acts as an OMA user (see the paragraph after Remark \ref{remark virtualuser}), we apply the proportional fairness scheme among the cluster-head users in which we define $\omega_{\Phi_n}$ as the weight of the cluster-head user $\Phi_n$. Finally, the power allocation problem for the mixed weighted SR/weighted minimum rate fairness can be formulated as
\begin{subequations}\label{MSWMR problem}
	\begin{align}\label{obf MSWMR problem}
		\max_{ \boldsymbol{p} \geq 0}\hspace{.0 cm}	
		~~ & \sum\limits_{n \in \mathcal{N}} \sum\limits_{k \in \mathcal{K}_n\setminus\{\Phi_n\}} R^n_{k} (\boldsymbol{p}^n) + \omega_{\Phi_n} R^n_{\Phi_n} (\boldsymbol{p}^n)
		\\
		\text{s.t.}~~&\eqref{Constraint cell power},~\eqref{Constraint mask power}, \nonumber
		\\
		\label{Constraint WMR}
		& R^n_{k} (\boldsymbol{p}^n) \geq S^{\rm{min}}_{k,n},~\forall n \in \mathcal{N}, k \in \mathcal{K}_n.
	\end{align}
\end{subequations}
According to the discussions in Subsection \ref{subsection pf}, it is straightforward to show that the objective function \eqref{obf MSWMR problem} is strictly concave if we set $\omega_{\Phi_n} \geq 1,~\forall n \in \mathcal{N}$. For any given $\omega_{\Phi_n} \geq 1,~\forall n \in \mathcal{N}$, the intra-cluster optimal powers of problem \eqref{MSWMR problem} can be obtained by using Proposition \ref{proposition jointcluster}, in which we substitute $R^{\rm{min}}_{k,n}$ with $S^{\rm{min}}_{k,n}$. In this fairness scheme, $\omega_{\Phi_n},~\forall n \in \mathcal{N}$, are chosen to only tune the fairness among cluster-head users. Thus, we can set $\omega_{\Phi_n} \geq 1,~\forall n \in \mathcal{N}$, such that the fairness of non-cluster-head users is guaranteed by parameter $\Lambda_{k,n}$ in $S^{\rm{min}}_{k,n}=\Lambda_{k,n} R^{\rm{min}}_{k,n}$ in constraint \eqref{Constraint WMR}. In summary, the fairness parameters in problem \eqref{MSWMR problem} satisfy $\Lambda_{k,n} \geq 1,~\forall n \in \mathcal{N},~k \in \mathcal{K}_n\setminus\{\Phi_n\}$, $\Lambda_{k,\Phi_n} = 1,~\forall n \in \mathcal{N}$, and $\omega_{\Phi_n} \geq 1,~\forall n \in \mathcal{N}$. Note that in the objective function \eqref{obf MSWMR problem}, the weight of each non-cluster-head user within each cluster is one. The feasible region of problem \eqref{MSWMR problem} can be characterized by using Proposition \ref{proposition feasiblecluster}, in which we substitute $R^{\rm{min}}_{k,n}$ with $S^{\rm{min}}_{k,n}$. Finally, the water-filling Alg. \ref{Alg bisection} can be applied to find the globally optimal solution of \eqref{MSWMR problem} in which the optimal $\tilde{q}^*_n$ is given by \eqref{bisection opt WSR}. 

It can be shown that similar to proportional fairness, by properly choosing the fairness parameters $\Lambda_{k,n} \geq 1,~\forall n \in \mathcal{N},~k \in \mathcal{K}_n\setminus\{\Phi_n\}$, and $\omega_{\Phi_n} \geq 1,~\forall n \in \mathcal{N}$, our proposed fairness scheme can also achieve any feasible desired rates for all the users in Hybrid-NOMA, which is important to guarantee any users rate fairness level. Similar to the proportional fairness, it is still difficult to properly assign the weight of the cluster-head users in our proposed fairness scheme denoted by $\omega_{\Phi_n} \geq 1,~\forall n \in \mathcal{N}$. Another challenge is properly setting $\Lambda_{k,n}$ of each non-cluster-head user $k \in \mathcal{K}_n\setminus\{\Phi_n\}$ prior to resource allocation optimization. This is because, the parameter $S^{\rm{min}}_{k,n},~\forall k \in \mathcal{K}_n\setminus\{\Phi_n\}$, significantly increases $Q^{\rm{min}}_n$ in \eqref{Qmin func}. Hence, significantly large $\Lambda_{k,n}$ for user $k$ may lead to empty feasible region for each subchannel $n \in \mathcal{N}_k,~k \neq \Phi_n$ (user $k$ is not cluster-head). It is worth noting that for any given $\Lambda_{k,n}$, Corollary \ref{corol feasible region} is useful to immediately verify whether the feasible region is empty or not. One interesting topic is how to achieve a preferable/absolute users rate fairness by properly choosing the fairness parameters $\Lambda_{k,n} \geq 1,~\forall n \in \mathcal{N},~k \in \mathcal{K}_n\setminus\{\Phi_n\}$, and $\omega_{\Phi_n} \geq 1,~\forall n \in \mathcal{N}$, in our proposed fairness scheme, which brings new theoretical insights on the fundamental relations between our proposed and the well-known proportional/max-min rate fairness schemes.

\subsection{Imperfect Channel State Information}
Unfortunately, it is difficult to acquire the perfect CSI of users, due to the existence of channel estimation errors, feedback delay, and quantization error. In NOMA with imperfect CSI, the imperfect CSI may lead to incorrect user ordering for SIC within a cluster resulting in outage \cite{7934461}, namely SIC outage. By employing the stochastic method, the CNR of user $k \in \mathcal{K}$ on subchannel $n \in \mathcal{N}$ can be modeled as $h^n_k=\hat{h}^n_k + e^n_k$, where $e^n_k \sim \mathcal{CN} \left(0,\sigma^2_e\right)$ and $\hat{h}^n_k \sim \mathcal{CN} \left(0,1-\sigma^2_e\right)$ denote the estimation error normalized by noise and estimated CNR, respectively. Assume that the estimated CNR $\hat{h}^n_k$ and normalized estimation error $e^n_k$ are uncorrelated \cite{8119791,7934461,9032198}. In each cluster $n$, by performing user ordering based on $\hat{h}^n_k,~\forall k \in \mathcal{K}_n$, the SIC outage occurs if and only if there exists at least on user pair $i,j \in \mathcal{K}_n,~\hat{h}^n_i > \hat{h}^n_j$, while $h^n_i < h^n_j$. 
Assume that $e^n_k \in [L^n_k,U^n_k],~\forall n \in \mathcal{N},~k \in \mathcal{K}_n$. The SIC outage is thus zero if and only if 
$\min\limits_{e^n_i \in [L^n_i,U^n_i]}h^n_i \geq \max\limits_{e^n_j \in [L^n_j,U^n_j]}h^n_j,~\forall n \in \mathcal{N},~k \in \mathcal{K}_n,~\hat{h}^n_i > \hat{h}^n_j.$
In the latter condition, $\min\limits_{e^n_i \in [L^n_i,U^n_i]}h^n_i= \hat{h}^n_i + L^n_i$, and $\max\limits_{e^n_j \in [L^n_j,U^n_j]}h^n_j= \hat{h}^n_j + U^n_j$. Therefore, the SIC outage is zero if and only if 
\begin{equation}\label{zero SIC outage}
	\hat{h}^n_i + L^n_i \geq \hat{h}^n_j + U^n_j,~\forall n \in \mathcal{N},~k \in \mathcal{K}_n,~\hat{h}^n_i > \hat{h}^n_j.
\end{equation}
The condition \eqref{zero SIC outage} states that with imperfect CSI, there exists an additional outage, called SIC outage in SC/MC-NOMA, if for a multiplexed user pair, the best case of the CNR of the weaker user is greater than the worst case of the CNR of the stronger user. The SIC outage depends on the region of normalized estimation errors and estimated CNRs. Thus, the SIC outage cannot be tuned by means of power allocation optimization. The latter result is due to the fact that the SIC decoding order of users in SISO Gaussian BCs is independent of power allocation. The SIC outage probability of the $2$-user SC-NOMA system is analyzed in \cite{7934461}.
Although when the condition in \eqref{zero SIC outage} is not fulfilled, we cannot achieve the zero-SIC outage by means of power allocation, the zero-SIC outage can be achieved by the user clustering of MC-NOMA, or in general, subchannel allocation. For example, when the lower-bound $L^n_i$ and upper-bound $U^n_j$ in \eqref{zero SIC outage} are available, we are able to impose the condition in \eqref{zero SIC outage} as a necessary constraint in user clustering problem to achieve the zero-SIC outage which increases the robustness of MC-NOMA. The impact of imposing the condition in \eqref{zero SIC outage} in user clustering can be considered as a future work.

The work in \cite{8119791} provided new analyses for the EE maximization problem of Hybrid-NOMA with imperfect CSI, when the large-scale fading factors are slowly varying, thus they can be estimated perfectly at the BS.
According to Subsection III in \cite{8119791}, it is straightforward to show that our analysis is also valid for Hybrid-NOMA with imperfect CSI, and considering per-symbol maximum outage probability and minimum rate constraints.
It is worth noting that the imperfect CSI merely impacts on the intra-cluster power allocation, due to the high insensitivity of optimal intra-cluster power allocation to the users exact CNRs for the SR/EE maximization problems discussed in Subsection \ref{subsection insights sum-rate}. One important future direction of our work is to evaluate the optimality and robustness of the approximated closed-form of optimal powers in \eqref{opt IC i approx} and \eqref{opt IC M approx} for the imperfect CSI scenarios.

\subsection{Admission Control}
One important application of NOMA is to support massive connectivity in the 5G networks, e.g., IoT use-cases \cite{9311149}. When the number of users and/or their minimum rate demands increases, the parameter $Q^{\rm{min}}_n$ increases leading to tightening the feasible region of the formulated optimization problems characterized by Proposition \ref{proposition feasiblecluster}. As a result, for significantly large number of users and/or their minimum rate demands, the feasible region will be empty, and subsequently, the problems will be infeasible. As such, the network cannot support all of the users simultaneously, thus an admission control policy is necessary to support the maximum possible number of users/transmitted symbols on subchannels. There are few works addressing the admission control for the SC/MC-NOMA systems \cite{7974731,8365765,8301007,8995591}. The globally optimal admission control policy for the general SC/MC-NOMA systems with individual minimum rate demands is still an open problem. To admit more desired symbols in Hybrid-NOMA while reducing the cellular power consumption, one suboptimal solution is to first calculate the power consumption of each user on each subchannel given by \eqref{optimal minpow Scell}. Then, eliminate the subchannel (thus transmitted symbol) for the user which consumes the highest power. After that, recalculate \eqref{optimal minpow Scell} for the updated $\mathcal{K}_n$. The latter steps will be continued until Corollary \ref{corol feasible region} is fulfilled. One future work can be how to incorporate the closed-form of optimal powers in Propositions \ref{proposition feasiblecluster} and \ref{proposition jointcluster} in the admission control policy to admit more users while minimizing the cellular power consumption, or maximizing the admitted users sum-rate, respectively.

\subsection{Reconfigurable Intelligent Surfaces-aided NOMA}
The NOMA technology has been recently integrated with RISs \cite{9475160,9122596}. In RIS-assisted NOMA, the joint power and phase shift allocation is shown to be necessary to achieve the optimal solution of the SR maximization problem \cite{9475160,9122596}. Unfortunately, the optimal joint power and phase shift optimization is intractable, thus many recent works applied the alternate optimization, where we find the optimal powers/phase shifts when the other is given.
In general, for any given phase shifts, the RIS-NOMA system, such as the considered model in \cite{9365004}, can be equivalently transformed to a NOMA system with users equivalent channel gains. In this way, it is straightforward to show that all the analysis of power allocation for the SR/EE maximization problems of the pure SC/MC-NOMA system are also valid for an RIS-NOMA system with the given phase shifts, thus the users corresponding equivalent channel gains. For example, the closed-form of optimal powers for the RIS-assisted NOMA system in \cite{9365004} can be obtained by using Proposition \ref{proposition jointcluster} with $N=1$. From \eqref{opt IC i approx} and \eqref{opt IC M approx}, it can be concluded that in the high-SINR regions of an RIS-NOMA system, the optimal powers are insensitive to the equivalent channel gains, thus phase shifts. Therefore, we expect that the alternate optimization approaches a near-optimal solution with a fast convergence speed in the high-SINR regions of an RIS-NOMA system. The extension of our analysis to an RIS-assisted MC-NOMA system can be considered as a future work.

\subsection{Long-Term Resource Allocation}
Similar to most of the related works, we assume a dynamic resource allocation framework, where the allocated powers to the users will be readopted every time slot based on the arrival set of active users, and instantaneous CSI. It is shown that the short-term designs may lead to inferior system performance in a long-term perspective \cite{7982783}. There are a number of works that addressed the long-term resource allocation optimization in NOMA, e.g., \cite{7982783,7845657,8318569}. In \cite{7982783}, the authors developed the well-known Lyapunov optimization framework to convert the long-term sum-rate maximization problem of SC-NOMA with long-term
average and short-term peak power constraints, and per-user maximum rate constraints into a series of online "weighted-sum-rate minus weighted-total power consumption" maximization problem in each time slot. The latter rusting problem can be classified as the power allocation problem for SC-NOMA with proportional fairness.
Although there has been some efforts in \cite{7982783} to further reduce the searching space of optimal power allocation, the closed-form expression of optimal power allocation for the long-term optimization framework in \cite{7982783} is still an open problem. The analysis will be more complicated if we consider the Hybrid-NOMA scheme with per-user/symbol minimum rate constraints, and optimal inter-cluster power allocation, which is still an open problem, and can be considered as a future work.

In \cite{7845657}, the long-term optimization is addressed by properly choosing the users weights in the proportional fairness scheme. In particular, the proportional fairness scheduler keeps track of the average rate of each user in the past time slots with limited length, and reflect these average rates to the users weights. A similar framework can also be applied to our proposed mixed weighted SR/weighted minimum rate fairness scheme in Subsection \ref{subsection PFMRF}, where the fairness parameters $\Lambda_{k,n} \geq 1,~\forall n \in \mathcal{N},~k \in \mathcal{K}_n\setminus\{\Phi_n\}$, and $\omega_{\Phi_n} \geq 1,~\forall n \in \mathcal{N}$, are chosen in \eqref{MSWMR problem} based on the average users rate in the past time slots, which can be considered as a future work.

\section{Simulation Results}\label{Section simulation}
In this section, we evaluate the performance of SC-NOMA, FDMA-NOMA (FD-NOMA) with different $U^{\rm{max}}$, and FDMA for different performance metrics as outage probability, BSs minimum power consumption to satisfy users minimum rate demands, maximum users SR, and maximum system EE. To reflect the randomness impact, we apply the Monte-Carlo simulations \cite{7587811,8422362,9044770,7557079,8114362,7982784,7523951,9276828,8119791,8448840,9032198} by averaging over $50,000$ channel realizations.
The outage probability is calculated by dividing the number of infeasible solutions determined according to Corollary \ref{corol feasible region}, by total number of channel realizations. According to Proposition \ref{proposition feasiblecluster}, the minimum BS's power consumption can be obtained by $P_{\rm{min}}=\sum\limits_{n \in \mathcal{N}} Q^{\rm{min}}_n$. All the algorithms in Table \ref{table complexity SR} can globally solve the SR maximization problem, however with different computational complexities. For our simulations, we select Alg. \ref{Alg bisection} with the lowest complexity compared to the others. Moreover, all the mentioned algorithms in Table \ref{table complexity EE} can optimally solve the EE maximization problem with different computational complexities. For our simulations, we select Alg. \ref{Alg dinkelbach} with inner Alg. \ref{Alg mixedWFSub} which has the lowest complexity compared to the others. Since SC-NOMA and FDMA are special cases of FDMA-NOMA, our selected algorithms are modified to optimally solve these problems.
The simulation settings are shown in Table \ref{Table parameters}.
\begin{table}[tp]
	\centering
	\caption{System Parameters}
	\begin{adjustbox}{width=\columnwidth,center}
		\begin{tabular}{c c}
			\hline
			\textbf{Parameter} & \textbf{Value} \\
			\hline
			BS maximum transmit power $(P^{\rm{max}})$ & 46 dBm \\
			Circuit power consumption $(P_{\rm{C}})$ & 30 dBm \\
			Coverage of BS & Circular with radii of $500$ m  \\
			Wireless bandwidth $(W)$ & 5 MHz \\
			Number of users $(K)$ & $\{5,10,15,\dots,60\}$  \\
			User distribution model & Uniform distribution \\
			$U^{\rm{max}}$ in NOMA  & $\{2,4,6\}$ \\
			Minimum distance of users to BS & 20 m  \\
			Distance-depended path loss & $128.1 + 37.6 \log_{10} (d)$ dB, where $d$ is in Km  \\
			Lognormal shadowing standard deviation & $8$ dB  \\
			Small-scale fading & Rayleigh flat fading \\
			AWGN power density & -174 dBm/Hz \\
			Minimum rate demand of each user $(R^{\rm{min}}_{k})$  & $\{0.25,0.5,0.75,1,\dots,5\}$ Mbps \\
			\hline
		\end{tabular}
		\label{Table parameters}
	\end{adjustbox}
\end{table}
Without loss of generality, we set $P^{\rm{mask}}_n=P^{\rm{max}},~\forall n \in \mathcal{N}$.
In our simulations, we apply a fast suboptimal user clustering method\footnote{Since our globally optimal power allocation algorithms are valid for any given cluster sets, the existing suboptimal user clustering algorithms, such as heuristic methods in \cite{7557079,8114362,9276828,8119791,8448840}, and  matching-based algorithms in \cite{7982784,7523951} can also be applied.} for the flat fading channels of FD-NOMA presented in Alg. \ref{Alg cluster}. 
In this method, we first obtain $N=\lceil K/U^{\rm{max}}\rceil$ according to Proposition \ref{proposition optgroup}. The ranking vector $\mathbf{R}=[R_k],~\forall k \in \mathcal{K}$, is the vector of the ranking of users CNR, in which $R_k \in \{1,\dots,K\},~\forall k \in \mathcal{K}$, such that $R_k > R_{k'}$ if $h^n_{k} > h^n_{k'}$.
\begin{algorithm}[tp]
	\caption{Suboptimal User Clustering for FD-NOMA.} \label{Alg cluster}
	\begin{algorithmic}[1]
		\STATE Compute the number of clusters as $N=\lceil K/U^{\rm{max}}\rceil$.
		\STATE Initialize $\rho^n_{k}=0,~\forall n \in \mathcal{N},~k \in \mathcal{K}_n$, $n=0$, and ranking vector $\mathbf{R}=[R_k],~\forall k \in \mathcal{K}$.
		\WHILE {$\norm{\mathbf{R}} > 0$}
		\STATE Find $k^* = \argmax\limits_{k \in \mathcal{K}}~\mathbf{R}$.
		\\
		\STATE Set $n:=n+1$.
		\\
		\STATE\textbf{if}~~$n>N$~~\textbf{then}
		\\
		\STATE~~~~~Set $n=1$.
		\\
		\STATE\textbf{end if}
		\STATE Set $\rho^n_{k^*}=1$, and $R_{k^*}=0$.
		\ENDWHILE
	\end{algorithmic}
\end{algorithm}
In Alg. \ref{Alg cluster}, the first $N$ users with the highest CNRs are assigned to different clusters. The rest of the users with lower decoding orders are distributed over the subchannels based on their CNRs.
The subchannel allocation of FDMA in flat fading channels is straightforward, since any subchannel-to-user allocation is optimal.
The source code of the simulations including a user guide is available in \cite{sourcecode}. In the following, the term '$X$-NOMA' is referred to FD-NOMA with $U^{\rm{max}}=X$.

\subsection{System Outage Probability Performance}
The impact of minimum rate demands and number of users on the system outage probability of different multiple access techniques is shown in Fig. \ref{Figoutage}.
\begin{figure*}
	\centering
	\subfigure[Outage probability vs. users minimum rate demand for $K=10$.]{
		\includegraphics[scale=0.35]{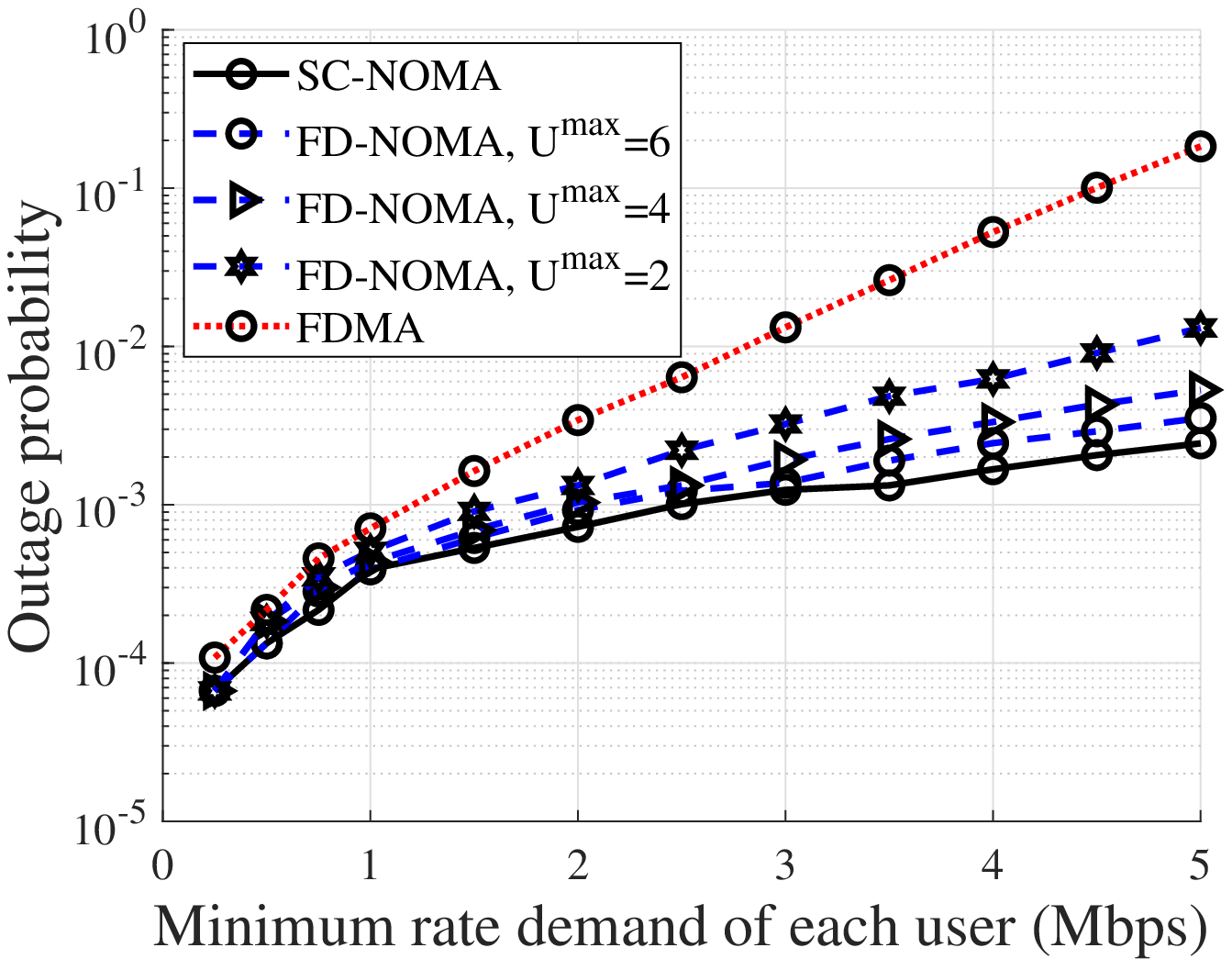}
		\label{Fig_outage_minrate1}
	}\hfill
	\subfigure[Outage probability vs. users minimum rate demand for $K=30$.]{
		\includegraphics[scale=0.35]{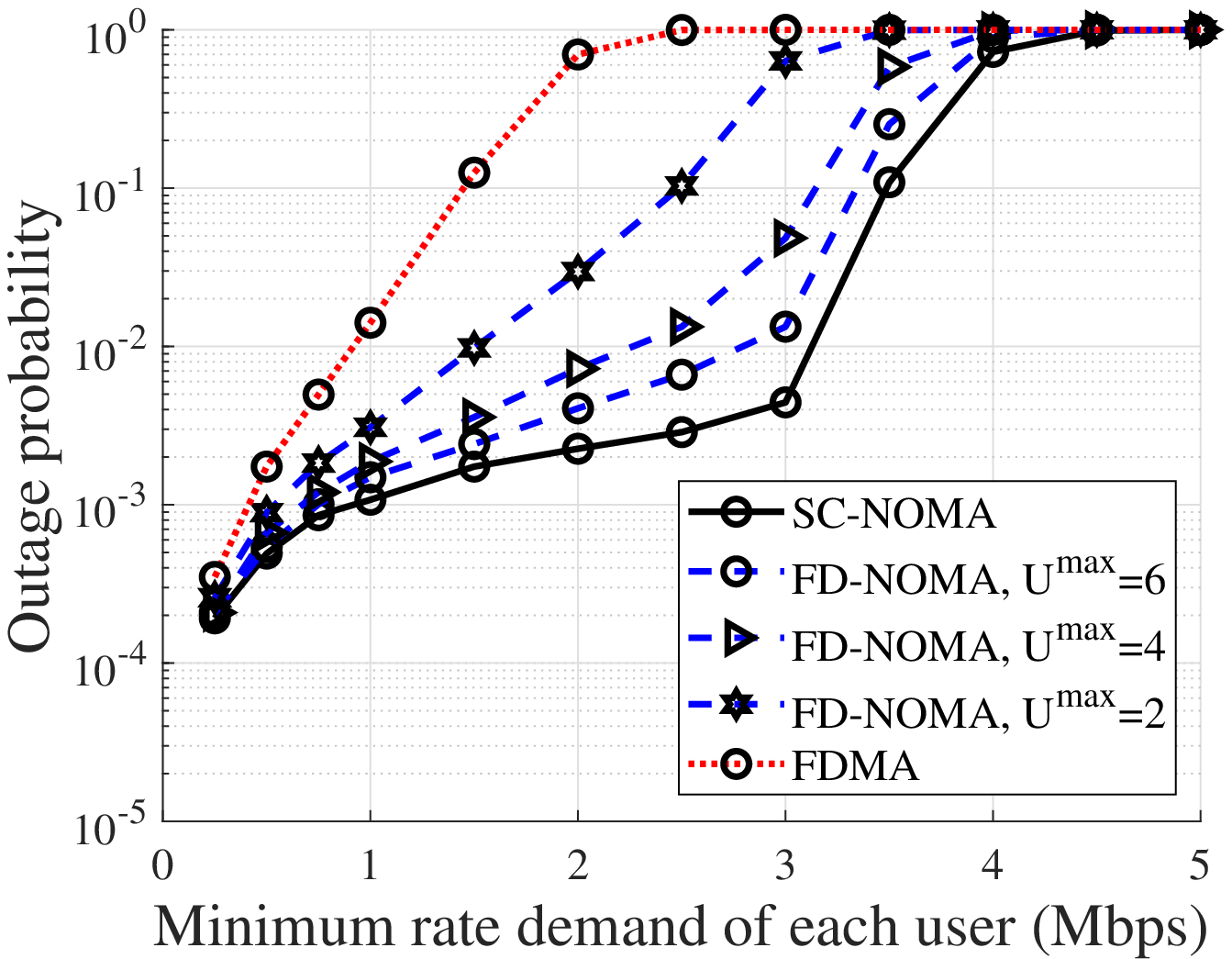}
		\label{Fig_outage_minrate2}
	}\hfill
	\subfigure[Outage probability vs. users minimum rate demand for $K=50$.]{
		\includegraphics[scale=0.35]{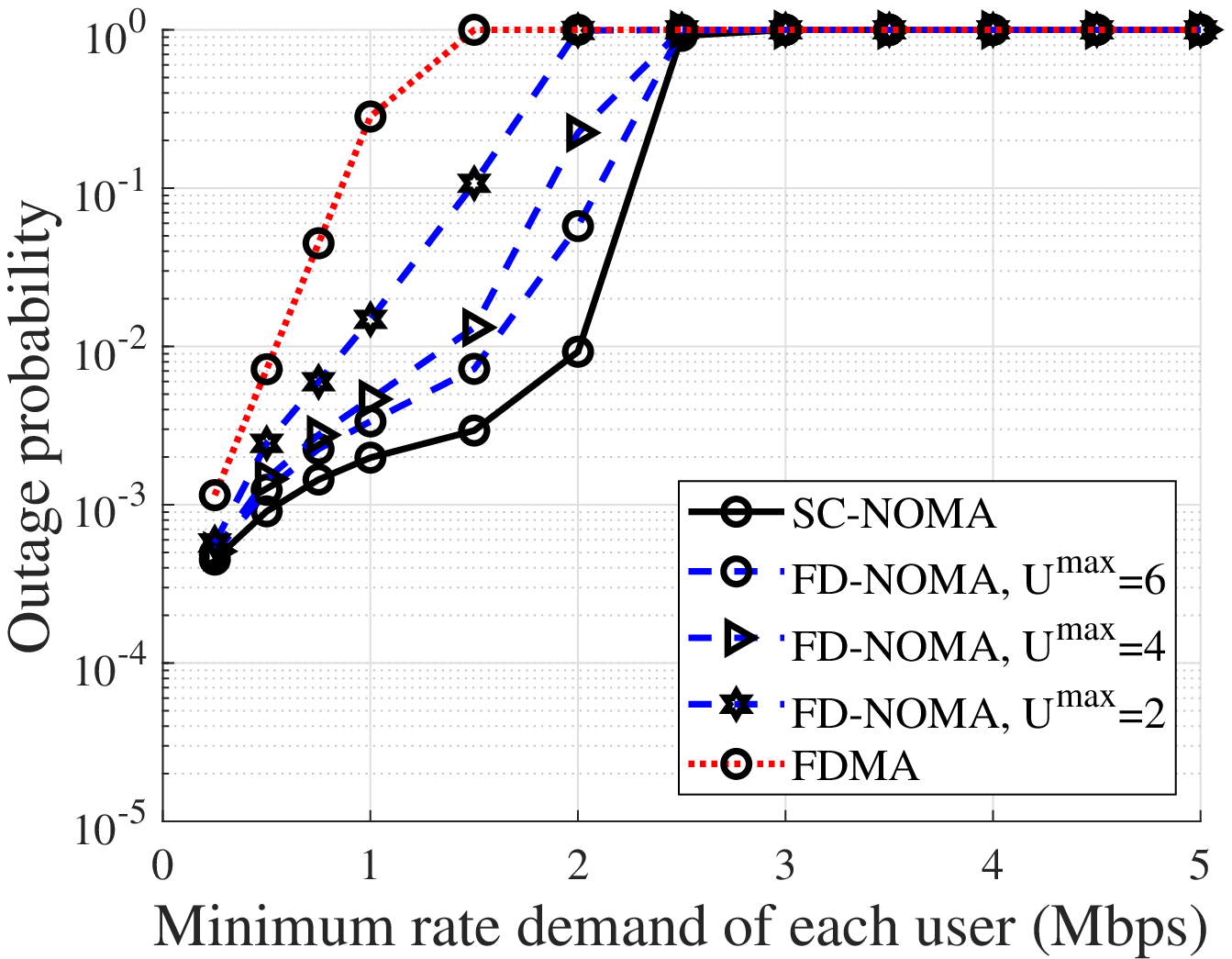}
		\label{Fig_outage_minrate3}
	}\hfill
	\subfigure[Outage probability vs. number of users for $R^{\rm{min}}_{k}=1.5~\text{Mbps},\forall k \in \mathcal{K}$.]{
		\includegraphics[scale=0.35]{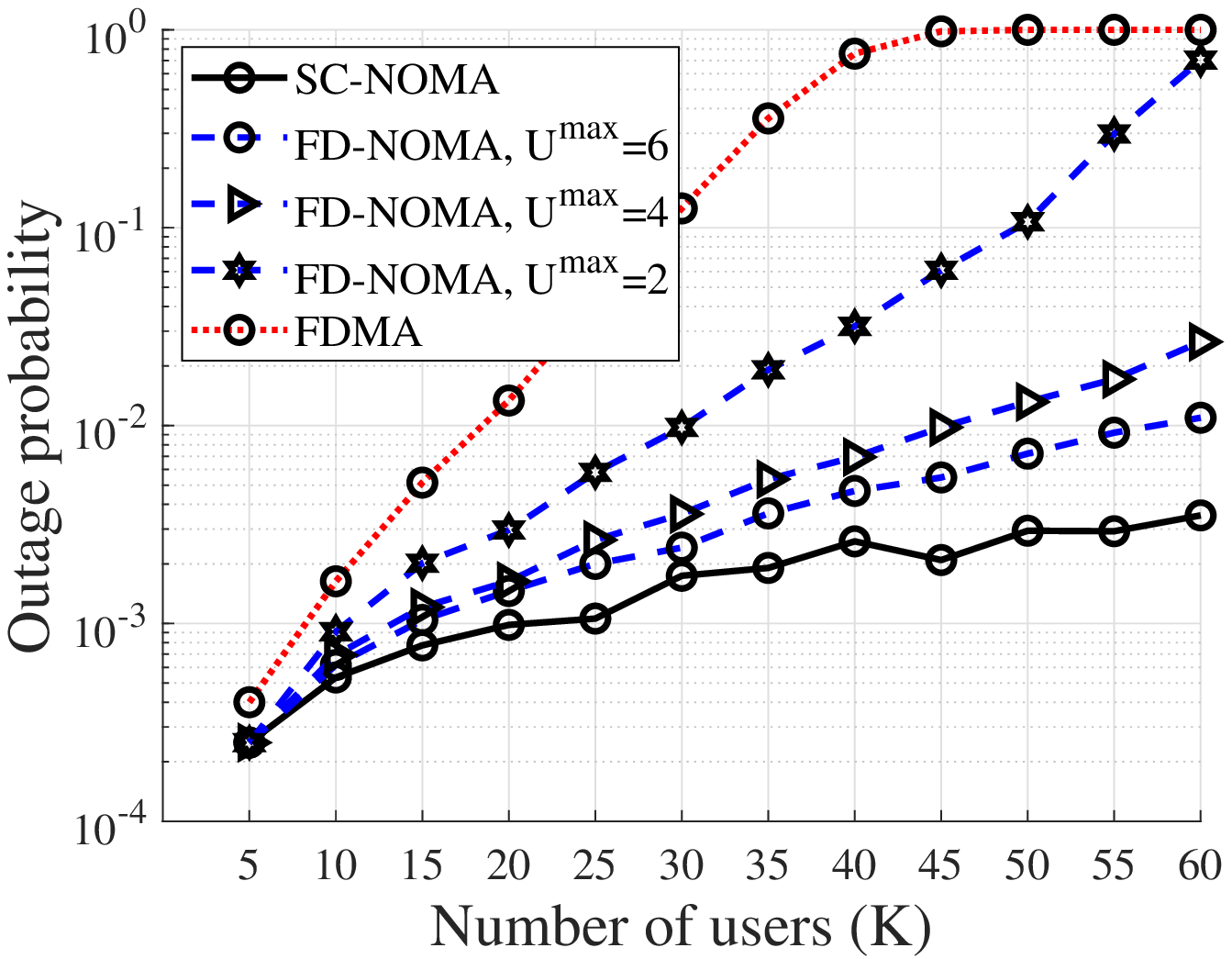}
		\label{Fig_outage_usernum1}
	}\hfill
	\subfigure[Outage probability vs. number of users for $R^{\rm{min}}_{k}=3~\text{Mbps},\forall k \in \mathcal{K}$.]{
		\includegraphics[scale=0.35]{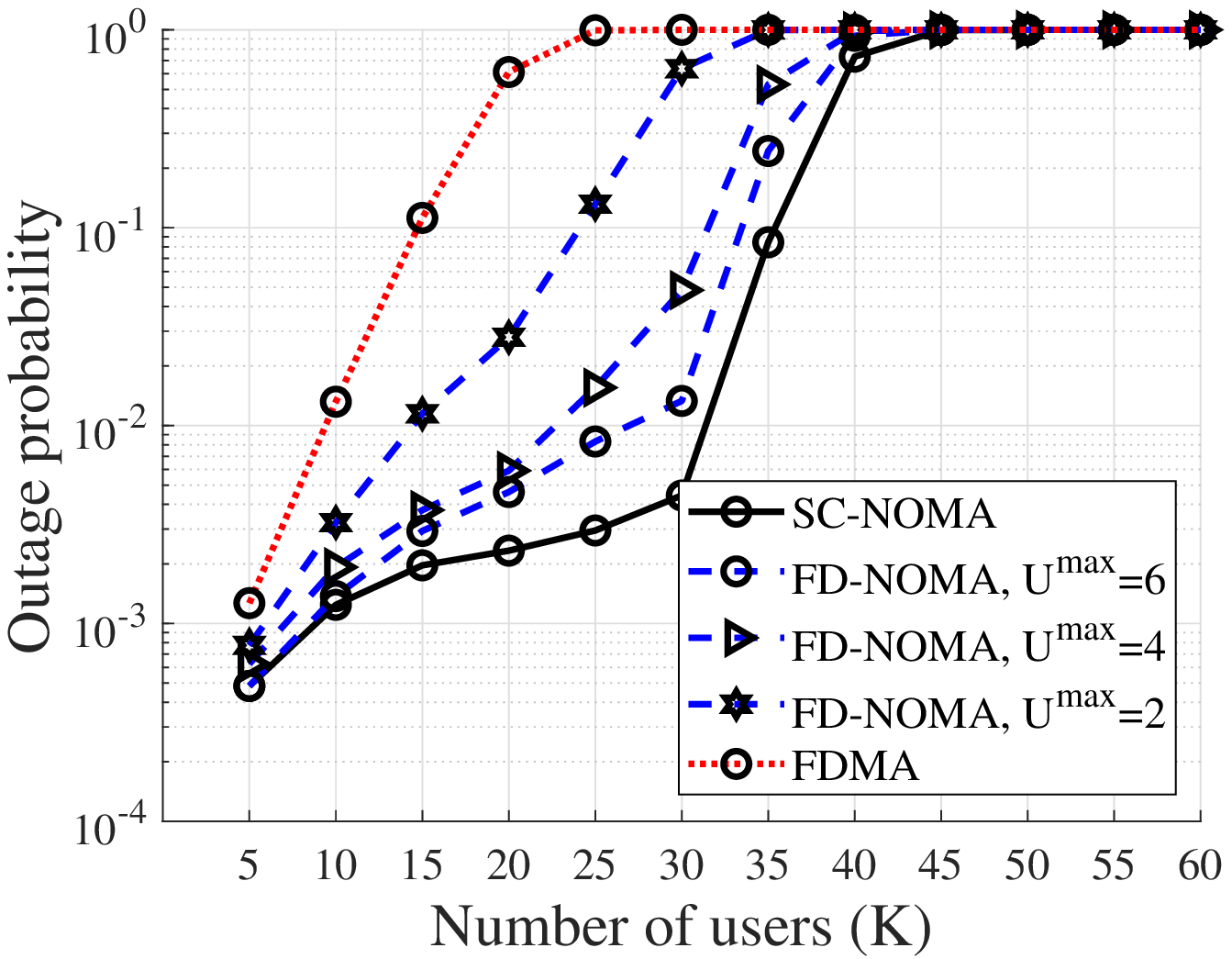}
		\label{Fig_outage_usernum2}
	}\hfill
	\subfigure[Outage probability vs. number of users for $R^{\rm{min}}_{k}=4.5~\text{Mbps},\forall k \in \mathcal{K}$.]{
		\includegraphics[scale=0.35]{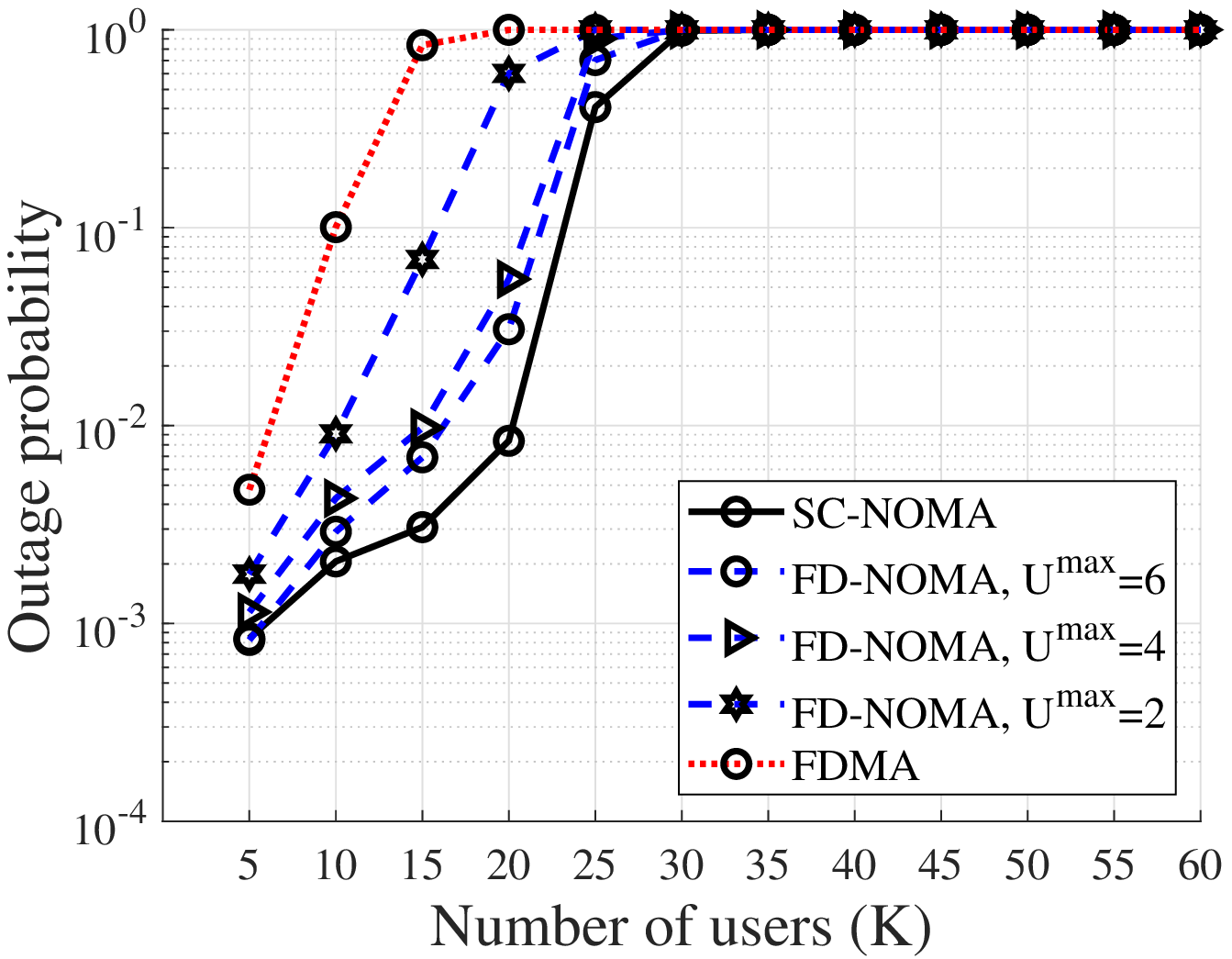}
		\label{Fig_outage_usernum3}
	}\hfill
	\caption
	{Impact of the minimum rate demand and number of users on the outage probability of SC-NOMA, FD-NOMA, and FDMA.}
	\label{Figoutage}
\end{figure*} 
According to \eqref{Qmin func}, $Q^{\rm{min}}_n$ is increasing in $R^{\rm{min}}_{k}$. For quite small $R^{\rm{min}}_{k}$ and/or $K$, the performance gap between different multiple access techniques is low. For larger $R^{\rm{min}}_{k}$ and/or $K$, we observe a significant performance gap between FDMA and $X$-NOMA ($X\geq 2$), and also between $2$-NOMA and $4$-NOMA. Moreover, it can be observed that the performance gap between $4$-NOMA and $6$-NOMA is low. Finally, for quite large $R^{\rm{min}}_{k}$ and/or $K$, the outage probability of all these techniques tends to $1$. In summary, the outage probability follows: $\text{outage}(\text{SC-NOMA}) < \text{outage}(6\text{-NOMA}) \approx \text{outage}(4\text{-NOMA}) < \text{outage}(2\text{-NOMA}) \ll \text{outage}(\text{FDMA})$.

\subsection{Average Minimum BS's Power Consumption Performance}
The impact of minimum rate demands and number of users on average total power consumption of different multiple access techniques is shown in Fig. \ref{FigPM}.
\begin{figure*}
	\centering
	\subfigure[Average total power consumption vs. users minimum rate demand for $K=10$.]{
		\includegraphics[scale=0.35]{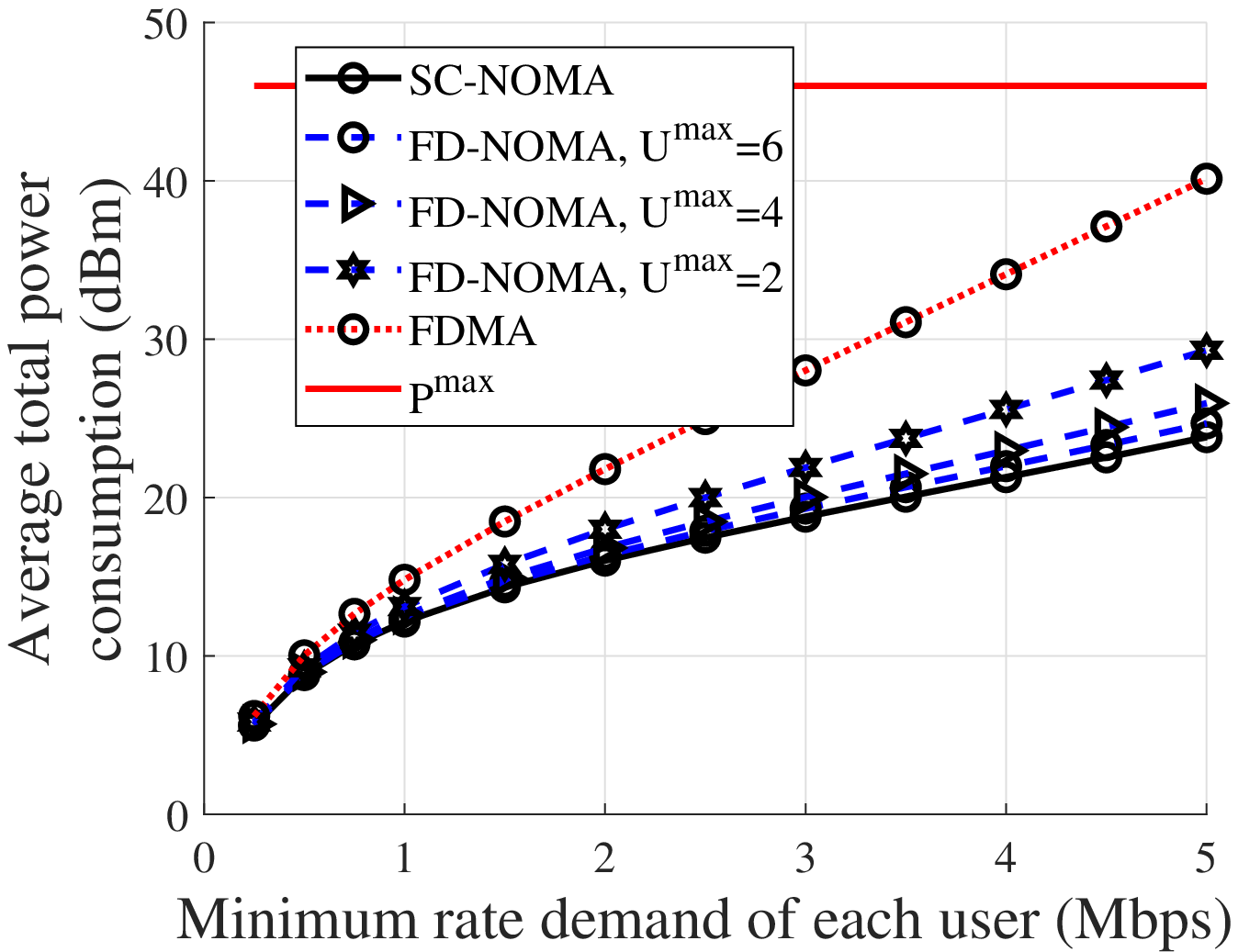}
		\label{Fig_PM_minrate1}
	}\hfill
	\subfigure[Average total power consumption vs. users minimum rate demand for $K=30$.]{
		\includegraphics[scale=0.35]{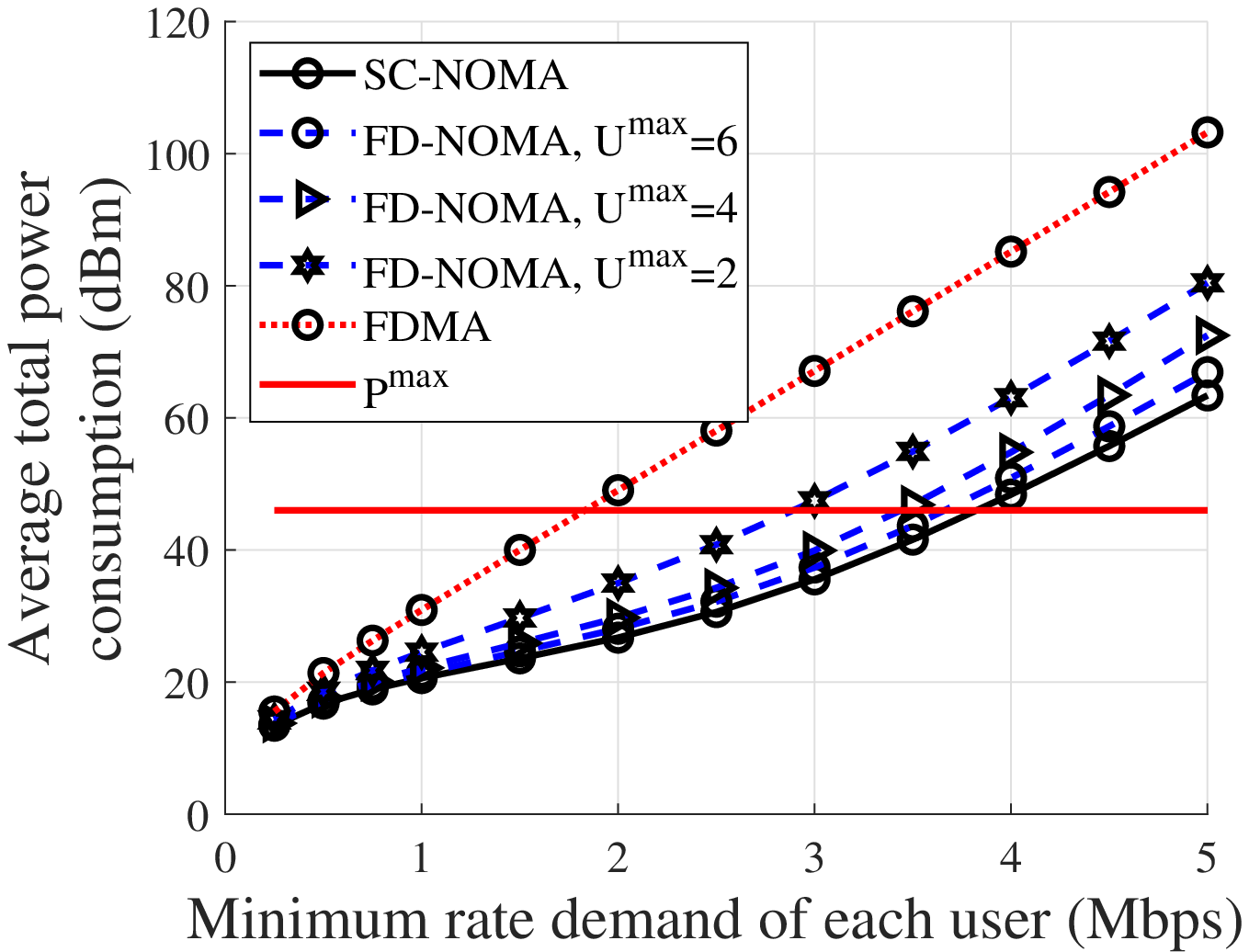}
		\label{Fig_PM_minrate2}
	}\hfill
	\subfigure[Average total power consumption vs. users minimum rate demand for $K=50$.]{
		\includegraphics[scale=0.35]{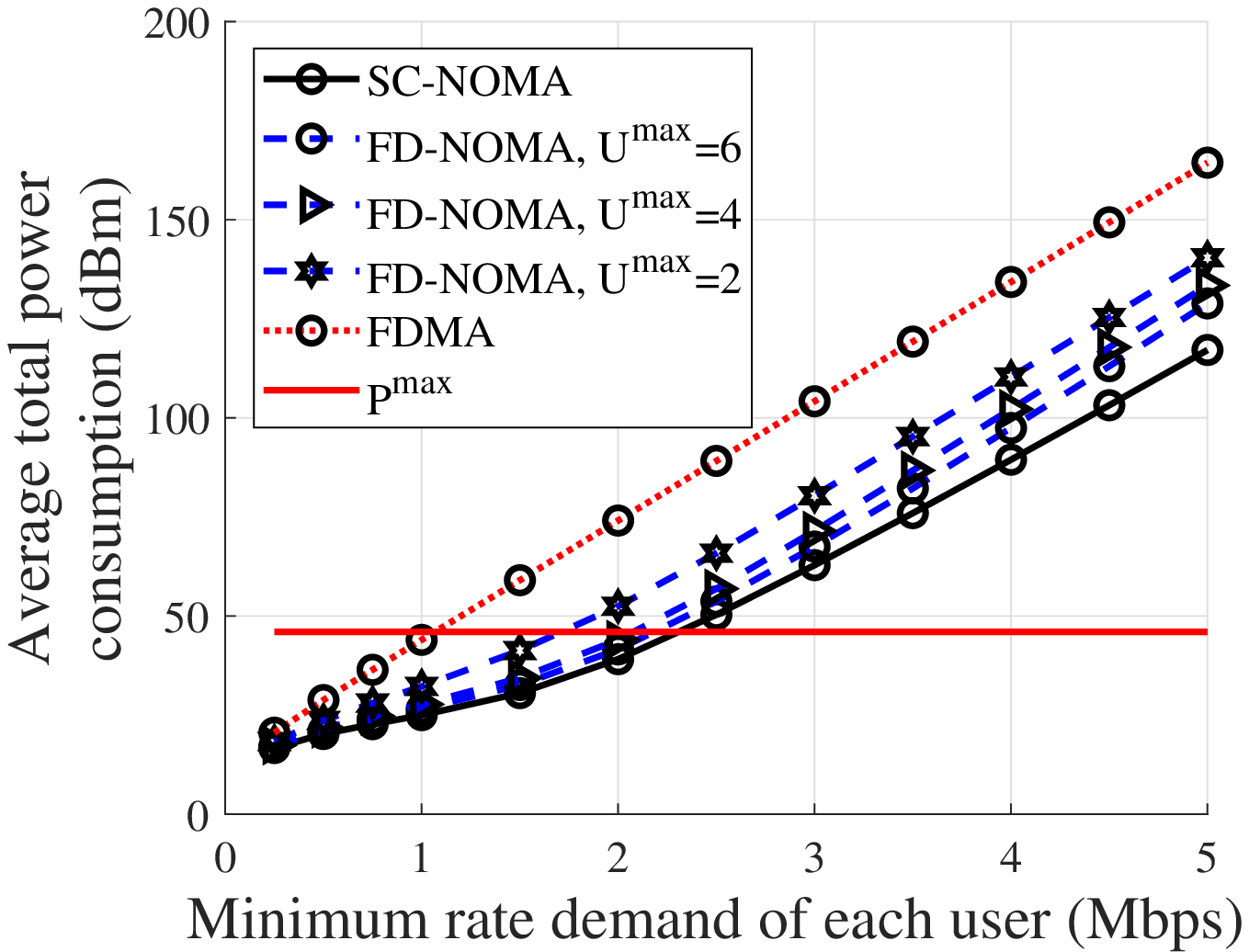}
		\label{Fig_PM_minrate3}
	}\hfill
	\subfigure[Average total power consumption vs. number of users for $R^{\rm{min}}_{k}=1.5~\text{Mbps},\forall k \in \mathcal{K}$.]{
		\includegraphics[scale=0.35]{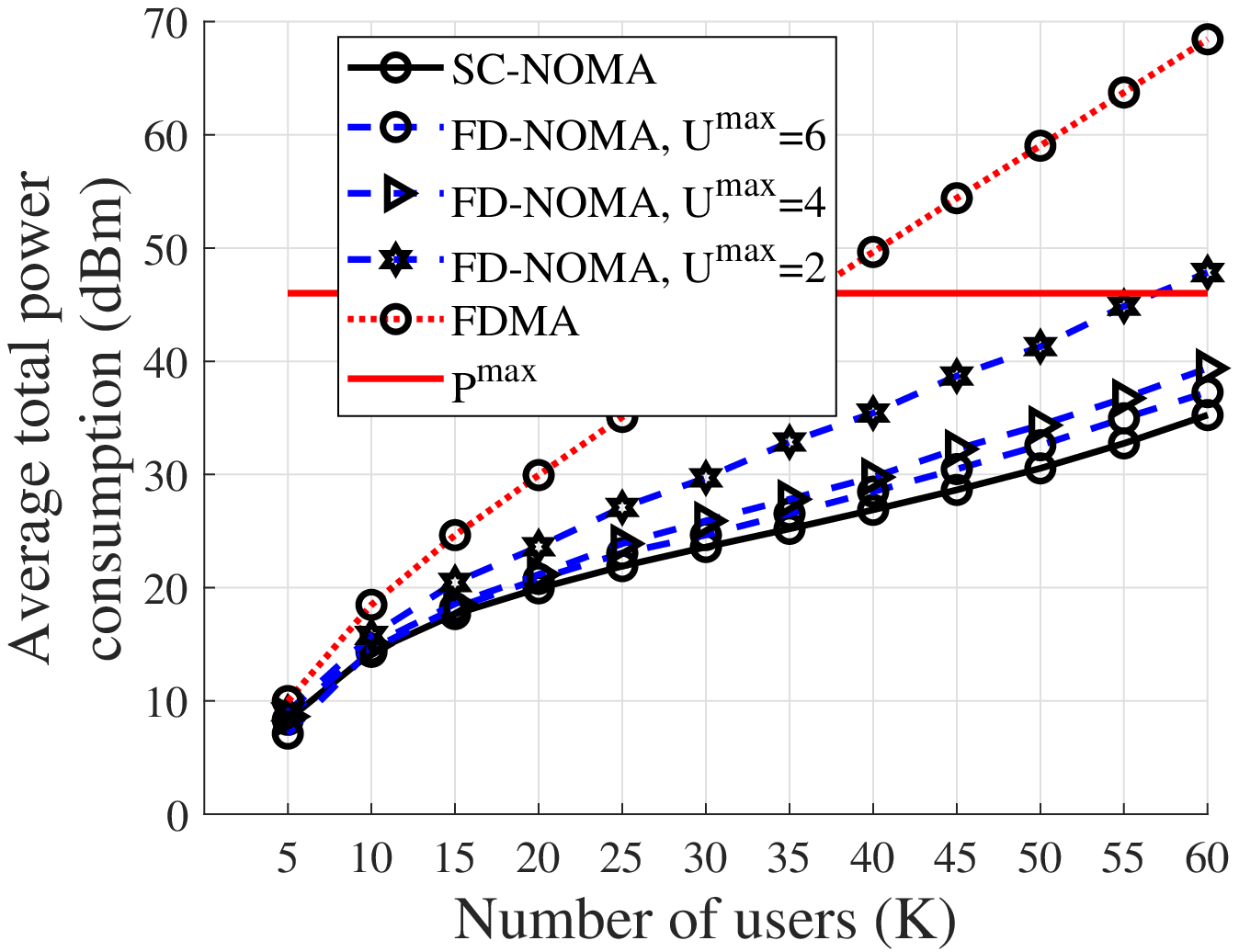}
		\label{Fig_PM_usernum1}
	}\hfill
	\subfigure[Average total power consumption vs. number of users for $R^{\rm{min}}_{k}=3~\text{Mbps},\forall k \in \mathcal{K}$.]{
		\includegraphics[scale=0.35]{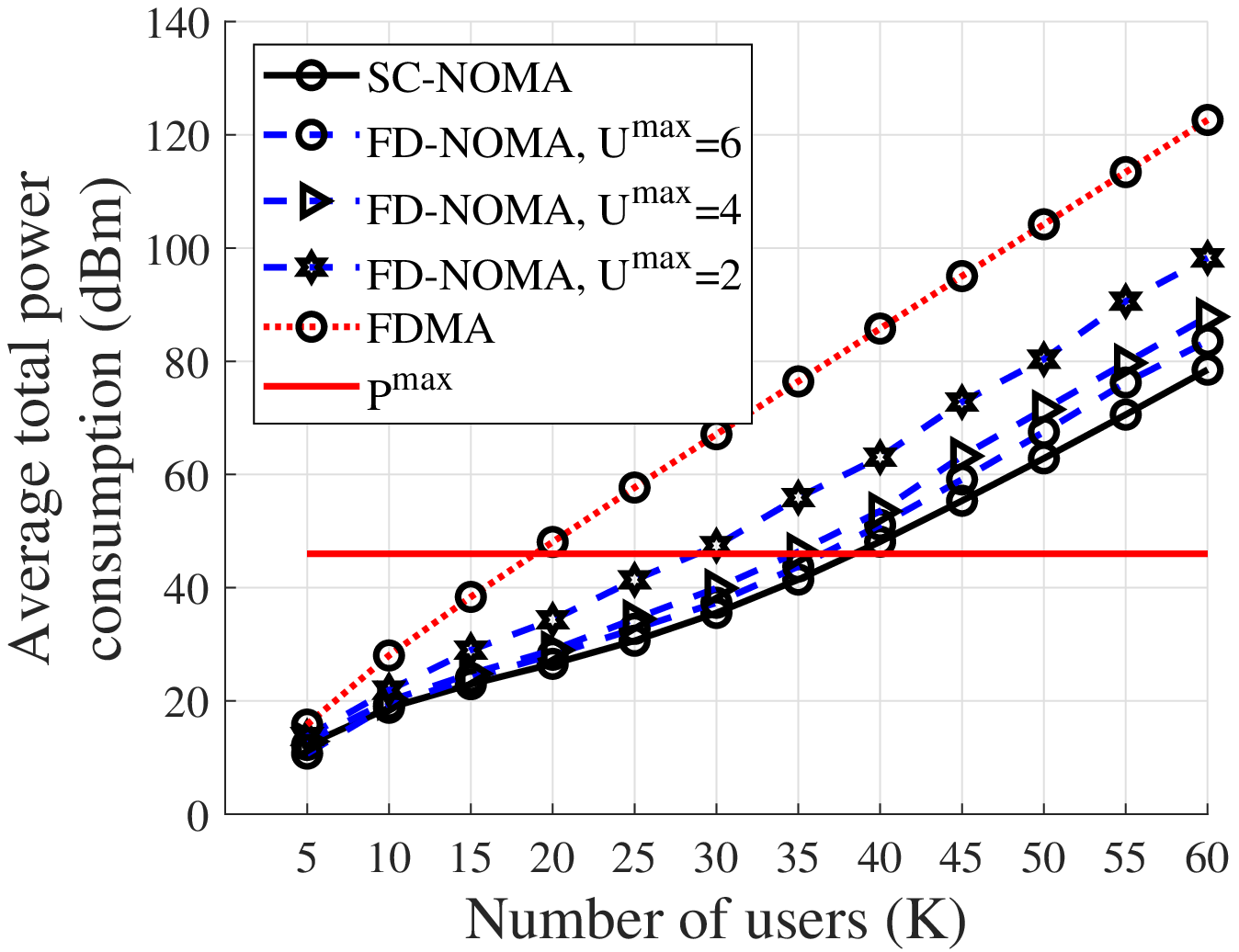}
		\label{Fig_PM_usernum2}
	}\hfill
	\subfigure[Average total power consumption vs. number of users for $R^{\rm{min}}_{k}=4.5~\text{Mbps},\forall k \in \mathcal{K}$.]{
		\includegraphics[scale=0.35]{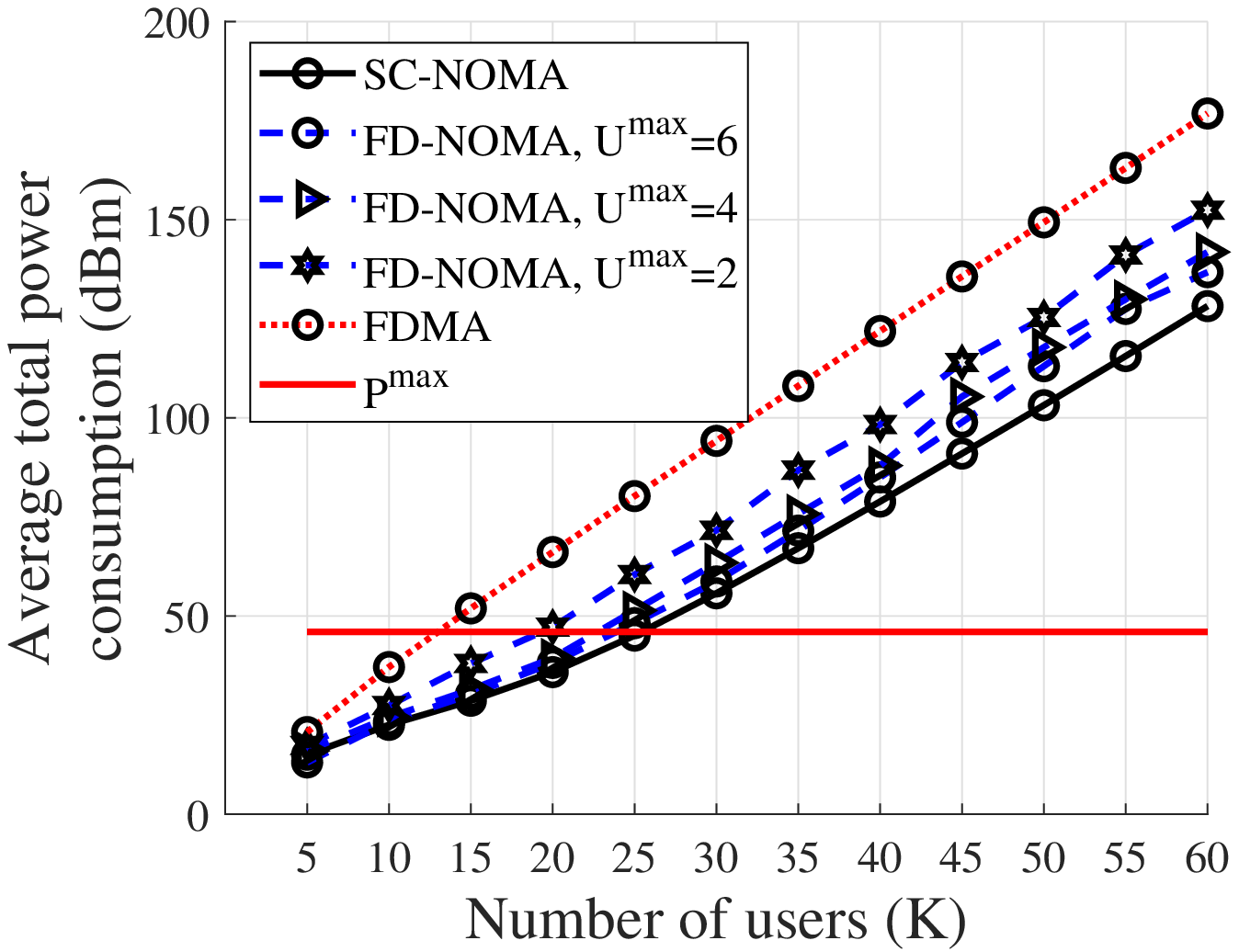}
		\label{Fig_PM_usernum3}
	}\hfill
	\caption
	{Impact of the minimum rate demand and number of users on the average total power consumption of SC-NOMA, FD-NOMA, and FDMA.}
	\label{FigPM}
\end{figure*}
As can be seen, there exists a significant performance gap between FDMA and FD-NOMA for larger $R^{\rm{min}}_{k}$ and/or $K$. However, the performance gap between $X$-NOMA and $(X+1)$-NOMA is highly decreasing for $X \geq 4$. The latter performance gaps are highly increasing in $R^{\rm{min}}_{k}$ and $K$.

\subsection{Average Users Sum-Rate Performance}
The impact of minimum rate demands and number of users on the average SR of different multiple access techniques is shown in Fig. \ref{FigSR}. For the case that outage occurs, the SR is set to zero.
\begin{figure*}
	\centering
	\subfigure[Average sum-rate vs. users minimum rate demand for $K=10$.]{
		\includegraphics[scale=0.35]{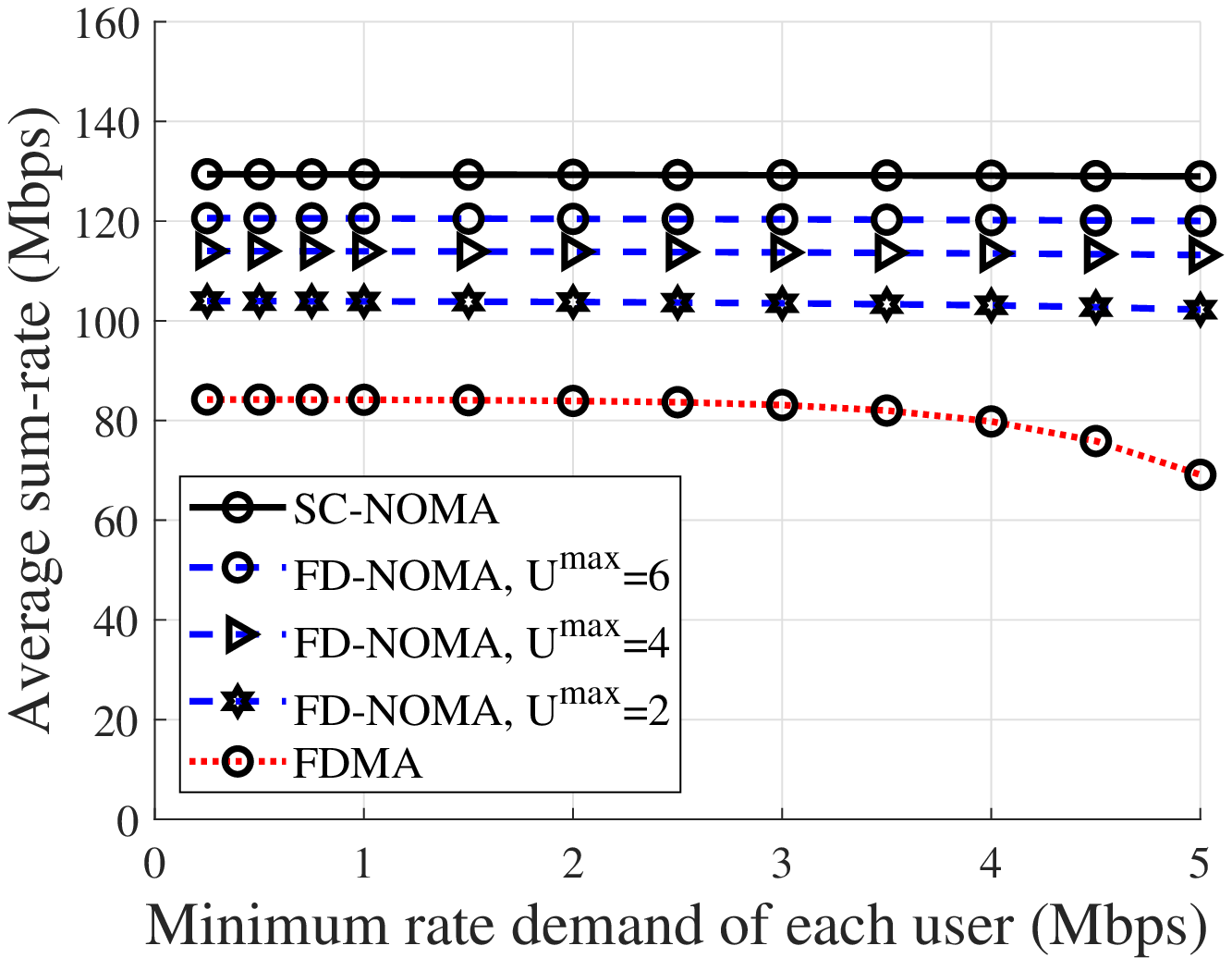}
		\label{Fig_SR_minrate1}
	}\hfill
	\subfigure[Average sum-rate vs. users minimum rate demand for $K=30$.]{
		\includegraphics[scale=0.35]{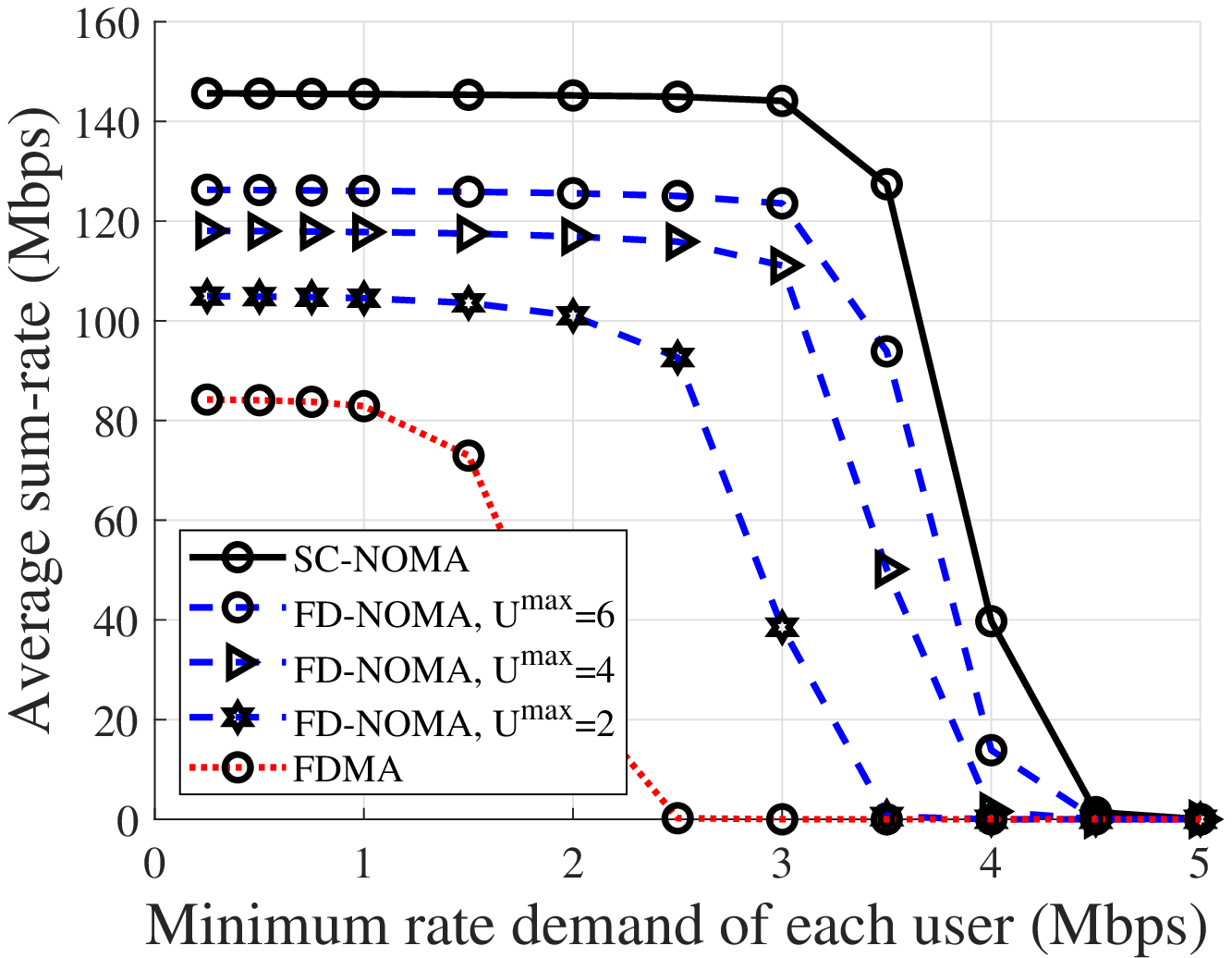}
		\label{Fig_SR_minrate2}
	}\hfill
	\subfigure[Average sum-rate vs. users minimum rate demand for $K=50$.]{
		\includegraphics[scale=0.35]{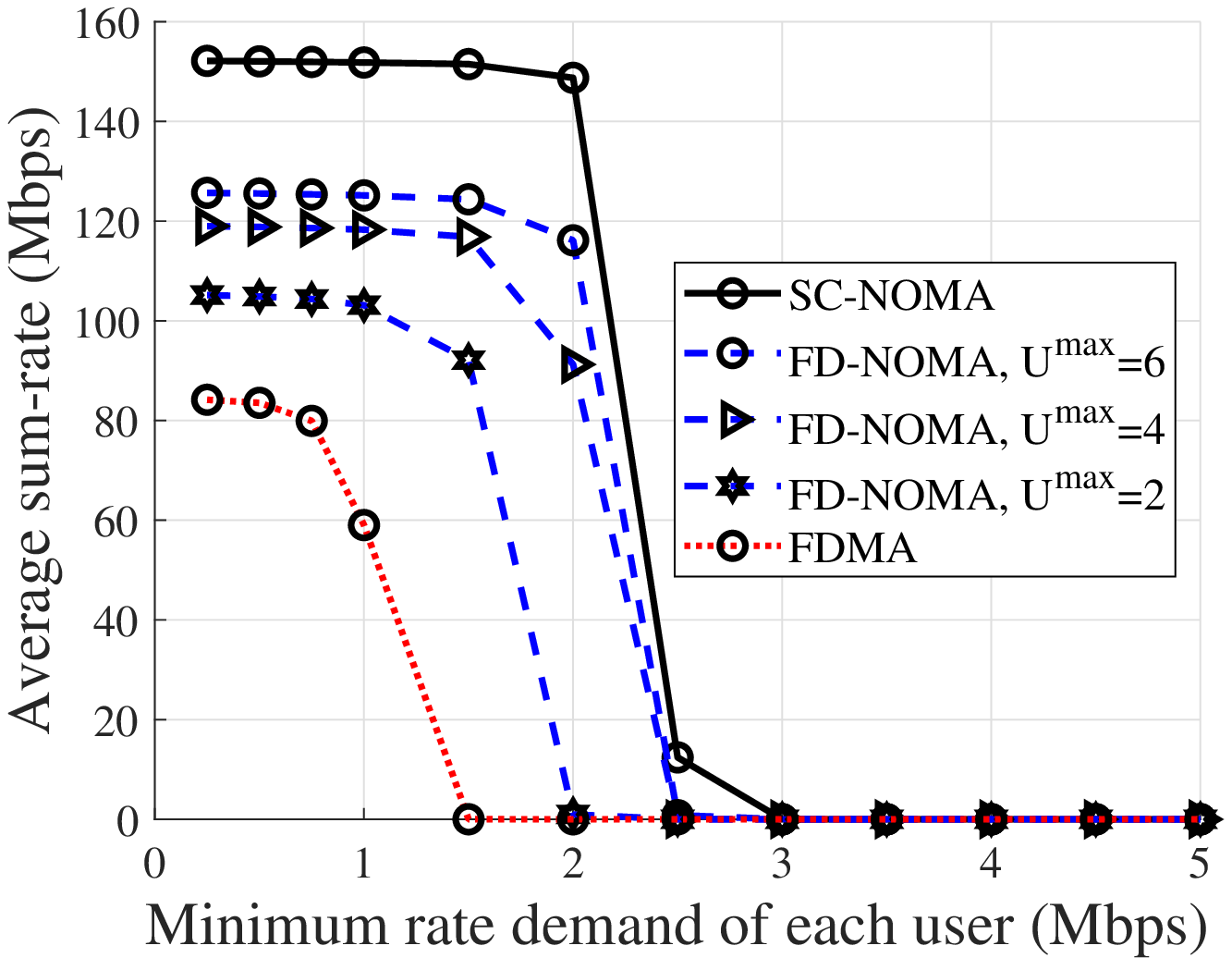}
		\label{Fig_SR_minrate3}
	}\hfill
	\subfigure[Average sum-rate vs. number of users for $R^{\rm{min}}_{k}=1.5~\text{Mbps},\forall k \in \mathcal{K}$.]{
		\includegraphics[scale=0.35]{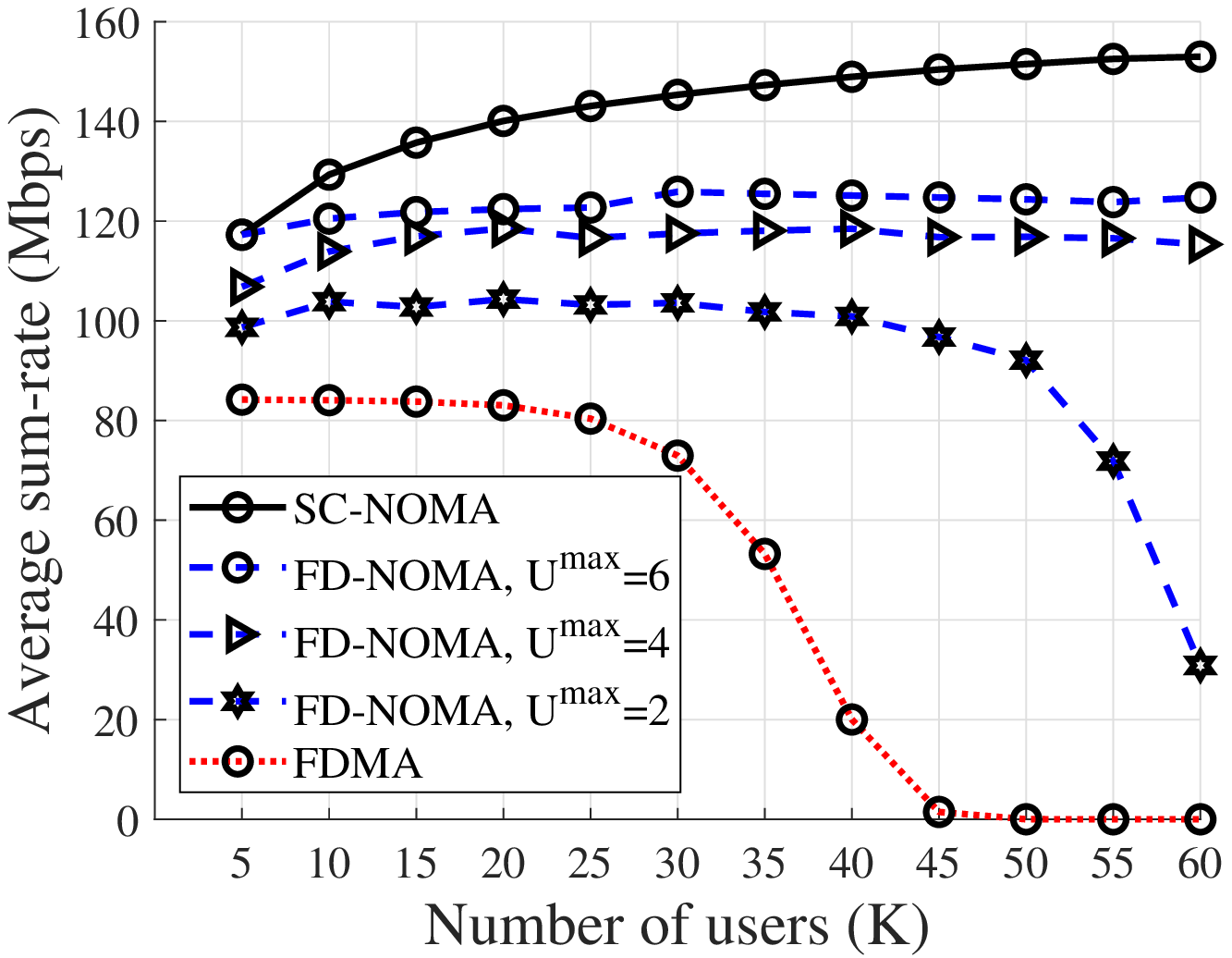}
		\label{Fig_SR_usernum1}
	}\hfill
	\subfigure[Average sum-rate vs. number of users for $R^{\rm{min}}_{k}=3~\text{Mbps},\forall k \in \mathcal{K}$.]{
		\includegraphics[scale=0.35]{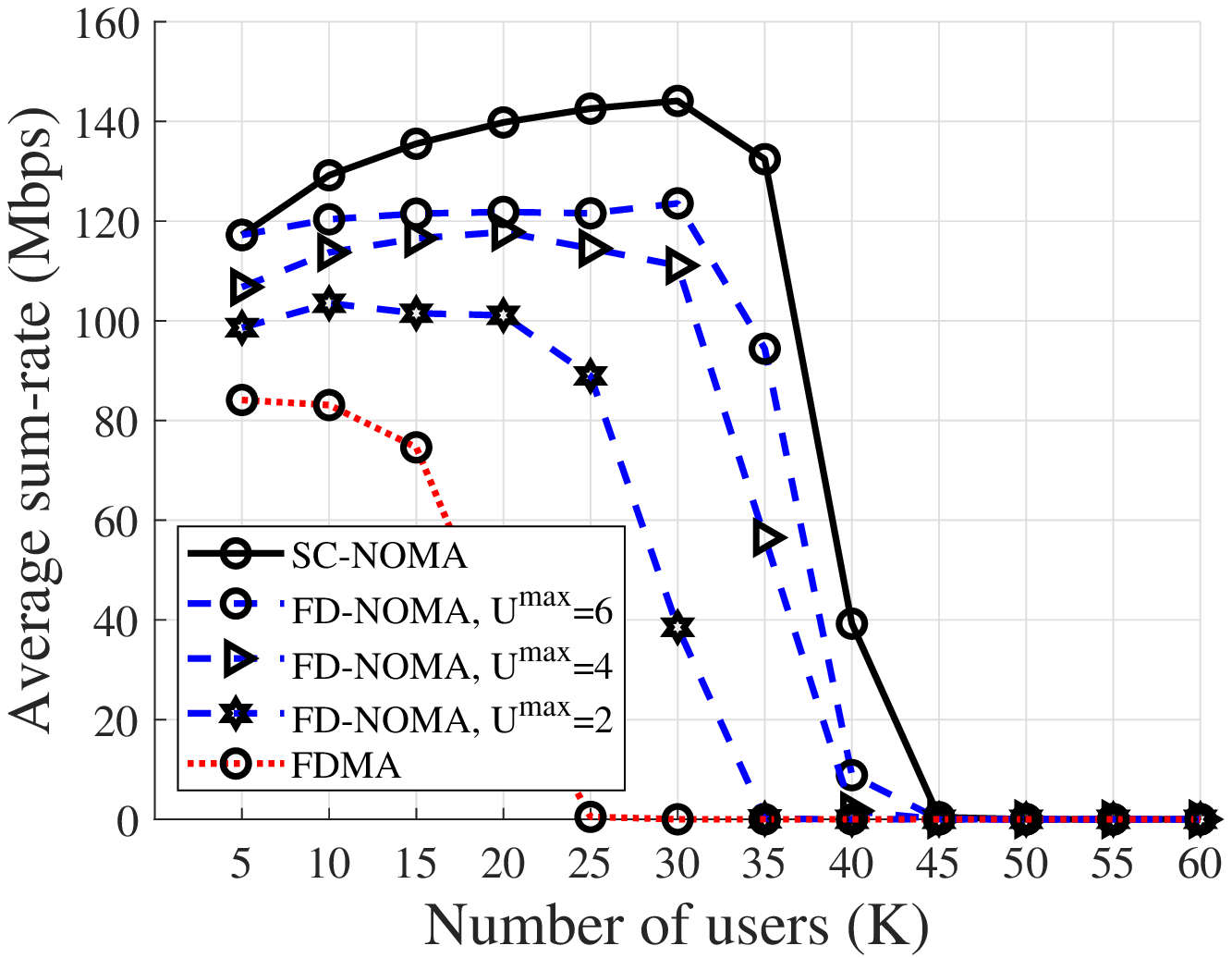}
		\label{Fig_SR_usernum2}
	}\hfill
	\subfigure[Average sum-rate vs. number of users for $R^{\rm{min}}_{k}=4.5~\text{Mbps},\forall k \in \mathcal{K}$.]{
		\includegraphics[scale=0.35]{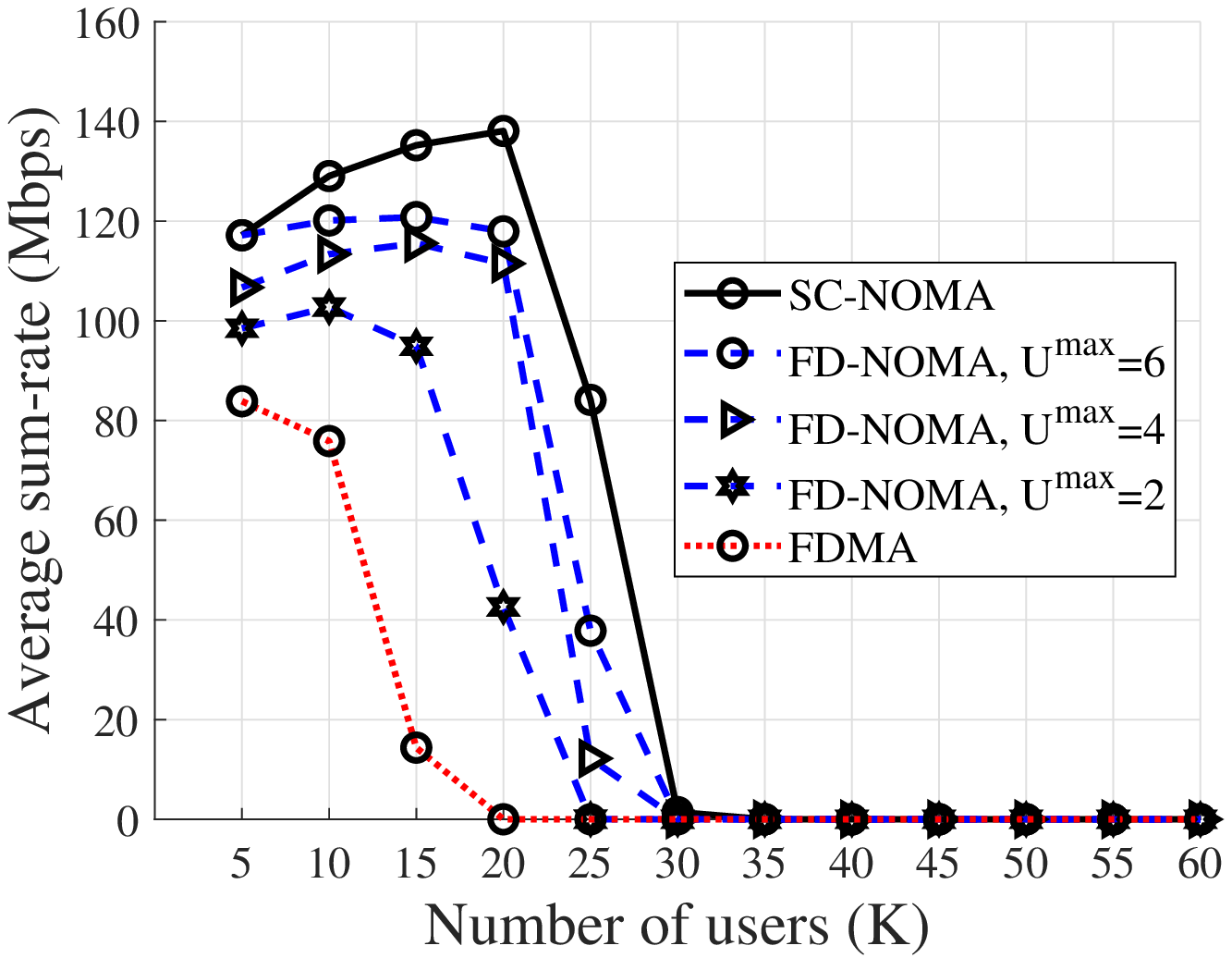}
		\label{Fig_SR_usernum3}
	}\hfill
	\caption
	{Impact of the minimum rate demand and number of users on the average sum-rate of SC-NOMA, FD-NOMA, and FDMA.}
	\label{FigSR}
\end{figure*}
The results in Figs. \ref{Fig_SR_minrate1}-\ref{Fig_SR_minrate3} show that the SR of users is highly insensitive to the minimum rate demands when $R^{\rm{min}}_{k}$ and $K$ are significantly low, specifically for SC-NOMA and FD-NOMA. For significantly high $R^{\rm{min}}_{k}$ and/or $K$, we observe that the average SR decreases, due to increasing the outage probability shown in Fig. \ref{Fig_outage_minrate1}-\ref{Fig_outage_minrate3}. Besides, Figs. \ref{Fig_SR_usernum1}-\ref{Fig_SR_usernum3} show that SC-NOMA well exploits the multiuser diversity, specifically for lower $R^{\rm{min}}_{k}$. In summary, the SR follows: $\text{SR}(\text{SC-NOMA}) > \text{SR}(6\text{-NOMA}) \approx \text{SR}(4\text{-NOMA}) > \text{SR}(2\text{-NOMA}) \gg \text{SR}(\text{FDMA})$.

\subsection{Average System Energy Efficiency Performance}
The impact of minimum rate demands and number of users on the average system EE of different multiple access techniques is shown in Fig. \ref{FigEE}.
\begin{figure*}
	\centering
	\subfigure[Average EE vs. users minimum rate demand for $K=10$.]{
		\includegraphics[scale=0.35]{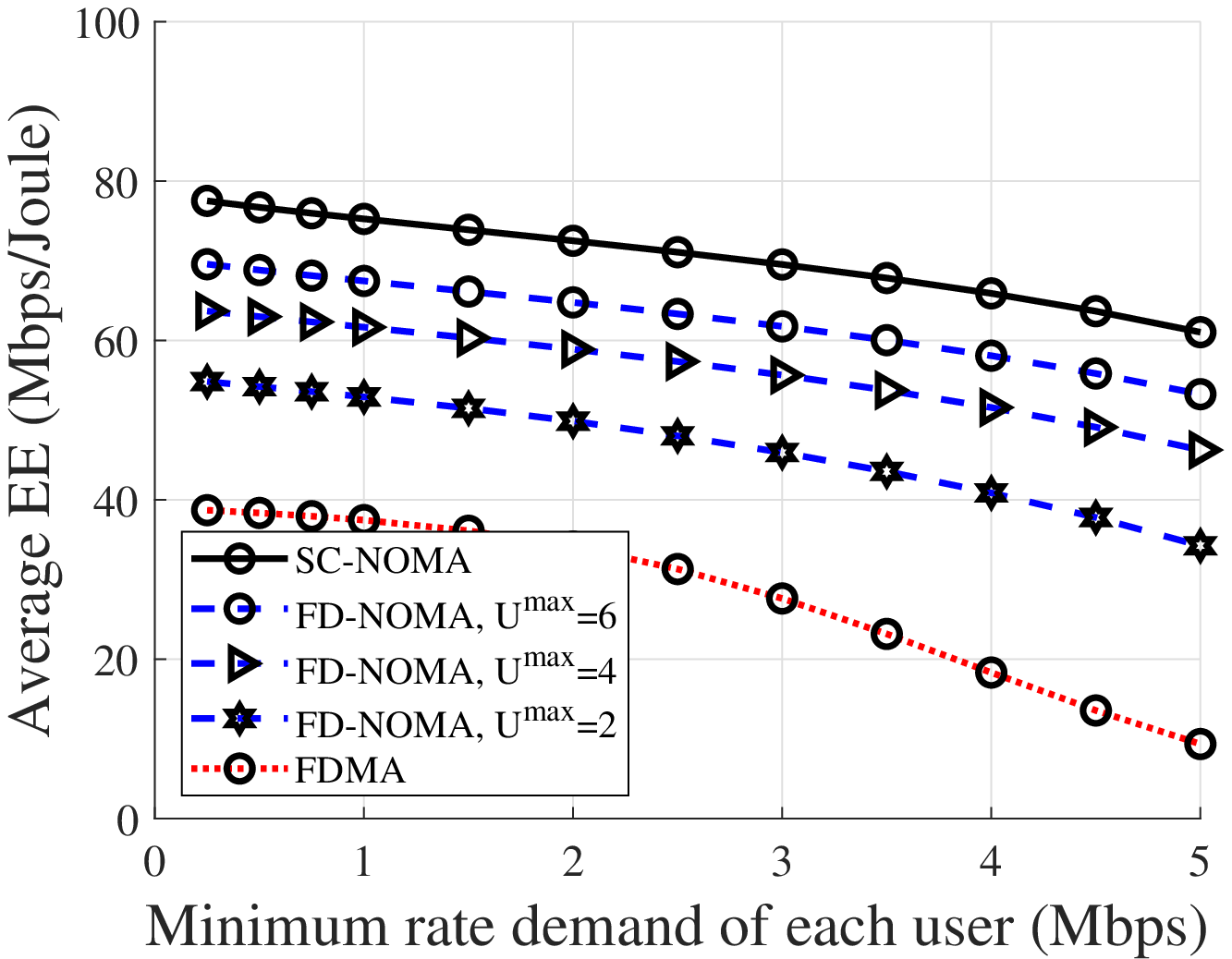}
		\label{Fig_EESYSTEM_minrate1}
	}\hfill
	\subfigure[Average EE vs. users minimum rate demand for $K=30$.]{
		\includegraphics[scale=0.35]{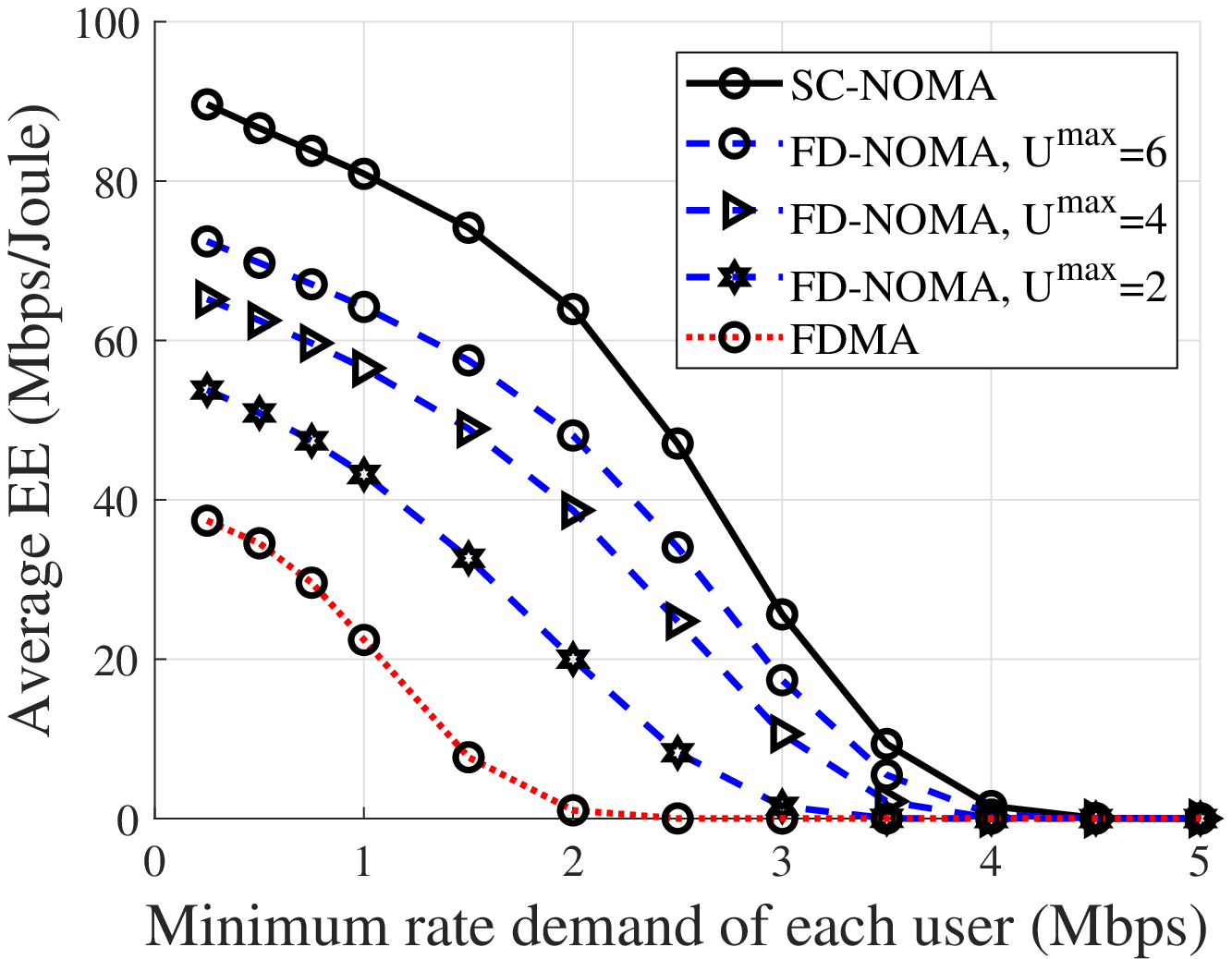}
		\label{Fig_EESYSTEM_minrate2}
	}\hfill
	\subfigure[Average EE vs. users minimum rate demand for $K=50$.]{
		\includegraphics[scale=0.35]{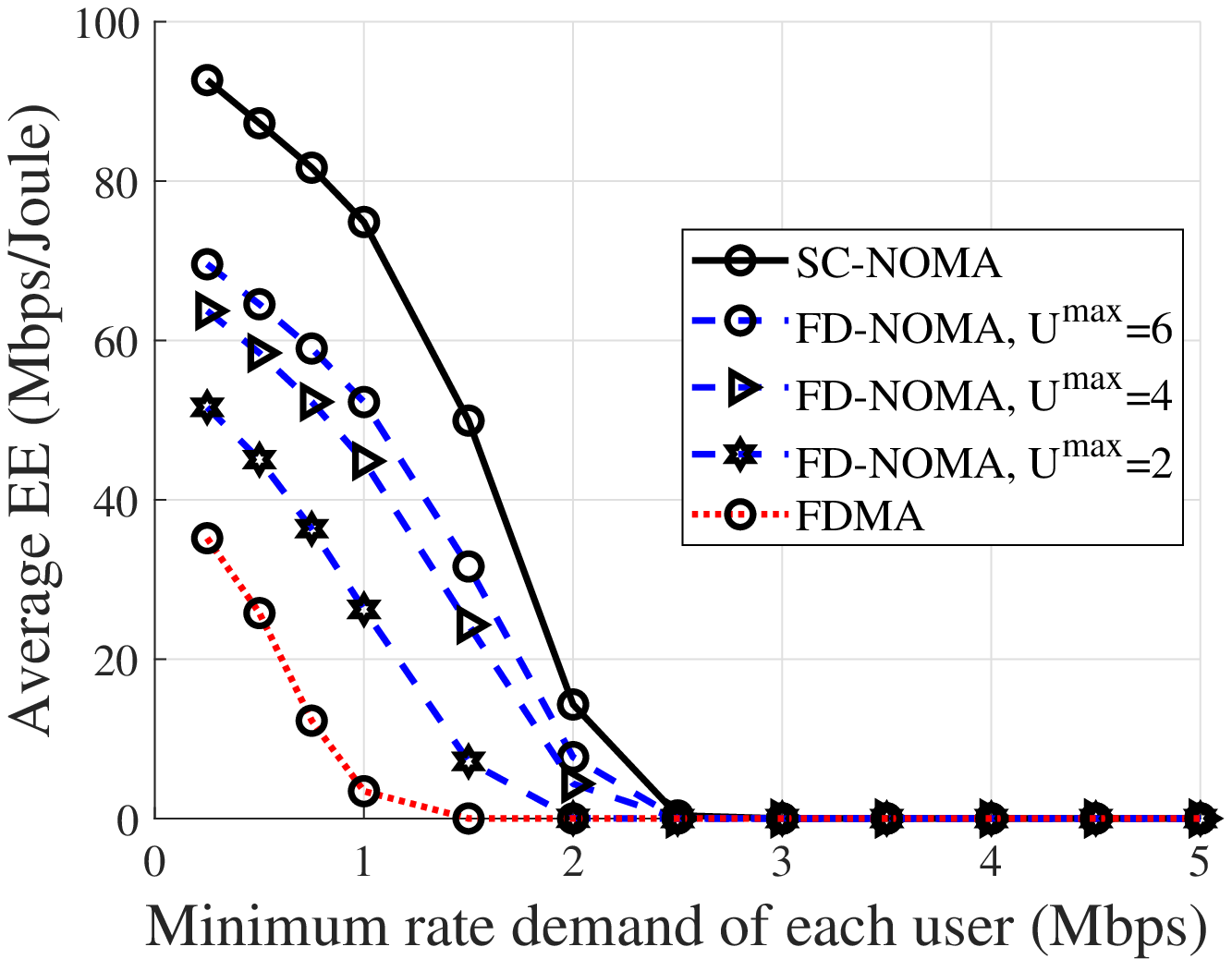}
		\label{Fig_EESYSTEM_minrate3}
	}\hfill
	\subfigure[Average EE vs. number of users for $R^{\rm{min}}_{k}=1.5~\text{Mbps},\forall k \in \mathcal{K}$.]{
		\includegraphics[scale=0.35]{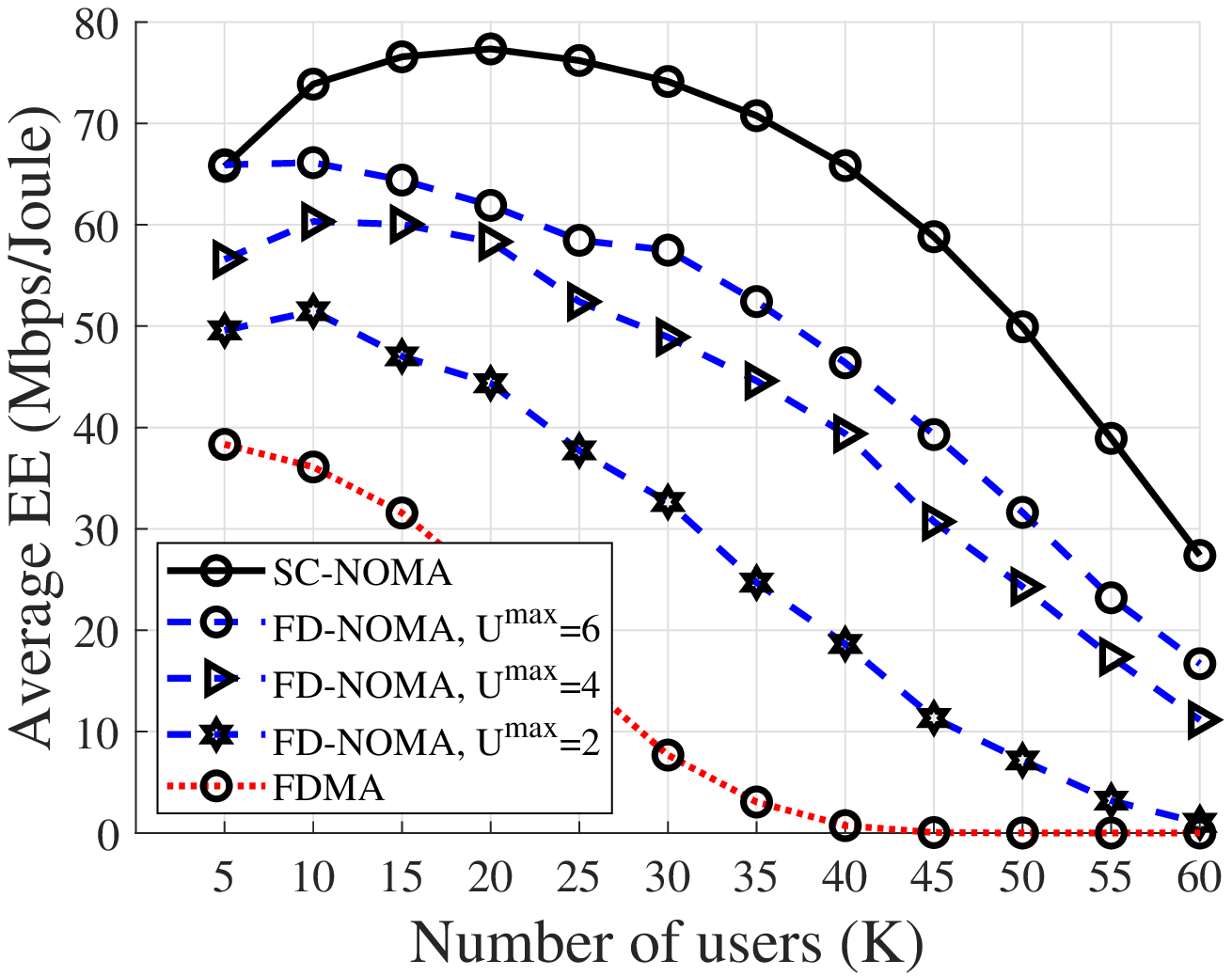}
		\label{Fig_EESYSTEM_usernum1}
	}\hfill
	\subfigure[Average EE vs. number of users for $R^{\rm{min}}_{k}=3~\text{Mbps},\forall k \in \mathcal{K}$.]{
		\includegraphics[scale=0.35]{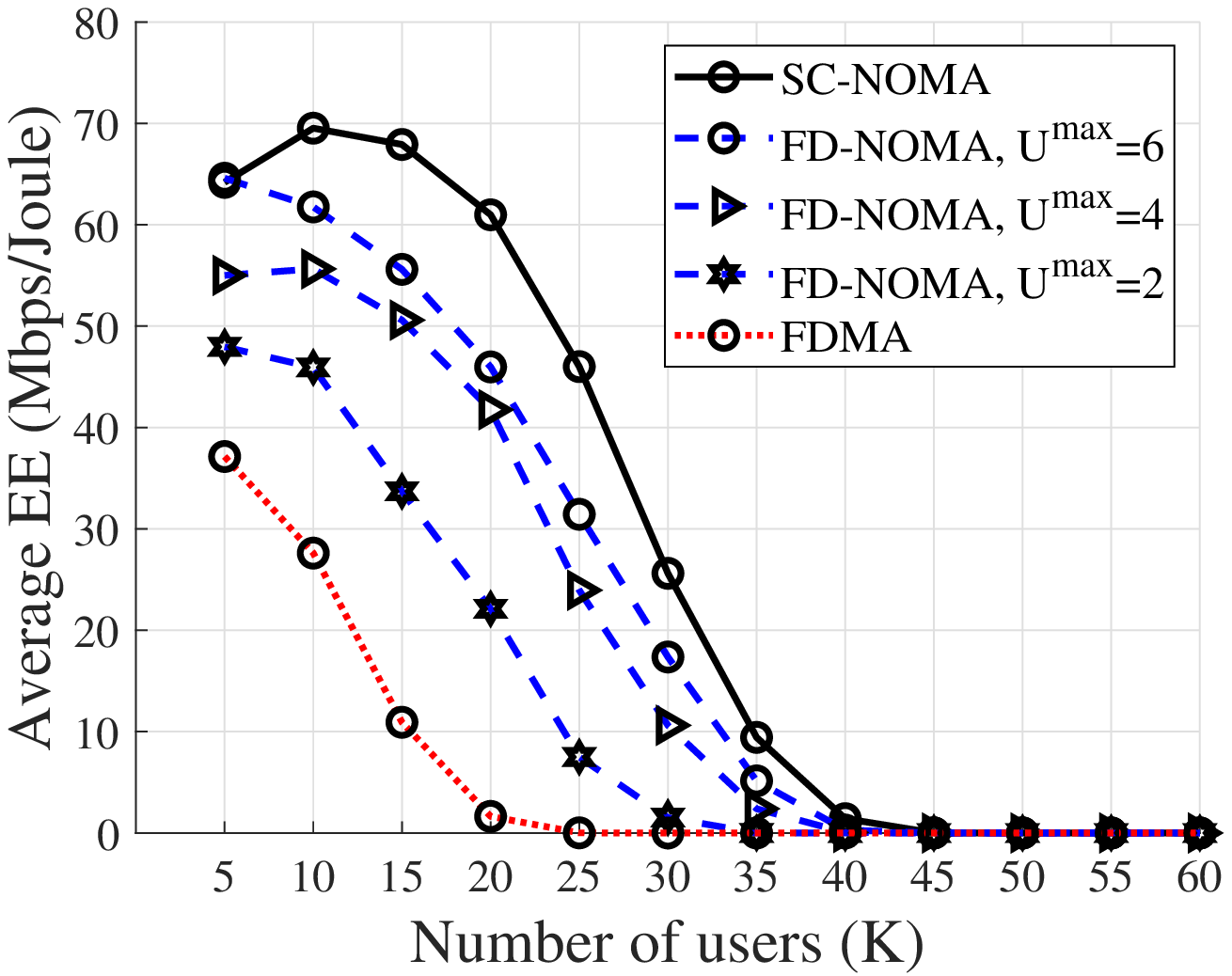}
		\label{Fig_EESYSTEM_usernum2}
	}\hfill
	\subfigure[Average EE vs. number of users for $R^{\rm{min}}_{k}=4.5~\text{Mbps},\forall k \in \mathcal{K}$.]{
		\includegraphics[scale=0.35]{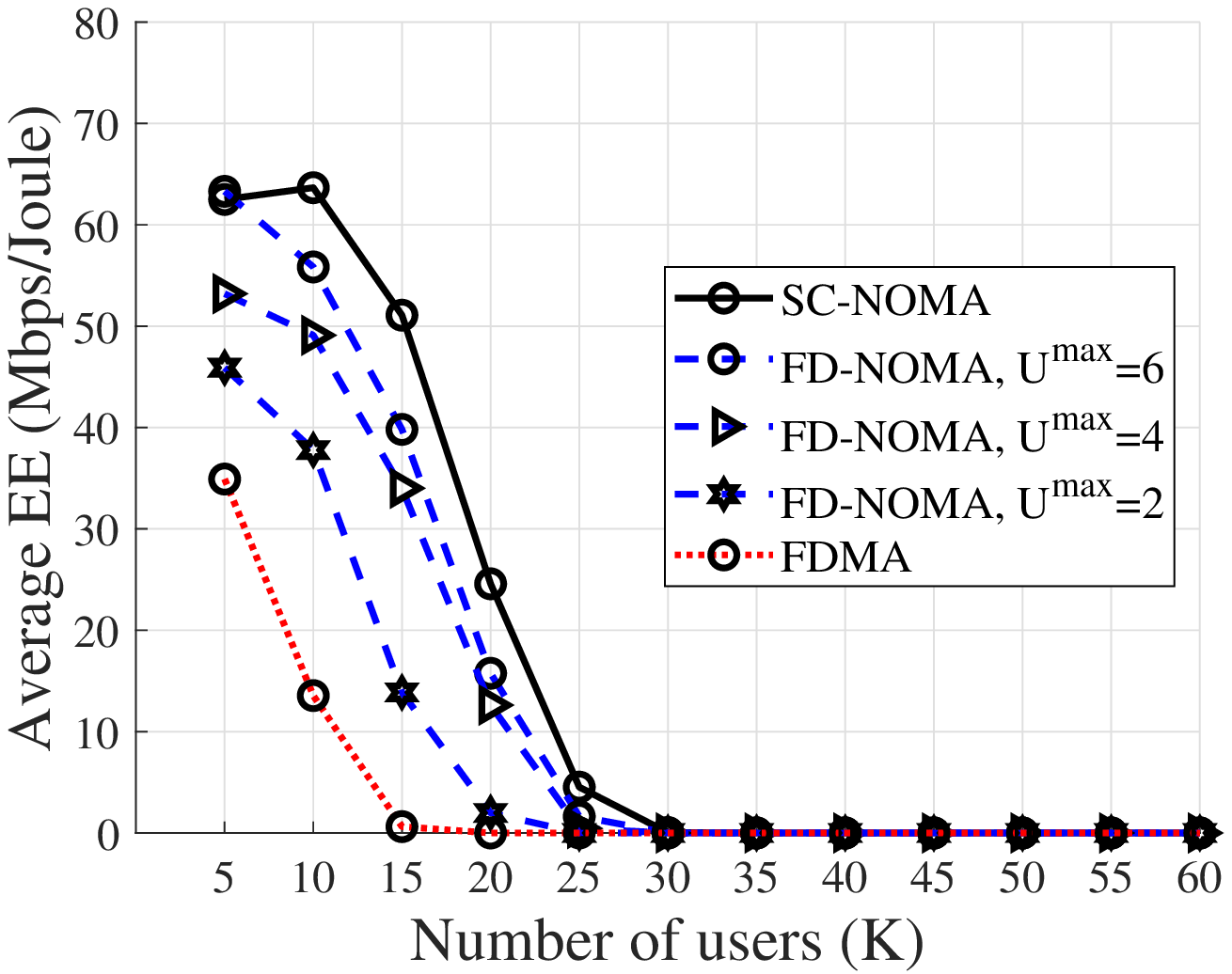}
		\label{Fig_EESYSTEM_usernum3}
	}\hfill
	\caption
	{Impact of the minimum rate demand and number of users on the average system EE of SC-NOMA, FD-NOMA, and FDMA.}
	\label{FigEE}
\end{figure*}
From Figs. \ref{Fig_EESYSTEM_minrate1}-\ref{Fig_EESYSTEM_minrate3}, we observe that the system EE is affected by $R^{\rm{min}}_{k}$ although the users SR are approximately insensitive to $R^{\rm{min}}_{k}$ shown in Figs. \ref{Fig_SR_minrate1}-\ref{Fig_SR_minrate3}. The main reason that EE is more affected by $R^{\rm{min}}_{k}$ compared to SR is the high sensitivity level of total power consumption to $R^{\rm{min}}_{k}$ shown in Figs. \ref{Fig_PM_minrate1}-\ref{Fig_PM_minrate3}. The impact of total power consumption on EE is highly affected by the circuit power consumption. It can be shown that when $P_{\rm{C}}$ increases, the system EE will be more insensitive to $R^{\rm{min}}_{k}$. From Figs. \ref{Fig_EESYSTEM_usernum1}-\ref{Fig_EESYSTEM_usernum3}, we observe that the system EE under minimum rate demands is increasing with $K$, when $K$ is small enough. In this situation, the system exploits the multiuser diversity, specifically for SC-NOMA. For significantly large $K$, the EE is decreasing with $K$ due to the existing minimum rate demands which highly affects the total power consumption. Following the results of Figs. \ref{FigPM} and \ref{FigSR}, the average system EE follows: $\text{EE}(\text{SC-NOMA}) > \text{EE}(6\text{-NOMA}) \approx \text{EE}(4\text{-NOMA}) > \text{EE}(2\text{-NOMA}) \gg \text{EE}(\text{FDMA})$.

\section{Concluding Remarks}\label{Section conclusion}
In this work, we addressed the problem of finding globally optimal power allocation algorithms to minimize the BSs power consumption, and maximize SR/EE of the general multiuser downlink single-cell Hybrid-NOMA systems. For these objectives, we showed that Hybrid-NOMA with $N$ clusters can be equivalently transformed to $N$-user virtual FDMA system, where the effective CNR of each virtual OMA user is obtained in closed form. In this transformation, we exploited the closed-form of optimal powers among multiplexed users within each cluster to further reduce the dimension of our problem as well as increase the accuracy of the iterative convex solvers. In particular, we showed that the feasible region of power allocation in NOMA can be defined as the intersection of closed boxes along with cellular power constraint. Then, we proposed a fast water-filling algorithm for the SR maximization problem, as well as fast iterative algorithms for the EE maximization problem based on the Dinkelbach algorithm with inner Lagrange dual with subgradient method/barrier algorithm with inner Newton's method. The complexity of our proposed algorithms are also analyzed.
The possible extensions of our analysis to more general cases with their corresponding new challenges are discussed in the paper.
Numerical assessments show that there exist a considerable performance gap in terms of outage probability, BSs power consumption, users SR, and system EE between FDMA and $2$-NOMA as well as between $2$-NOMA and $4$-NOMA. 
Moreover, we observed that the performance gaps between $X$-NOMA and $(X+1)$-NOMA highly decrease for $X \geq 4$, meaning that when $X \geq 4$, multiplexing more users merely improves the system performance.

\appendices

\section{Proof of Proposition \ref{proposition feasiblecluster}.}\label{appendix Lemma feasible}
The feasibility of \eqref{SCuser problem} can be determined by solving the power minimization problem as
\begin{equation}\label{min SCuser2 problem}
	\min_{\boldsymbol{p} \geq 0, \boldsymbol{q} \geq 0}~\sum\limits_{n \in \mathcal{N}} q_n~~~~~~\text{s.t.}~\eqref{Constraint minrate SCuser}\text{-}\eqref{Constraint mask q}.
\end{equation}
The problem \eqref{min SCuser2 problem} is also convex with affine objective function and constraints. Accordingly, the weak Slater's condition implies strong duality, thus \eqref{min SCuser2 problem} can be optimally solved by using the Lagrange dual method.
For SC-NOMA, in Appendix C of \cite{9583874}, we proved that the maximum power budget does not have any effect on the optimal powers obtained in the power minimization problem when the feasible region is nonempty. Accordingly, problem \eqref{min SCuser2 problem} can be decoupled into $N$ SC-NOMA power minimization subproblems when the feasible region of problem \eqref{min SCuser2 problem} is nonempty. The total power minimization of $M$-user downlink SC-NOMA is solved in Appendix C of \cite{9583874}. For convenience, consider cluster $n$ with $|\mathcal{K}_n|=K$ users whose CNRs are sorted as $h^n_1 < h^n_2 < \dots < h^n_K$ with optimal decoding order $K \to K-1 \to \dots \to 1$. As is proved in \cite{9583874}, at the optimal point $\boldsymbol{p^*}^n$, all the multiplexed users in $\mathcal{K}_n$ get power to only maintain their individual minimal rate demands, meaning that
\begin{equation*}\label{opt powmin SC}
	W_s \log_2\left( 1+ \frac{{p^n_k}^* h^n_k}{1+\sum\limits_{j=k+1}^{K} {p^n_j}^* h^n_k}  \right)= R^{\rm{min}}_{k,n},~\forall k=1,\dots,K.
\end{equation*}
The optimal power of each user $k \in \mathcal{K}_n$ (in Watts) can thus be obtained by
\begin{equation*}\label{opt totminpow SC}
	{p^n_k}^*=T^n_k \left( 1 + \sum\limits_{j=k+1}^{K} {p^n_j}^* h^n_k \right),~\forall k=1,\dots,K,
\end{equation*}
where $T^n_k=\frac{2^{(R^{\rm{min}}_{k,n}/W_s)} -1}{h^n_k},~\forall k=1,\dots,K$.
Let us rewrite the latter equation as
\begin{equation*}
	{p^n_k}^*=\beta^n_k \left(\frac{1}{h^n_k} + \sum\limits_{j=k+1}^{K} {p^n_j}^* \right),~\forall k=1,\dots,K,
\end{equation*}
where $\beta^n_k=2^{(R^{\rm{min}}_{k,n}/W_s)} -1,~\forall k=1,\dots,K$. The latter equation can be rewritten as
\begin{align*}
	{p^n_k}^*&=\beta^n_k \left(\frac{1}{h^n_k} + \sum\limits_{j=k+1}^{K} {p^n_j}^* \right)
	\\
	&=\beta^n_k \left(\frac{1}{h^n_k} + {p^n_{k+1}}^* + \sum\limits_{j=k+2}^{K} {p^n_j}^* \right)
	\\
	&=\beta^n_k \left(\frac{1}{h^n_k} + \beta^n_{k+1} \left(\frac{1}{h^n_{k+1}} + \sum\limits_{j=k+2}^{K} {p^n_j}^* \right) + \sum\limits_{j=k+2}^{K} {p^n_j}^* \right)
	\\
	&=\beta^n_k \left((1+\beta^n_{k+1}) \sum\limits_{j=k+2}^{K} {p^n_j}^* + \frac{1}{h^n_k} + \frac{\beta^n_{k+1}}{h^n_{k+1}} \right)
	\\
	&=\beta^n_k \left((1+\beta^n_{k+1}) \left({p^n_{k+2}}^* + \sum\limits_{j=k+3}^{K} {p^n_j}^*\right) + \frac{1}{h^n_k} + \frac{\beta^n_{k+1}}{h^n_{k+1}} \right)
	\\
	&=\beta^n_k \Bigg((1+\beta^n_{k+1}) \left( \beta^n_{k+2} \left(\frac{1}{h^n_{k+2}} + \sum\limits_{j=k+3}^{K} {p^n_j}^* \right) + \sum\limits_{j=k+3}^{K} {p^n_j}^*\right) 
	\\
	&~~~~+\frac{1}{h^n_k} + \frac{\beta^n_{k+1}}{h^n_{k+1}} \Bigg)
	\\
	&=\beta^n_k \Bigg((1+\beta^n_{k+1}) (1+\beta^n_{k+2}) \sum\limits_{j=k+3}^{K} {p^n_j}^* + \frac{1}{h^n_k} + \frac{\beta^n_{k+1}}{h^n_{k+1}} \\
	&~~~~+ \frac{\beta^n_{k+2}(1+\beta^n_{k+1})}{h^n_{k+2}} \Bigg)
	\\
	&\vdots
	\\
	&=\beta^n_k \bigg( (1+\beta^n_{k+1}) (1+\beta^n_{k+2})\dots (1+\beta^n_{K}) + \frac{1}{h^n_k} + \frac{\beta^n_{k+1}}{h^n_{k+1}}  
	\\
	&~~+ \frac{\beta^n_{k+2}(1+\beta^n_{k+1})}{h^n_{k+2}} + \dots + \frac{\beta^n_{K}(1+\beta^n_{K-1})\dots (1+\beta^n_{k+1})}{h^n_{K}} \bigg).
\end{align*}
As a result, we have
\begin{multline*}
	{p^n_k}^*=\beta^n_{k} \left(\prod\limits_{j=k+1}^{K} \left(1+\beta^n_{j}\right) +\frac{1}{h^n_k}+\sum\limits_{j=k+1}^{K} \frac{ \beta^n_{j} \prod\limits_{i=k+1}^{j-1} \left(1+\beta^n_{i}\right)}{h^n_j}\right), \\
	\forall k=1,\dots,K.
\end{multline*}
The latter equation can be rewritten as
\begin{multline}\label{optimal minpow Scell}
	\hspace{-0.3cm}{p^n_{k}}^*=\beta^n_{k} \left(\prod\limits_{j \in \mathcal{K}_n \atop h^n_{j} > h^n_{k}} \left(1+\beta^n_{j}\right) + \frac{1}{h_k}
	+ \sum\limits_{j \in \mathcal{K}_n \atop h^n_{j} > h^n_{k}} \frac{ \beta^n_{j} \prod\limits_{l \in \mathcal{K}_n \atop h^n_{k} < h^n_{l} < h^n_{j}} \left(1+\beta^n_{l}\right)}{h_j}\right) \\ \forall n \in \mathcal{N},~k \in \mathcal{K}_n.
\end{multline}
Therefore, the minimum power consumption of cluster $n$ is given by
\begin{multline}\label{optimal Qmin}
	Q^{\rm{min}}_n=\sum\limits_{k \in \mathcal{K}_n} {p^n_{k}}^*= \sum\limits_{k \in \mathcal{K}_n} \beta^n_{k} \Bigg(\prod\limits_{j \in \mathcal{K}_n \atop h^n_{j} > h^n_{k}} \left(1+\beta^n_{j}\right) +\frac{1}{h_k}+ \\
	\sum\limits_{j \in \mathcal{K}_n \atop h^n_{j} > h^n_{k}} \frac{ \beta^n_{j} \prod\limits_{l \in \mathcal{K}_n \atop h^n_{k} < h^n_{l} < h^n_{j}} \left(1+\beta^n_{l}\right)}{h_j}\Bigg), \forall n \in \mathcal{N}.
\end{multline}
The parameter $Q^{\rm{min}}_n$ is indeed the minimum power consumption of cluster $n$ to satisfy the minimum rate demand of users in $\mathcal{K}_n$. As a result, the feasible region of $q_n$ is lower-bounded by $Q^{\rm{min}}_n$ for each $n \in \mathcal{N}$. Accordingly, constraints \eqref{Constraint minrate SCuser} and \eqref{Constraint mask q} can be combined as $q_n \in \left[Q^{\rm{min}}_n,P^{\rm{mask}}_n\right],~\forall n \in \mathcal{N}$, which is the intersection of affine closed boxes, so is convex. According to \eqref{Constraint cell q}, we guarantee that any $q_n \in \left[Q^{\rm{min}}_n,P^{\rm{mask}}_n\right],~\forall n \in \mathcal{N}$, satisfying \eqref{Constraint cell q} is feasible, and the proof is completed.

\section{Proof of Proposition \ref{proposition jointcluster}.}\label{appendix closedform pow sum-rate}
For any given feasible $\boldsymbol{q}$, \eqref{Constraint cell q} and \eqref{Constraint mask q} can be removed from \eqref{SCuser problem}.
Then, problem \eqref{SCuser problem} can be equivalently divided into $N$ SC-NOMA subproblems. For each subproblem $n$, we find the intra-cluster power allocation among multiplexed users in $\mathcal{K}_n$ in closed-form.  According to the KKT optimality conditions analysis in Appendix B of \cite{9583874}, it is proved that at the optimal point of SC-NOMA with CNR-based decoding order, only the cluster-head user gets additional power, and all the other users get power to only maintain their minimal rate demands.

For convenience, consider cluster $n \in \mathcal{N}$ with $|\mathcal{K}_n|=K$ users whose CNRs are sorted as $h^n_1 < h^n_2 < \dots < h^n_K$ with optimal decoding order $K \to K-1 \to \dots \to 1$. According to Appendix B of \cite{9583874}, the optimal powers in $\boldsymbol{p^*}^n$ satisfy
\begin{equation}\label{opt poweri sumrate}
	\log_2\left( 1+\frac{{p^n_k}^* h^n_k}{1+\sum\limits_{j=k+1}^{K} {p^n_j}^* h^n_k} \right)= R^{\rm{min}}_{k,n},~\forall k=1,\dots,K-1,
\end{equation}
and
\begin{equation}\label{opt powerM sumrate}
	{p^n_K}^*=q_n-\sum\limits_{k=1}^{K-1} {p^n_k}^*.
\end{equation}
Let us rewrite \eqref{opt poweri sumrate} as
\begin{equation*}
	{p^n_k}^*=\frac{T^n_k \left( 1 + \left( q_n-\sum\limits_{j=1}^{k-1} {p^n_j}^*\right) h^n_k \right)}{1+T^n_k h^n_k},~\forall k=1,\dots,K-1,
\end{equation*}
where $T^n_k=\frac{2^{R^{\rm{min}}_{k,n}} -1}{h^n_k},~\forall k=1,\dots,K-1$.
To find a closed-form expression for $\boldsymbol{p^*}^n$, we rewire the latter equation as
\begin{equation*}
	{p^n_k}^*=\beta^n_{k} \left(q_n - \sum\limits_{j=1}^{k-1} {p^n_j}^* + \frac{1}{h^n_k} \right),~\forall k=1,\dots,K-1,
\end{equation*}
where $\beta^n_{k}=\frac{2^{R^{\rm{min}}_{k,n}} -1}{2^{R^{\rm{min}}_{k,n}}},~\forall k=1,\dots,K-1$.
The latter equation can also be rewritten as
\begin{align*}
	{p^n_k}^*&=\beta^n_{k} \left(q_n - {p^n_{k-1}}^* - \sum\limits_{j=1}^{k-2} {p^n_j}^* + \frac{1}{h^n_k}  \right)
	\\
	&
	\\
	&=\beta^n_{k} \bigg(q_n - \beta^n_{k-1} \left(q_n - \sum\limits_{j=1}^{k-2} {p^n_j}^* + \frac{1}{h^n_{k-1}} \right) - \sum\limits_{j=1}^{k-2} {p^n_j}^* + \frac{1}{h^n_k}  \bigg)
	\\
	&=\beta^n_{k} \bigg( \left(1-\beta^n_{k-1}\right) q_n - \left(1-\beta^n_{k-1}\right) 
	\sum\limits_{j=1}^{k-2} {p^n_j}^* + \frac{1}{h^n_k} - \frac{\beta^n_{k-1}}{h^n_{k-1}} \bigg)
	\\
	&\vdots
	\\
	=&\beta^n_{k} \bigg( \left(1-\beta^n_{k-1}\right) \left(1-\beta^n_{k-2}\right)\dots \left(1-\beta^n_{1}\right) q_n + \frac{1}{h^n_k} - \frac{\beta^n_{k-1}}{h^n_{k-1}} + \dots
	\\
	&- \frac{\left(1-\beta^n_{k-1}\right) \beta^n_{k-2}}{h^n_2}-\frac{\left(1-\beta^n_{k-1}\right) \left(1-\beta^n_{k-2}\right) \dots \left(1-\beta^n_{2}\right) \beta^n_{1}}{h^n_1} \bigg).
\end{align*}
According to the above, we have
\begin{multline}\label{power i closed sumrate series}
	{p^n_k}^*=\beta^n_{k} \left( \prod\limits_{j=1}^{k-1} \left(1-\beta^n_{j}\right) q_n
	+\frac{1}{h^n_k}-
	\sum\limits_{j=1}^{k-1} \frac{\beta^n_{j} \prod\limits_{i=j+1}^{k-1} \left(1-\beta^n_{i}\right)} {h^n_j} \right), \\
	\forall k=1,\dots,K-1.
\end{multline}
The optimal powers in \eqref{power i closed sumrate series} can be reformulated as
\begin{multline*}
	{p^n_{k}}^*=\left(\beta^n_{k} \prod\limits_{j \in \mathcal{K}_n \atop h^n_{j} < h^n_{k}} \left(1-\beta^n_{j}\right)\right) q_n + c^n_{k},~\forall n \in \mathcal{N},~ k \in \mathcal{K}_n\setminus\{\Phi_n\},
\end{multline*}
where $\beta^n_{k} = \frac{2^{(R^{\rm{min}}_{k,n}/W_s)}-1} {2^{(R^{\rm{min}}_{k,n}/W_s)}},~ \forall n \in \mathcal{N},~k \in \mathcal{K}_n$, 
$c^n_{k}=\beta^n_{k} \left( \frac{1}{h^n_{k}} - \sum\limits_{j \in \mathcal{K}_n \atop h^n_{j} < h^n_{k}} \frac{\prod\limits_{l \in \mathcal{K}_n \atop h^n_{j} < h^n_{l} < h^n_{k}} \left(1-\beta^n_{l}\right) \beta^n_{j}} {h^n_{j}}\right),~ \forall n \in \mathcal{N},~k \in \mathcal{K}_n$. Subsequently, based on \eqref{opt powerM sumrate}, the optimal power of the cluster-head users can be formulated by
\begin{multline*}
	{p^n_{\Phi_n}}^*= \left(1 - \sum\limits_{i \in \mathcal{K}_n \atop h^n_{i} < h^n_{\Phi_n}} \beta^n_{i} \prod\limits_{j \in \mathcal{K}_n \atop h^n_{j} < h^n_{i}} \left(1-\beta^n_{j}\right)\right) q_n - \sum\limits_{i \in \mathcal{K}_n \atop h^n_{i} < h^n_{\Phi_n}} c^n_{i},\\
	\forall n \in \mathcal{N}.
\end{multline*}

\section{Water-Filling Algorithm for Solving \eqref{eq1 problem}.}\label{appendix bisection sumrate}
The Lagrange function of \eqref{eq1 problem} is given by
\begin{equation}
	L(\boldsymbol{\tilde{q}},\nu)=\sum\limits_{n \in \mathcal{N}} W_s \log_2\left( 1 + \tilde{q}_n H_n \right) + \nu \left( \tilde{P}^{\rm{max}} - \sum\limits_{n \in \mathcal{N}} \tilde{q}_n \right),
\end{equation}
where $\nu$ is the Lagrange multiplier for the cellular power constraint \eqref{Constraint 1}, and $q_n \in [\tilde{Q}^{\rm{min}}_n,\tilde{P}^{\rm{mask}}_n],~\forall n \in \mathcal{N}$.
The Lagrange dual function is
\begin{multline}
	g(\nu)=
	\sup\limits_{\boldsymbol{\tilde{q}} \in \mathcal{P}} L(\boldsymbol{\tilde{q}},\nu)=
	\\
	\sup\limits_{\boldsymbol{\tilde{q}} \in \mathcal{P}} \Bigg\{ \sum\limits_{n \in \mathcal{N}} W_s \log_2\left( 1 + \tilde{q}_n H_n \right)	+ \nu \left( \tilde{P}^{\rm{max}} - \sum\limits_{n \in \mathcal{N}} \tilde{q}_n \right)
	\Bigg\},
\end{multline}
where $\mathcal{P}$ is the feasible set of problem \eqref{eq1 problem}. The Lagrange dual problem is formulated by
\begin{equation}\label{ppadual}
	\min\limits_{\nu}~g(\nu),~~~~~~\textrm{s.t.}~\nu \in \mathbb{R}.
\end{equation}
Assume that $\nu^*$ is the dual optimal. Moreover, $\boldsymbol{\tilde{q}}^*=[\tilde{q}^*_n], \forall n \in \mathcal{N}$, is primal. The KKT conditions are listed below
\begin{multline*}
	\text{C1}: q_n \in [\tilde{Q}^{\rm{min}}_n,\tilde{P}^{\rm{mask}}_n],~\forall n \in \mathcal{N},~~~~
	\text{C2}:\tilde{P}^{\rm{max}} - \sum\limits_{n \in \mathcal{N}} \tilde{q}^*_n = 0,\\
	\text{C3}:\nabla_{\boldsymbol{\tilde{q}}^*} L(\boldsymbol{\tilde{q}}^*,\nu^*) = 0.
\end{multline*}
Condition C3 can be rewritten as
$\frac{W_s H_n/\ln 2}{1 + \tilde{q}^*_n H_n} - \nu^* = 0,~\forall n \in \mathcal{N}$. 
Summing-up, for each $n \in \mathcal{N}$, we have
\begin{equation}\label{bisection opt}
	\tilde{q}^*_n=
	\begin{cases}
		\frac{W_s/(\ln 2)}{\nu^*} - \frac{1}{H_n}, &\quad \left(\frac{W_s/(\ln 2)}{\nu^*} - \frac{1}{H_n}\right) \in [\tilde{Q}^{\rm{min}}_n,\tilde{P}^{\rm{mask}}_n];\\
		0, &\quad \text{otherwise}. \\ 
	\end{cases}
\end{equation}
To ease of convenience, we reformulate \eqref{bisection opt} as
$\tilde{q}^*_n = \max \left\{\tilde{Q}^{\rm{min}}_n, \min \left\{ \left(\frac{W_s/(\ln 2)}{\nu^*} - \frac{1}{H_n}\right) , \tilde{P}^{\rm{mask}}_n\right\} \right\}$. 
By substituting $\tilde{q}^*_n$ to the cellular power constraint \eqref{Constraint 1}, we have 
\begin{equation}\label{bisection constraint}
	\sum\limits_{n \in \mathcal{N}} \max \left\{\tilde{Q}^{\rm{min}}_n, \min \left\{ \left(\frac{W_s/(\ln 2)}{\nu^*} - \frac{1}{H_n}\right) , \tilde{P}^{\rm{mask}}_n\right\} \right\} = \tilde{P}^{\rm{max}}.
\end{equation}
The left-hand side is a piecewise-linear increasing function of $\frac{W_s/(\ln 2)}{\nu^*}$ with breakpoints at $\frac{1}{H_n},~\forall n \in \mathcal{N}$, so the equation has a unique solution which is readily determined. To find optimal $\nu^*$, we first initialize tolerance $\epsilon$, lower-bound $\nu_l$ and upper-bound $\nu_h$. 
The lower-bound $\nu_l$ should satisfy $\sum\limits_{n \in \mathcal{N}} \max \left\{\tilde{Q}^{\rm{min}}_n, \min \left\{ \left(\frac{W_s/(\ln 2)}{\nu_l} - \frac{1}{H_n}\right) , \tilde{P}^{\rm{mask}}_n\right\} \right\} > \tilde{P}^{\rm{max}}$, and the upper-bound $\nu_h$ should satisfy $\sum\limits_{n \in \mathcal{N}} \max \left\{\tilde{Q}^{\rm{min}}_n, \min \left\{ \left(\frac{W_s/(\ln 2)}{\nu_h} - \frac{1}{H_n}\right) , \tilde{P}^{\rm{mask}}_n\right\} \right\} < \tilde{P}^{\rm{max}}$. 
After the initialization step, we apply the bisection method to find $\nu^*$ presented in Alg. \ref{Alg bisection}.

\section{Proof of the Optimality of Alg. \ref{Alg dinkelbach}.}\label{appendix converge Dinkelbach}
Let us define the EE function as $E(\boldsymbol{p})=\frac{f_1(\boldsymbol{p})}{f_2(\boldsymbol{p})},~\forall \boldsymbol{p} \in \mathcal{P}$, 
where $f_1(\boldsymbol{p})=\sum\limits_{n \in \mathcal{N}} \sum\limits_{k \in \mathcal{K}_n} R^n_{k} (\boldsymbol{p}^n)$, $f_2(\boldsymbol{p})=\sum\limits_{n \in \mathcal{N}} \sum\limits_{k \in \mathcal{K}_n} p^n_{k} + P_{\rm{C}}$, and $\mathcal{P}$ denotes the feasible set of problem \eqref{EE problem}. In this formulation, $f_1(\boldsymbol{p})$ is concave, and $f_2(\boldsymbol{p})$ is affine, so is convex. Moreover, both $f_1$ and $f_2$ are differentiable. The feasible set $\mathcal{P}$ can be characterized by using Proposition \ref{proposition feasiblecluster} which is shown to be affine, so is convex.  
For any non-empty $\mathcal{P}$ (which can be determined by Corollary \ref{corol feasible region}), the objective function $E(\boldsymbol{p})=\frac{f_1(\boldsymbol{p})}{f_2(\boldsymbol{p})}$ is pseudoconcave, implying that any stationary point is indeed global maximum and the KKT conditions are sufficient if a constraint qualification is fulfilled \cite{6213038,Jorswieckfractional}. Therefore, the globally optimal solution of problem \eqref{EE problem} can be obtained by using convex optimization algorithms \cite{6213038,Jorswieckfractional}. In particular, \eqref{EE problem} can be equivalently transformed to the following problem as
$$
\max_{ \boldsymbol{p} \in \mathcal{P}, \lambda \in \mathbb{R}}~\lambda
~~~~~~~~~\text{s.t.}~\frac{f_1(\boldsymbol{p})}{f_2(\boldsymbol{p})} - \lambda \geq 0,
$$
which can be rewritten as
$$
\max_{ \boldsymbol{p} \in \mathcal{P}, \lambda \in \mathbb{R}}~\lambda
~~~~~~~~~\text{s.t.}~f_1(\boldsymbol{p}) - \lambda f_2(\boldsymbol{p}) \geq 0.
$$
It can be proved that solving the latter problem is equivalent to finding the root of the following nonlinear function \cite{6213038}
$$
F(\lambda)=\max\limits_{ \boldsymbol{p} \in \mathcal{P}} f_1(\boldsymbol{p}) - \lambda f_2(\boldsymbol{p}),
$$
so the condition for the global optimality is
$$
F(\lambda^*)=\max\limits_{ \boldsymbol{p} \in \mathcal{P}} f_1(\boldsymbol{p}) - \lambda^* f_2(\boldsymbol{p}) = 0.
$$
Various methods can find the root of $F(\lambda)$, such as the Dinkelbach algorithm \cite{Dinkelbach} which is based on the Newton’s method.
For more details, please see Proposition 3.2 in \cite{Jorswieckfractional}.

\section{Lagrange Dual with Subgradient Method for Solving \eqref{EEeq problem} .}\label{appendix subgradient}
The Lagrange function of \eqref{EEeq problem} is formulated by
\begin{multline}
	L(\boldsymbol{\tilde{q}},\nu)=\sum\limits_{n \in \mathcal{N}} W_s \log_2\left( 1 + \tilde{q}_n H_n \right) - \lambda \left(\sum\limits_{n \in \mathcal{N}} \tilde{q}_n\right) + \\
	\nu \left( \tilde{P}^{\rm{max}} - \sum\limits_{n \in \mathcal{N}} \tilde{q}_n \right),
\end{multline}
where $\nu$ is the Lagrange multiplier for the cellular power constraint \eqref{ConstraintEEeq 1}, and $q_n \in [\tilde{Q}^{\rm{min}}_n,\tilde{P}^{\rm{mask}}_n],~\forall n \in \mathcal{N}$.
The dual function is given by
\begin{multline}
	g(\nu)=
	\sup\limits_{\boldsymbol{\tilde{q}} \in \mathcal{P}} L(\boldsymbol{\tilde{q}},\nu)=\sup\limits_{\boldsymbol{\tilde{q}} \in \mathcal{P}} \Bigg\{ \sum\limits_{n \in \mathcal{N}} W_s \log_2\left( 1 + \tilde{q}_n H_n \right) -
	\\
	\lambda \left(\sum\limits_{n \in \mathcal{N}} \tilde{q}_n\right)	+ \nu \left( \tilde{P}^{\rm{max}} - \sum\limits_{n \in \mathcal{N}} \tilde{q}_n \right)
	\Bigg\},
\end{multline}
where $\mathcal{P}$ is the feasible domain of problem \eqref{eq1 problem}. The Lagrange dual problem is formulated by
\begin{equation}\label{dualEE}
	\min\limits_{\nu}~g(\nu),~~~~~~\textrm{s.t.}~\nu \in \mathbb{R}.
\end{equation}
The optimal $\boldsymbol{\tilde{q}}^*$ can be obtained by $\nabla_{\boldsymbol{\tilde{q}}} L(\boldsymbol{\tilde{q}},\nu) = 0$. Then, we have
\begin{equation}\label{optqsubgrad}
	\tilde{q}^*_n= \left[\frac{W_s/(\ln 2) }{\lambda + \nu^*} - \frac{1}{H_n}\right]_{\tilde{Q}^{\rm{min}}_n}^{\tilde{P}^{\rm{mask}}_n},~n \in \mathcal{N},
\end{equation}
where $\nu^*$ is the dual optimal, which can be obtained by using the subgradient method \cite{Boydconvex}. In this algorithm, we iteratively update $\nu$ such that at iteration $(t+1)$ 
\begin{equation}\label{eqsubg}
	\nu^{(t+1)}=\left[\nu^{(t)} - \epsilon_\text{s} \left( \tilde{P}^{\rm{max}} - \sum\limits_{n \in \mathcal{N}} \tilde{q}^{(t)}_n \right)\right]^+,
\end{equation}
where $\nu^{(t)}$ is the Lagrange multiplier $\nu$ at iteration $t$, and $\tilde{q}^{(t)}_n$ is the optimal solution obtained by \eqref{optqsubgrad} at iteration $t$. Moreover, $\epsilon_\text{s} > 0$ is the step size tuning the accuracy of the algorithm \cite{nonlprgrmbook}. The iterations are repeated until the convergence is achieved. It is verified that the subgradient method will converge to the globally optimal solution after few iterations \cite{nonlprgrmbook}.

\section{Barrier Algorithm with Inner Newton's Method for Solving \eqref{EEeq problem}.}\label{appendix barrier}
Let us reformulate \eqref{EEeq problem} as the following standard convex problem 
\begin{subequations}\label{EEeq problem eq1} 
	\begin{align}\label{obf EEeq problem eq1}
		\min_{\boldsymbol{\tilde{q}}}
		~~ & f_0 (\boldsymbol{\tilde{q}}) = - \sum\limits_{n \in \mathcal{N}} W_s \log_2\left( 1 + \tilde{q}_n H_n \right) + \lambda
		\left(\sum\limits_{n \in \mathcal{N}} \tilde{q}_n\right)
		\\
		\text{s.t.}~~
		\label{Constraint eq1}
		& f_1 (\boldsymbol{\tilde{q}}) = \sum\limits_{n \in \mathcal{N}} \tilde{q}_n - \tilde{P}^{\rm{max}} \leq 0,
		\\
		\label{Constraint eq2}
		& q_n \in [\tilde{Q}^{\rm{min}}_n,\tilde{P}^{\rm{mask}}_n],~\forall n \in \mathcal{N}.
	\end{align}
\end{subequations}
Then, we approximate \eqref{EEeq problem eq1} to an unconstrained minimization problem as 
\begin{equation}\label{EEeq problem eq2}
	\min_{\boldsymbol{\tilde{q}}}~~
	U(\boldsymbol{\tilde{q}})= t f_0(\boldsymbol{\tilde{q}}) + \phi(\boldsymbol{\tilde{q}}),
\end{equation}
in which $$\phi(\boldsymbol{\tilde{q}}) = - \log\left( - f_1(\boldsymbol{\tilde{q}}) \right),$$
such that the domain of $\phi$ is
$$
\textbf{dom}~\phi=\{ q_n \in [\tilde{Q}^{\rm{min}}_n,\tilde{P}^{\rm{mask}}_n],~\forall n \in \mathcal{N} | f_1(\boldsymbol{\tilde{q}}) < 0 \},$$
and $t \gg 1$ is a positive real constant. The problem \eqref{EEeq problem eq2} is convex since $t f_0(\boldsymbol{\tilde{q}})$ and $\phi(\boldsymbol{\tilde{q}})$ are convex. In each barrier iteration, we solve \eqref{EEeq problem eq2} by using the Newton's method. 
The gradient of $U(\boldsymbol{\boldsymbol{\tilde{q}}})$ is formulated by
$$\nabla_{\boldsymbol{\tilde{q}}} U(\boldsymbol{\tilde{q}})=t \nabla_{\boldsymbol{\tilde{q}}} f_0(\boldsymbol{\tilde{q}}) + \nabla_{\boldsymbol{\tilde{q}}} \phi(\boldsymbol{\tilde{q}}),$$
where $\nabla_{\boldsymbol{\tilde{q}}} f_0(\boldsymbol{\tilde{q}}) = \left[\frac{\partial f_0}{\partial {\tilde{q}}_n}\right],~\forall n \in \mathcal{N}$, in which 
$$
\frac{\partial f_0(\boldsymbol{\tilde{q}})}{\partial {\tilde{q}}_n} = - \frac{W_s H_n}{\ln(2) \left( 1+\tilde{q}_n H_n \right)} + \lambda,~\forall n \in \mathcal{N}.
$$
In addition, $\nabla_{\boldsymbol{\tilde{q}}} \phi(\boldsymbol{\tilde{q}}) = \left[\frac{\partial \phi(\boldsymbol{\tilde{q}})}{\partial {\tilde{q}}_n}\right],~\forall n \in \mathcal{N}$, in which 
$
\frac{\partial \phi(\boldsymbol{\tilde{q}})}{\partial \tilde{q}_n} = -\frac{\frac{\partial f_1(\boldsymbol{\tilde{q}})} {\partial \tilde{q}_n}}{f_1(\boldsymbol{\boldsymbol{\tilde{q}}})},~\forall n \in \mathcal{N},
$
such that $\frac{\partial f_1(\boldsymbol{\tilde{q}})} {\partial \tilde{q}_n}=1,~\forall n \in \mathcal{N}$. Therefore, we have
$
\frac{\partial \phi(\boldsymbol{\tilde{q}})}{\partial \tilde{q}_n} = - \frac{1}{\sum\limits_{n \in \mathcal{N}} \tilde{q}_n - \tilde{P}^{\rm{max}}}
$.
Summing up, the $n$-th element in the vector $\nabla_{\boldsymbol{\tilde{q}}} U(\boldsymbol{\tilde{q}}) = \left[\frac{\partial U(\boldsymbol{\tilde{q}})}{\partial {\tilde{q}}_n}\right],~\forall n \in \mathcal{N}$, is given by
$$
\frac{\partial U(\boldsymbol{\tilde{q}})}{\partial {\tilde{q}}_n} = - t \left( \frac{W_s H_n}{\ln(2) \left( 1+\tilde{q}_n H_n \right)} + \lambda\right) - \frac{1}{\sum\limits_{n \in \mathcal{N}} \tilde{q}_n - \tilde{P}^{\rm{max}}},~\forall n \in \mathcal{N}.
$$
The Hessian of $U(\boldsymbol{\tilde{q}})$ is formulated by
$$
\nabla^2_{\boldsymbol{\tilde{q}}} U(\boldsymbol{\tilde{q}})=t \nabla^2_{\boldsymbol{\tilde{q}}} f_0(\boldsymbol{\tilde{q}}) + \nabla^2_{\boldsymbol{\tilde{q}}} \phi(\boldsymbol{\tilde{q}}),
$$
where 
$$
\nabla^2_{\boldsymbol{\tilde{q}}} f_0(\boldsymbol{\tilde{q}}) = \begin{bmatrix} 
	\frac{\partial^2 f_0}{\partial \tilde{q}^2_1} & \frac{\partial^2 f_0}{\partial \tilde{q}_1 \partial \tilde{q}_2} & \dots & \frac{\partial^2 f_0}{\partial \tilde{q}_1 \partial \tilde{q}_{N}}
	\\ 
	\frac{\partial^2 f_0}{\partial \tilde{q}_2 \partial \tilde{q}_1} & \frac{\partial^2 f_0}{\partial \tilde{q}^2_2} & \dots & \frac{\partial^2 f_0}{\partial \tilde{q}_2 \partial \tilde{q}_{N}}
	\\ 
	\vdots & \vdots & \ddots & \vdots
	\\
	\frac{\partial^2 f_0}{\partial \tilde{q}_{N} \partial \tilde{q}_1} & \frac{\partial^2 f_0}{\partial \tilde{q}_{N} \partial \tilde{q}_2} & \dots & \frac{\partial^2 f_0}{\partial \tilde{q}^2_{N}}
\end{bmatrix},
$$
such that its entries are $\frac{\partial^2 f_0(\boldsymbol{\tilde{q}})}{\partial \tilde{q}^2_n} = \frac{W_s}{\ln(2)} \frac{H^2_n}{\left( 1+\tilde{q}_n H_n \right)^2},~\forall n \in \mathcal{N}$, and $\frac{\partial^2 f_0(\boldsymbol{\tilde{q}})}{\partial \tilde{q}_i \tilde{q}_j} = 0,~\forall i,j \in \mathcal{N},~i \neq j$. As a result, we have
$$\nabla^2_{\boldsymbol{\tilde{q}}} f_0(\boldsymbol{\tilde{q}})= \textbf{diag} \left(\left[\frac{\partial^2 f_0(\boldsymbol{\tilde{q}})}{\partial \tilde{q}^2_n}\right], ~\forall n \in \mathcal{N}\right),$$ 
which is positive definite, since each element $\frac{\partial^2 f_0}{\partial \tilde{q}^2_i}$ in the main diagonal of $\nabla^2_{\boldsymbol{\tilde{q}}} f_0(\boldsymbol{\tilde{q}})$ is positive and the others are zero. The Hessian of $\phi(\boldsymbol{\tilde{q}})$ can be obtained by
$$\nabla^2_{\boldsymbol{\tilde{q}}} \phi(\boldsymbol{\tilde{q}})=\frac{1}{\left(\sum\limits_{n \in \mathcal{N}} \tilde{q}_n - \tilde{P}^{\rm{max}}\right)^2} \boldsymbol{1}_{N\times N}.$$
The eigenvector of $\phi(\boldsymbol{\tilde{q}})$ is
$\left[0,0,\dots,0,\frac{1}{\left(\sum\limits_{n \in \mathcal{N}} \tilde{q}_n - \tilde{P}^{\rm{max}}\right)^2}\right]_{1 \times N}$. Then, we conclude that $\nabla^2_{\boldsymbol{\tilde{q}}} \phi(\boldsymbol{\tilde{q}}) \succeq 0$. Finally, we have $\nabla^2_{\boldsymbol{\tilde{q}}} U(\boldsymbol{\tilde{q}})=t \nabla^2_{\boldsymbol{\tilde{q}}} f_0(\boldsymbol{\tilde{q}}) + \nabla^2_{\boldsymbol{\tilde{q}}} \phi(\boldsymbol{\tilde{q}}) \succ 0$ since $\nabla^2_{\boldsymbol{\tilde{q}}} f_0(\boldsymbol{\tilde{q}}) \succ 0 $, $\nabla^2_{\boldsymbol{\tilde{q}}} \phi(\boldsymbol{\tilde{q}}) \succeq 0$, and $t>0$. The latter result proves that $U(\boldsymbol{\tilde{q}})$ is strictly convex and its Hessian is nonsingular. Accordingly, $\left(\nabla^2_{\boldsymbol{\tilde{q}}} U(\boldsymbol{\tilde{q}})\right)^{-1}$ is positive and finite.
In this regard, the barrier method with inner Newton's method (with backtracking line search) achieves an $\epsilon$-suboptimal solution \cite{Boydconvex}.

\section{Optimal Power Allocation for Maximizing Sum-Rate of $K$-User SC-NOMA with Per-User Minimum and Maximum Rate Constraints}\label{appendix max rate constraint}
Consider a $K$-user SC-NOMA system, whose users CNRs are sorted as $h_1 < h_2 < \dots < h_K$ with optimal decoding order $K \to K-1 \to \dots \to 1$.
Let $M_{k}=\left(2^{(R^{\rm{max}}_{k}/W_s)} - 1\right) \left(\sum\limits_{j=k+1}^{K} p_{j} h_{k} + 1\right)/h_{k},~\forall k=1,\dots,K$. It can be shown that the optimal powers can be obtained by first calculating the optimal powers in Proposition \ref{proposition jointcluster}. Then, we obtain $M_{K}$ of the strongest user. If $p^*_{K} \leq M_{K}$, the obtained powers are the optimal solution. If $p^*_{K} > M_{K}$, we set $p^*_{K}=M_{K}$. Then, we calculate $M_{K-1}$ and $p_{K-1}=P^{\rm{max}} - \left(\sum\limits_{j=1 \atop j \neq K-1}^{K} p_{j}\right)$ with the updated $p^*_{K}=M_{K}$. If $p^*_{K-1} \leq M_{K-1}$, the obtained powers are the optimal solution. If $p^*_{K-1} > M_{K-1}$, we set $p^*_{K-1}=M_{K-1}$. Then, we calculate $M_{K-2}$ and $p_{K-2}=P^{\rm{max}} - \left(\sum\limits_{j=1 \atop j \neq K-2}^{K} p_{j}\right)$ with the updated $p^*_{K}=M_{K}$, and $p^*_{K-1}=M_{K-1}$. If $p^*_{K-2} \leq M_{K-2}$, the obtained powers are the optimal solution. Otherwise, we continue these series until a user denoted by $i$ satisfies $p^*_{i} \leq M_{i}$. Accordingly, the achievable rate of users $K,\dots,1$ can be obtained as
$$
\underbrace{R^{\rm{max}}_{K},R^{\rm{max}}_{K-1},\dots,R^{\rm{max}}_{i+1}}_\text{Maximum Rates}, \underbrace{\left[R^{\rm{min}}_{i}, R^{\rm{max}}_{i}\right]}_\text{Between}, \underbrace{R^{\rm{min}}_{i-1},\dots,R^{\rm{min}}_{1}}_\text{Minimum Rates},
$$
respectively. 

%


\bibliographystyle{IEEEtran}
\bibliography{Bibliography}

\begin{thebibliography}{10}
\providecommand{\url}[1]{#1}
\csname url@samestyle\endcsname
\providecommand{\newblock}{\relax}
\providecommand{\bibinfo}[2]{#2}
\providecommand{\BIBentrySTDinterwordspacing}{\spaceskip=0pt\relax}
\providecommand{\BIBentryALTinterwordstretchfactor}{4}
\providecommand{\BIBentryALTinterwordspacing}{\spaceskip=\fontdimen2\font plus
\BIBentryALTinterwordstretchfactor\fontdimen3\font minus
  \fontdimen4\font\relax}
\providecommand{\BIBforeignlanguage}[2]{{%
\expandafter\ifx\csname l@#1\endcsname\relax
\typeout{** WARNING: IEEEtran.bst: No hyphenation pattern has been}%
\typeout{** loaded for the language `#1'. Using the pattern for}%
\typeout{** the default language instead.}%
\else
\language=\csname l@#1\endcsname
\fi
#2}}
\providecommand{\BIBdecl}{\relax}
\BIBdecl

\bibitem{PHDdiss}
D.~Hughes-Hartogs, \emph{The Capacity of a Degraded Spectral Gaussian Broadcast
  Channel}.\hskip 1em plus 0.5em minus 0.4em\relax Ph.D. dissertation, Inform.
  Syst. Lab., Ctr. Syst Res., Stanford Univ., Stanford, CA, July 1975.

\bibitem{10.5555/1146355}
T.~M. Cover and J.~A. Thomas, \emph{Elements of Information Theory (Wiley
  Series in Telecommunications and Signal Processing)}.\hskip 1em plus 0.5em
  minus 0.4em\relax USA: Wiley-Interscience, 2006.

\bibitem{NIFbook}
A.~E. Gamal and Y.-H. Kim, \emph{Network Information Theory}.\hskip 1em plus
  0.5em minus 0.4em\relax Cambridge University Press, 2011.

\bibitem{915665}
L.~Li and A.~Goldsmith, ``Capacity and optimal resource allocation for fading
  broadcast channels-{Part I: Ergodic capacity},'' \emph{IEEE Transactions on
  Information Theory}, vol.~47, no.~3, pp. 1083--1102, 2001.

\bibitem{1246013}
N.~Jindal and A.~Goldsmith, ``Capacity and optimal power allocation for fading
  broadcast channels with minimum rates,'' \emph{IEEE Transactions on
  Information Theory}, vol.~49, no.~11, pp. 2895--2909, 2003.

\bibitem{10.5555/3294318}
M.~Vaezi, Z.~Ding, and H.~V. Poor, \emph{Multiple Access Techniques for 5G
  Wireless Networks and Beyond}, 1st~ed.\hskip 1em plus 0.5em minus 0.4em\relax
  Springer Publishing Company, Incorporated, 2018.

\bibitem{6692652}
Y.~{Saito}, Y.~{Kishiyama}, A.~{Benjebbour}, T.~{Nakamura}, A.~{Li}, and
  K.~{Higuchi}, ``Non-orthogonal multiple access {(NOMA)} for cellular future
  radio access,'' in \emph{Proc. IEEE 77th Vehicular Technology Conference (VTC
  Spring)}, 2013, pp. 1--5.

\bibitem{8357810}
L.~{Dai}, B.~{Wang}, Z.~{Ding}, Z.~{Wang}, S.~{Chen}, and L.~{Hanzo}, ``A
  survey of non-orthogonal multiple access for {5G},'' \emph{IEEE
  Communications Surveys \& Tutorials}, vol.~20, no.~3, pp. 2294--2323, 2018.

\bibitem{7973146}
Z.~{Ding}, X.~{Lei}, G.~K. {Karagiannidis}, R.~{Schober}, J.~{Yuan}, and V.~K.
  {Bhargava}, ``A survey on non-orthogonal multiple access for {5G} networks:
  Research challenges and future trends,'' \emph{IEEE Journal on Selected Areas
  in Communications}, vol.~35, no.~10, pp. 2181--2195, 2017.

\bibitem{9311149}
J.~Du, W.~Liu, G.~Lu, J.~Jiang, D.~Zhai, F.~R. Yu, and Z.~Ding, ``When
  mobile-edge computing {(MEC)} meets nonorthogonal multiple access {(NOMA)}
  for the internet of things {(IoT)}: System design and optimization,''
  \emph{IEEE Internet of Things Journal}, vol.~8, no.~10, pp. 7849--7862, 2021.

\bibitem{3gppNOMA2015}
\emph{Study on Downlink Multiuser Superposition Transmission for {LTE}},
  document, 3rd Generation Partnership Project (3GPP), Mar. 2015.

\bibitem{8792153}
M.~Vaezi, G.~A. Aruma~Baduge, Y.~Liu, A.~Arafa, F.~Fang, and Z.~Ding,
  ``Interplay between {NOMA} and other emerging technologies: A survey,''
  \emph{IEEE Transactions on Cognitive Communications and Networking}, vol.~5,
  no.~4, pp. 900--919, 2019.

\bibitem{9382272}
S.~Rezvani, N.~Mokari, M.~R. Javan, and E.~Jorswieck, ``Resource allocation in
  virtualized {CoMP-NOMA} {HetNets}: Multi-connectivity for joint
  transmission,'' \emph{IEEE Transactions on Communications}, vol.~69, no.~6,
  pp. 4172--4185, 2021.

\bibitem{8456624}
M.~Moltafet, R.~Joda, N.~Mokari, M.~R. Sabagh, and M.~Zorzi, ``Joint access and
  fronthaul radio resource allocation in {PD-NOMA}-based {5G} networks enabling
  dual connectivity and {CoMP},'' \emph{IEEE Transactions on Communications},
  vol.~66, no.~12, pp. 6463--6477, 2018.

\bibitem{8758862}
S.~Rezvani, S.~Parsaeefard, N.~Mokari, M.~R. Javan, and H.~Yanikomeroglu,
  ``Cooperative multi-bitrate video caching and transcoding in multicarrier
  {NOMA}-assisted heterogeneous virtualized {MEC} networks,'' \emph{IEEE
  Access}, vol.~7, pp. 93\,511--93\,536, 2019.

\bibitem{9122596}
S.~Gong, X.~Lu, D.~T. Hoang, D.~Niyato, L.~Shu, D.~I. Kim, and Y.-C. Liang,
  ``Toward smart wireless communications via intelligent reflecting surfaces: A
  contemporary survey,'' \emph{IEEE Communications Surveys \& Tutorials},
  vol.~22, no.~4, pp. 2283--2314, 2020.

\bibitem{9475160}
C.~Pan, H.~Ren, K.~Wang, J.~F. Kolb, M.~Elkashlan, M.~Chen, M.~Di~Renzo,
  Y.~Hao, J.~Wang, A.~L. Swindlehurst, X.~You, and L.~Hanzo, ``Reconfigurable
  intelligent surfaces for {6G} systems: Principles, applications, and research
  directions,'' \emph{IEEE Communications Magazine}, vol.~59, no.~6, pp.
  14--20, 2021.

\bibitem{9365004}
X.~Mu, Y.~Liu, L.~Guo, J.~Lin, and N.~Al-Dhahir, ``Capacity and optimal
  resource allocation for {IRS}-assisted multi-user communication systems,''
  \emph{IEEE Transactions on Communications}, vol.~69, no.~6, pp. 3771--3786,
  2021.

\bibitem{9154358}
O.~{Maraqa}, A.~S. {Rajasekaran}, S.~{Al-Ahmadi}, H.~{Yanikomeroglu}, and S.~M.
  {Sait}, ``A survey of rate-optimal power domain {NOMA} with enabling
  technologies of future wireless networks,'' \emph{IEEE Communications Surveys
  \& Tutorials}, vol.~22, no.~4, pp. 2192--2235, 2020.

\bibitem{7263349}
L.~Dai, B.~Wang, Y.~Yuan, S.~Han, I.~Chih-lin, and Z.~Wang, ``Non-orthogonal
  multiple access for {5G}: solutions, challenges, opportunities, and future
  research trends,'' \emph{IEEE Communications Magazine}, vol.~53, no.~9, pp.
  74--81, 2015.

\bibitem{7676258}
S.~M.~R. {Islam}, N.~{Avazov}, O.~A. {Dobre}, and K.~{Kwak}, ``Power-domain
  non-orthogonal multiple access {(NOMA)} in {5G} systems: Potentials and
  challenges,'' \emph{IEEE Communications Surveys \& Tutorials}, vol.~19,
  no.~2, pp. 721--742, 2017.

\bibitem{8823873}
M.~{Vaezi}, R.~{Schober}, Z.~{Ding}, and H.~V. {Poor}, ``Non-orthogonal
  multiple access: Common myths and critical questions,'' \emph{IEEE Wireless
  Communications}, vol.~26, no.~5, pp. 174--180, 2019.

\bibitem{9583874}
S.~Rezvani, E.~A. Jorswieck, N.~Mokari, and M.~R. Javan, ``Optimal {SIC}
  ordering and power allocation in downlink multi-cell {NOMA} systems,''
  \emph{IEEE Transactions on Wireless Communications}, Oct. 2021.

\bibitem{Jorswieckfractional}
A.~Zappone and E.~Jorswieck, ``Energy efficiency in wireless networks via
  fractional programming theory,'' \emph{Foundations and Trends in
  Communications and Information Theory}, vol.~11, no. 3-4, pp. 185--396, 2015.

\bibitem{6213038}
C.~Isheden, Z.~Chong, E.~Jorswieck, and G.~Fettweis, ``Framework for link-level
  energy efficiency optimization with informed transmitter,'' \emph{IEEE
  Transactions on Wireless Communications}, vol.~11, no.~8, pp. 2946--2957,
  2012.

\bibitem{8861078}
W.~U. {Khan}, F.~{Jameel}, T.~{Ristaniemi}, S.~{Khan}, G.~A.~S. {Sidhu}, and
  J.~{Liu}, ``Joint spectral and energy efficiency optimization for downlink
  {NOMA} networks,'' \emph{IEEE Transactions on Cognitive Communications and
  Networking}, vol.~6, no.~2, pp. 645--656, 2020.

\bibitem{7587811}
L.~{Lei}, D.~{Yuan}, C.~K. {Ho}, and S.~{Sun}, ``Power and channel allocation
  for non-orthogonal multiple access in {5G} systems: Tractability and
  computation,'' \emph{IEEE Transactions on Wireless Communications}, vol.~15,
  no.~12, pp. 8580--8594, 2016.

\bibitem{8422362}
L.~{Salaün}, C.~S. {Chen}, and M.~{Coupechoux}, ``Optimal joint subcarrier and
  power allocation in {NOMA} is strongly {NP-Hard},'' in \emph{Proc. IEEE
  International Conference on Communications (ICC)}, 2018, pp. 1--7.

\bibitem{9044770}
L.~Salaün, M.~Coupechoux, and C.~S. Chen, ``Joint subcarrier and power
  allocation in {NOMA}: Optimal and approximate algorithms,'' \emph{IEEE
  Transactions on Signal Processing}, vol.~68, pp. 2215--2230, 2020.

\bibitem{7557079}
M.~S. {Ali}, H.~{Tabassum}, and E.~{Hossain}, ``Dynamic user clustering and
  power allocation for uplink and downlink non-orthogonal multiple access
  {(NOMA)} systems,'' \emph{IEEE Access}, vol.~4, pp. 6325--6343, 2016.

\bibitem{8114362}
Z.~{Yang}, C.~{Pan}, W.~{Xu}, Y.~{Pan}, M.~{Chen}, and M.~{Elkashlan}, ``Power
  control for multi-cell networks with non-orthogonal multiple access,''
  \emph{IEEE Transactions on Wireless Communications}, vol.~17, no.~2, pp.
  927--942, 2018.

\bibitem{7982784}
J.~{Zhu}, J.~{Wang}, Y.~{Huang}, S.~{He}, X.~{You}, and L.~{Yang}, ``On optimal
  power allocation for downlink non-orthogonal multiple access systems,''
  \emph{IEEE Journal on Selected Areas in Communications}, vol.~35, no.~12, pp.
  2744--2757, 2017.

\bibitem{7523951}
F.~Fang, H.~Zhang, J.~Cheng, and V.~C.~M. Leung, ``Energy-efficient resource
  allocation for downlink non-orthogonal multiple access network,'' \emph{IEEE
  Transactions on Communications}, vol.~64, no.~9, pp. 3722--3732, 2016.

\bibitem{9276828}
P.~Gupta and D.~Ghosh, ``User fairness based energy efficient power allocation
  for downlink cellular {NOMA} system,'' in \emph{Proc. International
  Conference on Computing, Communication and Security (ICCCS)}, 2020, pp. 1--5.

\bibitem{8119791}
F.~Fang, H.~Zhang, J.~Cheng, S.~Roy, and V.~C.~M. Leung, ``Joint user
  scheduling and power allocation optimization for energy-efficient {NOMA}
  systems with imperfect {CSI},'' \emph{IEEE Journal on Selected Areas in
  Communications}, vol.~35, no.~12, pp. 2874--2885, 2017.

\bibitem{8448840}
F.~Fang, J.~Cheng, Z.~Ding, and H.~V. Poor, ``Energy efficient resource
  optimization for a downlink {NOMA} heterogeneous small-cell network,'' in
  \emph{Proc. IEEE 10th Sensor Array and Multichannel Signal Processing
  Workshop (SAM)}, 2018, pp. 51--55.

\bibitem{9032198}
A.~J. Muhammed, Z.~Ma, Z.~Zhang, P.~Fan, and E.~G. Larsson, ``Energy-efficient
  resource allocation for {NOMA} based small cell networks with wireless
  backhauls,'' \emph{IEEE Transactions on Communications}, vol.~68, no.~6, pp.
  3766--3781, 2020.

\bibitem{7560605}
B.~{Di}, L.~{Song}, and Y.~{Li}, ``Sub-channel assignment, power allocation,
  and user scheduling for non-orthogonal multiple access networks,'' \emph{IEEE
  Transactions on Wireless Communications}, vol.~15, no.~11, pp. 7686--7698,
  2016.

\bibitem{7812683}
Y.~{Sun}, D.~W.~K. {Ng}, Z.~{Ding}, and R.~{Schober}, ``Optimal joint power and
  subcarrier allocation for full-duplex multicarrier non-orthogonal multiple
  access systems,'' \emph{IEEE Transactions on Communications}, vol.~65, no.~3,
  pp. 1077--1091, 2017.

\bibitem{7996797}
Y.~Fu, L.~Salaün, C.~W. Sung, C.~S. Chen, and M.~Coupechoux, ``Double
  iterative waterfilling for sum rate maximization in multicarrier {NOMA}
  systems,'' in \emph{Proc. IEEE International Conference on Communications
  (ICC)}, 2017, pp. 1--6.

\bibitem{8489876}
Y.~{Fu}, L.~{Salaün}, C.~W. {Sung}, and C.~S. {Chen}, ``Subcarrier and power
  allocation for the downlink of multicarrier {NOMA} systems,'' \emph{IEEE
  Transactions on Vehicular Technology}, vol.~67, no.~12, pp. 11\,833--11\,847,
  2018.

\bibitem{8737495}
L.~{Salaün}, M.~{Coupechoux}, and C.~S. {Chen}, ``Weighted sum-rate
  maximization in multi-carrier {NOMA} with cellular power constraint,'' in
  \emph{Proc. IEEE Conference on Computer Communications (INFOCOM)}, 2019, pp.
  451--459.

\bibitem{8885490}
E.~Carmona, H.~Zhu, J.~Wang, and O.~Alluhaibi, ``A fast algorithm for resource
  allocation in downlink multicarrier {NOMA},'' in \emph{Proc. IEEE Wireless
  Communications and Networking Conference (WCNC)}, 2019, pp. 1--5.

\bibitem{9264208}
A.~B.~M. Adam, X.~Wan, and Z.~Wang, ``Energy efficiency maximization in
  downlink multi-cell multi-carrier {NOMA} networks with hardware
  impairments,'' \emph{IEEE Access}, vol.~8, pp. 210\,054--210\,065, 2020.

\bibitem{8807992}
K.~Huang, Z.~Wang, H.~Zhang, Z.~Fan, X.~Wan, and Y.~Xu, ``Energy efficient
  resource allocation algorithm in multi-carrier noma systems,'' in \emph{Proc.
  IEEE 20th International Conference on High Performance Switching and Routing
  (HPSR)}, 2019, pp. 1--5.

\bibitem{sourcecode}
S.~Rezvani and E.~A. Jorswieck, ``Optimal power allocation in downlink
  multicarrier {NOMA} systems: {Theory} and fast algorithms,''
  \url{https://gitlab.com/sepehrrezvani/MC-NOMA-SR-EE.git}, Nov. 2021.

\bibitem{Boydconvex}
S.~Boyd and L.~Vandenberghe, \emph{Convex Optimization}.\hskip 1em plus 0.5em
  minus 0.4em\relax Cambridge University Press, 2009.

\bibitem{8352643}
M.~S. Ali, E.~Hossain, A.~Al-Dweik, and D.~I. Kim, ``Downlink power allocation
  for {CoMP-NOMA} in multi-cell networks,'' \emph{IEEE Transactions on
  Communications}, vol.~66, no.~9, pp. 3982--3998, Sept. 2018.

\bibitem{1381759}
D.~Palomar and J.~Fonollosa, ``Practical algorithms for a family of
  waterfilling solutions,'' \emph{IEEE Transactions on Signal Processing},
  vol.~53, no.~2, pp. 686--695, 2005.

\bibitem{1576943}
W.~Yu and J.~Cioffi, ``Constant-power waterfilling: performance bound and
  low-complexity implementation,'' \emph{IEEE Transactions on Communications},
  vol.~54, no.~1, pp. 23--28, 2006.

\bibitem{6547819}
P.~He, L.~Zhao, S.~Zhou, and Z.~Niu, ``Water-filling: A geometric approach and
  its application to solve generalized radio resource allocation problems,''
  \emph{IEEE Transactions on Wireless Communications}, vol.~12, no.~7, pp.
  3637--3647, 2013.

\bibitem{10.1155/2008/643081}
N.~Papandreou and T.~Antonakopoulos, ``Bit and power allocation in constrained
  multicarrier systems: The single-user case,'' \emph{EURASIP J. Adv. Signal
  Process}, vol. 2008, Jan. 2008.

\bibitem{8995606}
C.~Xing, Y.~Jing, S.~Wang, S.~Ma, and H.~V. Poor, ``New viewpoint and
  algorithms for water-filling solutions in wireless communications,''
  \emph{IEEE Transactions on Signal Processing}, vol.~68, pp. 1618--1634, 2020.

\bibitem{MinimizationMethods}
N.~Z. Shor, \emph{Minimization Methods for Non-Differentiable Functions}.\hskip
  1em plus 0.5em minus 0.4em\relax Springer-Verlag, 1985, translated by K. C.
  Kiwiel and A. Ruszczynski.

\bibitem{boydsubgradient}
L.~X. S.~Boyd and A.~Mutapcic, \emph{Subgradient Methods (lecture
  notes)}.\hskip 1em plus 0.5em minus 0.4em\relax Stanford University, Oct.
  2003.

\bibitem{8125713}
R.~Ruby, S.~Zhong, H.~Yang, and K.~Wu, ``Enhanced uplink resource allocation in
  non-orthogonal multiple access systems,'' \emph{IEEE Transactions on Wireless
  Communications}, vol.~17, no.~3, pp. 1432--1444, 2018.

\bibitem{6816086}
E.~Che, H.~D. Tuan, and H.~H. Nguyen, ``Joint optimization of cooperative
  beamforming and relay assignment in multi-user wireless relay networks,''
  \emph{IEEE Transactions on Wireless Communications}, vol.~13, no.~10, pp.
  5481--5495, 2014.

\bibitem{7954630}
J.~{Zhao}, Y.~{Liu}, K.~K. {Chai}, A.~{Nallanathan}, Y.~{Chen}, and Z.~{Han},
  ``Spectrum allocation and power control for non-orthogonal multiple access in
  {HetNets},'' \emph{IEEE Transactions on Wireless Communications}, vol.~16,
  no.~9, pp. 5825--5837, 2017.

\bibitem{7934461}
Z.~Wei, D.~W.~K. Ng, J.~Yuan, and H.-M. Wang, ``Optimal resource allocation for
  power-efficient {MC-NOMA} with imperfect channel state information,''
  \emph{IEEE Transactions on Communications}, vol.~65, no.~9, pp. 3944--3961,
  2017.

\bibitem{7974731}
M.~Zeng, A.~Yadav, O.~A. Dobre, G.~I. Tsiropoulos, and H.~V. Poor, ``Capacity
  comparison between {MIMO-NOMA} and {MIMO-OMA} with multiple users in a
  cluster,'' \emph{IEEE Journal on Selected Areas in Communications}, vol.~35,
  no.~10, pp. 2413--2424, 2017.

\bibitem{8365765}
D.~Zhai and R.~Zhang, ``Joint admission control and resource allocation for
  multi-carrier uplink {NOMA} networks,'' \emph{IEEE Wireless Communications
  Letters}, vol.~7, no.~6, pp. 922--925, 2018.

\bibitem{8301007}
M.~Zeng, A.~Yadav, O.~A. Dobre, and H.~V. Poor, ``Energy-efficient power
  allocation for {MIMO-NOMA} with multiple users in a cluster,'' \emph{IEEE
  Access}, vol.~6, pp. 5170--5181, 2018.

\bibitem{8995591}
R.~Wang, W.~Kang, G.~Liu, R.~Ma, and B.~Li, ``Admission control and power
  allocation for {NOMA}-based satellite multi-beam network,'' \emph{IEEE
  Access}, vol.~8, pp. 33\,631--33\,643, 2020.

\bibitem{7982783}
W.~Bao, H.~Chen, Y.~Li, and B.~Vucetic, ``Joint rate control and power
  allocation for non-orthogonal multiple access systems,'' \emph{IEEE Journal
  on Selected Areas in Communications}, vol.~35, no.~12, pp. 2798--2811, 2017.

\bibitem{7845657}
M.-R. Hojeij, C.~Abdel~Nour, J.~Farah, and C.~Douillard, ``Waterfilling-based
  proportional fairness scheduler for downlink non-orthogonal multiple
  access,'' \emph{IEEE Wireless Communications Letters}, vol.~6, no.~2, pp.
  230--233, 2017.

\bibitem{8318569}
D.~Zhai, R.~Zhang, L.~Cai, B.~Li, and Y.~Jiang, ``Energy-efficient user
  scheduling and power allocation for {NOMA}-based wireless networks with
  massive {IoT} devices,'' \emph{IEEE Internet of Things Journal}, vol.~5,
  no.~3, pp. 1857--1868, 2018.

\bibitem{Dinkelbach}
W.~Dinkelbach, ``On nonlinear fractional programming,'' \emph{Management
  Science}, vol.~13, no.~7, p. 492–498, Mar. 1967.

\bibitem{nonlprgrmbook}
D.~P. Bertsekas, \emph{Nonlinear Programming}.\hskip 1em plus 0.5em minus
  0.4em\relax 2nd ed. Athena Scientific, 1999.

\end{thebibliography}

%
%

\end{document}